\documentclass[runningheads]{llncs}
\usepackage{rotating}
\usepackage{etex}
\usepackage[table]{xcolor} 
\usepackage[utf8]{inputenc}
\usepackage{float}
\usepackage{xspace}
\usepackage{calc}
\usepackage{caption}
\usepackage{subcaption}
\usepackage{enumitem}
\usepackage{booktabs}
\usepackage{tabu}
\usepackage{multirow}
\usepackage{url}
\usepackage{xurl}
\usepackage{graphicx} 
\usepackage{longtable} 
\usepackage{dcolumn}
\usepackage{pgfplots}
\usepackage{pgfkeys}
\usepackage{tikz}
\usepackage{times}
\usepackage{verbatim}  %for \begin{comment}\end{comment}
\usepackage{booktabs}
\usepackage{array}
\newcounter{rowcount}
\usetikzlibrary{patterns}

\usetikzlibrary{shapes,shapes.symbols,arrows,decorations,decorations.text,decorations.pathmorphing,decorations.pathreplacing,decorations.shapes,backgrounds,fit,positioning,calc,arrows,automata}
\usepackage{xspace}
\usepackage{xstring}
\usepackage[]{amsmath}
\usepackage{amssymb} %/checkmark
\usepackage{multirow} %for results table
\usepackage{mathtools,booktabs}
\usepackage{algorithm}
\usepackage{algorithmicx}
\usepackage{algcompatible}
\usepackage[noend]{algpseudocode}
\algnewcommand\algorithmicforeach{\textbf{for each}}
\algdef{S}[FOR]{ForEach}[1]{\algorithmicforeach\ #1\ \algorithmicdo}
\usepackage{float}
\usepackage{listings} 
\usepackage{hyperref} 
\usepackage[toc,page]{appendix}
\usepackage{mathtools}
\usepackage{soul}

\newcommand{\pseudo}[1]{\texttt{#1}}

\newcommand{\true}{\textsf{true}}

\newcommand{\T}{{\mathcal{T}}}

\renewcommand{\true}{\mathit{true}\,}
\newcommand{\stam}[1]{}

\newcommand{\agl}{\textsc{AG}_{\lstar}}

\pgfplotsset{compat=1.14} 

\pgfplotsset{width=7cm,height=5.2cm}
\definecolor{apricot}{rgb}{0.98, 0.81, 0.69}
\definecolor{bondiblue}{rgb}{0.0, 0.58, 0.71}
\definecolor{ghostwhite}{rgb}{0.97, 0.97, 1.0}
\definecolor{lightcarminepink}{rgb}{0.9, 0.4, 0.38}

\newcommand{\algo}[1]{{\ensuremath{{\textbf{#1}}}}}
\newcommand{\lstar}{{\ensuremath{\algo{L}^*}}}
\newcommand{\aut}[1]{\ensuremath{\mathcal{#1}}}
\newcommand{\dom}[1]{\ensuremath{\mathbb{#1}}}

\algrenewcommand\algorithmicrequire{\textbf{Input:}}
\algrenewcommand\algorithmicensure{\textbf{Output:}}

\newcommand{\commentout}[1]{}
\newcommand{\query}[1]{\textsc{#1}}

\newcommand{\mq}{\query{mq}}
\newcommand{\eq}{\query{eq}}
\newcommand{\mqor}{\query{mq-oracle}}
\newcommand{\eqor}{\query{eq-oracle}}

\newcommand{\repair}{\query{repair}}

\newcommand{\ansrep}[1]{\pseudo{repair}}
\newcommand{\ansyes}[1]{\pseudo{yes}}
\newcommand{\ansno}[1]{\pseudo{no}}

\begin{document}


%%%

\title{Assume, Guarantee or Repair -- A Regular Framework for Non Regular Properties\thanks{This research was partially supported by the ISRAEL SCIENCE FOUNDATION (ISF) grant No. 346/17.
} 
}

\author{Hadar Frenkel$^1$, Orna Grumberg$^2$, Corina S. P\u{a}s\u{a}reanu$^{3}$ and Sarai Sheinvald$^4$}

\authorrunning{H. Frenkel et al.}

\institute{
 CISPA Helmholtz Center for Information Security, Saarbr{\"u}cken, Germany \and
Department of Computer Science, The Technion, Israel
\and Carnegie Mellon University and NASA Ames Research Center, CA, USA
\and Department of Software Engineering, ORT Braude College, Israel}
	
\date{}
\maketitle

%%%

\begin{abstract}
%Insert your abstract here. Include keywords, PACS and mathematical subject classification numbers as needed.
We present Assume-Guarantee-Repair (AGR) -- a novel framework which  verifies that a program satisfies a set of properties and also {\em repairs} the program in case the verification fails. 
We consider {\em communicating programs} -- these are simple C-like programs, extended with synchronous actions over communication  channels. 
Our method, which consists of a learning-based approach to assume-guarantee reasoning, performs verification and repair simultaneously: in every iteration, AGR  either makes another step towards proving that the (current) system satisfies the required properties, or alters the system in a way that brings it closer to satisfying the properties.  To handle infinite-state systems we build finite abstractions, 
for which we check the satisfaction of complex properties that contain first-order constraints, using both syntactic and semantic-aware methods. 
We implemented AGR and evaluated it on various communication protocols. Our experiments present compact proofs of correctness and quick repairs.

\end{abstract}

\section{Introduction}\label{Intro}
Verification of large-scale systems is a main challenge in the field of formal verification. Often, the verification process of such systems does not scale well.  
{\em Compositional verification}
%\cite{} 
aims to address this challenge by breaking up the verification of a large system into the verification of its smaller components which can be checked separately. The results of the verification can be composed back to conclude 
%of a system separately, and from the correctness of the individual components, to conclude 
the correctness of the entire system. 
This, however, is not always possible, since the correctness of a component often depends on the behavior of its environment. 

The Assume-Guarantee (AG) style compositional verification~\cite{DBLP:journals/tse/MisraC81,Pnueli85} suggests a solution to this problem.
The simplest AG rule checks if a system composed of $M_1$ and $M_2$ satisfies a property $P$ by checking that $M_1$ under assumption $A$ satisfies $P$ and that any system containing $M_2$ as a component satisfies $A$. 
%Several frameworks have been proposed to support this style of reasoning. Finding a suitable assumption $A$  is then a common challenge in such frameworks.
  
In this work, we present \emph{Assume-Guarantee-Repair} (AGR) -- a fully automated framework which applies the AG rule, and while seeking a suitable assumption $A$, incrementally repairs the given program in case the verification fails.
Our framework is inspired by~\cite{DBLP:journals/fmsd/PasareanuGBCB08}, which presented a learning-based method to finding a suitable assumption $A$, using the \lstar~\cite{DBLP:journals/iandc/Angluin87} algorithm for learning regular languages. However, unlike previous work, AGR not only performs verification but also repair.

Our AGR framework handles \emph{communicating programs}, which are commonly used for modeling concurrent systems. These are infinite-state C-like programs, extended with the ability to synchronously read and write messages over communication channels. We model such programs as finite-state automata over an \emph{action alphabet}, which reflects the program statements. 
%They also enable us to verify regular properties, which are much more expressive than safety. Every regular language can specify acceptable finite program behaviours.  
The automata representation is similar in nature to that of control-flow graphs. %~\cite{CFG}. 
Its advantage, however, is in the ability to exploit an automata-learning algorithm 
such as $\lstar$. The accepting states in the automaton representation model points of interest in the program, to which the specification can relate.

We have implemented a tool for AGR and evaluated it on examples modeling communication protocols of various sizes and with various types of errors. Our experiments show that for most examples, AGR converges and finds a repair after a few (2-5) iterations of verify-repair. Moreover, our tool generates assumptions that are significantly smaller then the (possibly repaired) $M_2$, thus constructing a compact and efficient proof of correctness.  
%To demonstrate the effectiveness of AGR, we have applied it on several examples using SMT solvers, learning, and model checking tools.

\subsubsection*{Contributions}
To summarize, the main contributions of this paper are:
\begin{enumerate}
\item A learning-based Assume-Guarantee algorithm for infinite-state communicating programs, which manages overcoming the difficulties such programs present. In particular, our algorithm overcomes the inherent irregularity of the first-order constraints in these programs, and offers syntactic solutions to the semantic problems they impose. 
\item An Assume-Guarantee-Repair algorithm, in which the Assume-Guarantee and the Repair procedures intertwine to produce a repaired program which, due to our construction, maintains many of the ``good'' behaviors of the original program.  Moreover, in case the original program satisfies the property, our algorithm is guaranteed to terminate and return this conclusion. 
\item An incremental learning algorithm that uses query results from previous iterations in learning a new language with a richer alphabet.
\item A novel use of abduction to repair communicating programs over first order constraints.
\item An implementation of our algorithm, demonstrating the effectiveness of our framework. 
\end{enumerate}
	
\paragraph{Contribution over conference version.}
Preliminary results of this work were published in~\cite{DBLP:conf/tacas/FrenkelGPS20}. This paper extends the results of~\cite{DBLP:conf/tacas/FrenkelGPS20} by the following new contributions. 
\begin{itemize}
    \item A formal definition of the weakest assumption for communicating programs, and a study of special cases for which the weakest assumption is regular (Section~\ref{AGR:AR_rule_reg}). 
    \item A full description of the implementation of the teacher in the context of communicating programs (Section~\ref{sec:queries} and Alg.~\ref{alg:queries}). 
    \item Convergence of the syntactic repair -- characterizing cases in which this repair does not terminate (Section~\ref{sec:syntactic_convg}). 
   \item A formal discussion regarding the soundness of the AG-rule for communicating programs, completeness and termination (Section~\ref{Sec:termination}).
   \item Full proofs and multiple additional examples throughout the paper. 
   \item A detailed discussion regarding future work and possible extensions (Section~\ref{sec:future}).
\end{itemize}

\paragraph{Paper organization.} The rest of the paper is organized as follows. In the next section, we give a high-level overview of AGR, highlighting the salient features of our approach. Section~\ref{sec:related} describes related work. The following three sections set up the background necessary for understanding our approach.  Section~\ref{sec:prelim} gives preliminary definitions, section~\ref{sec:communicating} describes communicating programs and their properties (expressed as regular languages), and section~\ref{sec:traces} proves properties of traces that we use later when proving soundness and completeness of AGR. Section~\ref{sec:rule} describes and analyzes the assume-guarantee rule used by AGR, while section~\ref{sec:AGR} describes in detail the AGR approach. Finally, section~\ref{sec:experiments} describes our experimental evaluation and section~\ref{sec:conclusion} concludes the paper.
\section{Overview}
%\subsubsection*{Example}
%\vspace{-0.6cm}
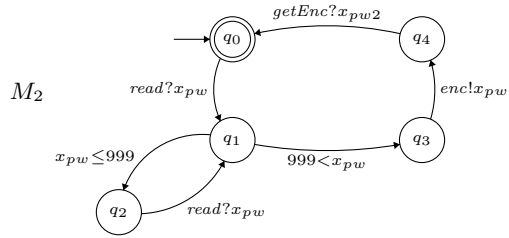
\begin{figure}
		\centering
		\begin{tabular}{c}%{m{7cm}m{7cm}}
			%	\begin{minipage}{.2\textwidth}
			{{\begin{tabular}{ll}
					{$1$:} & \pseudo{while(true)  }\\
                    {$2$:} & $\qquad$ \pseudo{password:=readInput; }\\
					{$3$}  & $\qquad$ \pseudo{while(password$\leq 999$)  }\\
					{$4$:} & $\qquad\qquad$ \pseudo{password:=readInput;}\\
					{$5$:} & $\qquad$ \pseudo{password2:=encrypt(password);}\\
					%{$6$:} & $\qquad$ \pseudo{return password2;}\\
                    %{$6$}  & $\qquad$ \pseudo{else}\\
				\end{tabular}}}
			
			%\hspace{0.6cm}
            \\
            \\
            \\
			%		\end{minipage}
			%		\begin{minipage}{.2\textwidth}
	
\begin{tikzpicture}[scale=0.1]
\tikzstyle{every node}+=[inner sep=0pt]
\draw [black] (23.2,-18.8) circle (3);
\draw (23.2,-18.8) node {${\scriptstyle q_0}$};
\draw [black] (23.2,-18.8) circle (2.4);
\draw [black] (23.2,-32.2) circle (3);
\draw (23.2,-32.2) node {${\scriptstyle q_1}$};
\draw [black] (8.1,-41.8) circle (3);
\draw (8.1,-41.8) node {${\scriptstyle q_2}$};
\draw (-4.1,-25.8) node {$M_2$};
\draw [black] (48.4,-32.2) circle (3);
\draw (48.4,-32.2) node {${\scriptstyle q_3}$};
\draw [black] (48.4,-18.8) circle (3);
\draw (48.4,-18.8) node {${\scriptstyle q_4}$};
\draw [black] (15.4,-18.8) -- (20.2,-18.8);
\fill [black] (20.2,-18.8) -- (19.4,-18.3) -- (19.4,-19.3);
\draw [black] (21.594,-29.679) arc (-155.90802:-204.09198:10.236);
\fill [black] (21.59,-29.68) -- (21.72,-28.74) -- (20.81,-29.15);
\draw (20.2,-25.5) node [left] {${\scriptstyle \mathit{read}?x_{pw}}$};
\draw [black] (8.708,-38.871) arc (160.55127:84.342:11.121);
\fill [black] (8.71,-38.87) -- (9.45,-38.28) -- (8.5,-37.95);
\draw (4.9,-35.69) node [above] {${\scriptstyle x_{pw}\leq 999}$};
\draw [black] (22.118,-34.991) arc (-27.72715:-87.37958:13.141);
\fill [black] (22.12,-34.99) -- (21.3,-35.47) -- (22.19,-35.93);
\draw (22.42,-40.47) node [below] {${\scriptstyle \mathit{read}?x_{pw}}$};
\draw [black] (45.458,-32.784) arc (-80.27359:-99.72641:57.166);
\fill [black] (45.46,-32.78) -- (44.58,-32.43) -- (44.75,-33.41);
\draw (35.8,-34.11) node [below] {${\scriptstyle 999<x_{pw}}$};
\draw [black] (49.539,-21.569) arc (16.24661:-16.24661:14.05);
\fill [black] (49.54,-21.57) -- (49.28,-22.48) -- (50.24,-22.2);
\draw (50.6,-25.5) node [right] {${\scriptstyle \mathit{enc}!x_{pw}}$};
\draw [black] (26.105,-18.055) arc (102.46919:77.53081:44.9);
\fill [black] (26.11,-18.05) -- (26.99,-18.37) -- (26.78,-17.39);
\draw (35.8,-16.5) node [above] {${\scriptstyle \mathit{getEnc}?x_{pw2}}$};
\end{tikzpicture}
		\end{tabular}
		\caption{Modeling a communicating program ($M_2$) as an automaton}\label{model}%\vspace{-0.6cm}
\end{figure}
We give a high level overview of AGR via an example.
%\subsubsection{Overview}
Figure~\ref{model} presents the code of a communicating program (upper part) and its corresponding automaton $M_2$ (lower part). The automaton alphabet consists of constraints (e.g. $x_{pw} \leq 999$), assignment actions (e.g. $y_{pw} := 2\cdot y_{pw}$ in $M_1$ of Figure~\ref{fig:two_systems_spec}), and communication actions (e.g. $enc!x_{pw}$ sends the value of variable $x_{pw}$ over channel $enc$, and $getEnc?x_{pw2}$ reads a value to $x_{pw2}$ on channel~$getEnc$).

The specification $P$\footnote{Throughout the paper we use {\em property} and {\em specification} interchangeably.} is modeled as an automaton that does not contain assignment actions. It may contain communication actions in order to specify behavioral requirements, as well as constraints over the variables of both system components, that express requirements on their values in various points in the runs.

Consider, for example, the program $M_1$ and the specification $P$ seen in Figure~\ref{fig:two_systems_spec}, and the program $M_2$ of Figure~\ref{model}. 
$M_2$ reads a password on channel $read$ to the variable $x_{pw}$, and once it is long enough (has at least four digits), it sends the value of $x_{pw}$ to $M_1$ through channel $enc$. $M_1$ reads this value to variable $y_{pw}$ and then applies a simple function that changes its value, and sends the changed variable back to $M_2$. The property $P$ reasons about the parallel run of the two programs. The pair $(\mathit{getEnc}!y_{pw},\mathit{getEnc}?x_{pw2})$ denotes a synchronization of $M_1$ and $M_2$ on channel $\mathit{getEnc}$. $P$
makes sure that the parallel run of $M_1$ and $M_2$ always reads a value and then encrypts it -- a temporal requirement. In addition, it makes sure that the value after encryption is different from the original value, and that there is no overflow -- both are semantic requirements over the program variables. That is, $P$ expresses temporal requirements that contain first order constraints. 
In case one of the requirements does not hold, $P$ reaches the state $r_4$ which is an error state. Note that $P$ here is not complete, for simplicity of presentation (see Definition~\ref{def:detAndComplete} for a formal definition of a complete program).

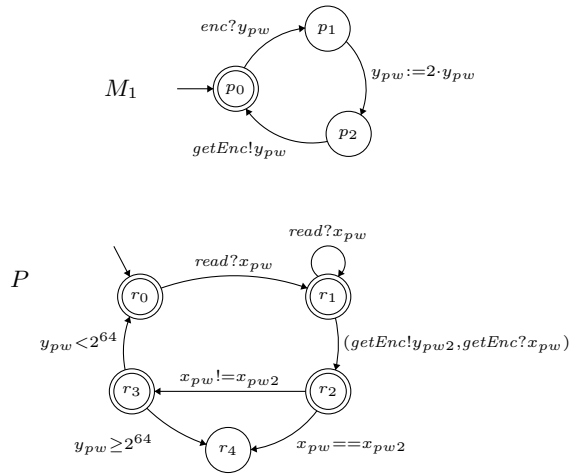
\begin{figure}[]
	\centering
	\begin{tabular}{c}
    
\begin{tikzpicture}[scale=0.1]

\tikzstyle{every node}+=[inner sep=0pt]
\draw [black] (23.2,-18.8) circle (3);
\draw (23.2,-18.8) node {${\scriptstyle p_0}$};
\draw (8,-18.8) node {$M_1$};

\draw [black] (23.2,-18.8) circle (2.4);
\draw [black] (35.4,-10.8) circle (3);
\draw (35.4,-10.8) node {${\scriptstyle p_1}$};
\draw [black] (38.2,-24.8) circle (3);
\draw (38.2,-24.8) node {${\scriptstyle p_2}$};
\draw [black] (15.4,-18.8) -- (20.2,-18.8);
\fill [black] (20.2,-18.8) -- (19.4,-18.3) -- (19.4,-19.3);
\draw [black] (24.314,-16.025) arc (150.10242:96.40631:10.719);
\fill [black] (32.41,-10.72) -- (31.56,-10.31) -- (31.67,-11.3);
\draw (23.18,-11.9) node [above] {${\scriptstyle \mathit{enc}?y_{pw}}$};
\draw [black] (37.857,-12.497) arc (45.659:-23.03914:8.849);
\fill [black] (39.82,-22.29) -- (40.59,-21.75) -- (39.67,-21.36);
\draw (41.09,-16.78) node [right] {${\scriptstyle y_{pw}:=2\cdot y_{pw}}$};
\draw [black] (35.383,-25.802) arc (-78.51493:-145.08789:10.63);
\fill [black] (24.55,-21.47) -- (24.6,-22.41) -- (25.42,-21.84);
\draw (23.75,-25.86) node [below] {${\scriptstyle \mathit{getEnc}!y_{pw}}$};

\end{tikzpicture}
%	\hspace{0.6cm} 
		
	 \\
	 \\
	 \\
		%		\end{minipage}
		%		\begin{minipage}{.2\textwidth}

\begin{tikzpicture}[scale=0.1]
\tikzstyle{every node}+=[inner sep=0pt]
\draw [black] (31,-17.8) circle (3);
\draw (31,-17.8) node {${\scriptstyle r_1}$};
\draw [black] (31,-17.8) circle (2.4);
\draw [black] (31,-30.4) circle (3);
\draw (31,-30.4) node {${\scriptstyle r_2}$};
\draw [black] (31,-30.4) circle (2.4);
\draw [black] (6,-17.8) circle (3);
\draw (6,-17.8) node {${\scriptstyle r_0}$};
\draw [black] (6,-17.8) circle (2.4);
\draw [black] (17.7,-38.2) circle (3);
\draw (17.7,-38.2) node {${\scriptstyle r_4}$};
\draw [black] (4.9,-30.4) circle (3);
\draw (4.9,-30.4) node {${\scriptstyle r_3}$};
\draw (-10,-15.4) node {$P$};

\draw [black] (4.9,-30.4) circle (2.4);
\draw [black] (29.677,-15.12) arc (234:-54:2.25);
\draw (31,-10.55) node [above] {${\scriptstyle \mathit{read}?x_{pw}}$};
\fill [black] (32.32,-15.12) -- (33.2,-14.77) -- (32.39,-14.18);
\draw [black] (31.957,-20.638) arc (13.03629:-13.03629:15.347);
\fill [black] (31.96,-27.56) -- (32.62,-26.9) -- (31.65,-26.67);
\draw (32.85,-24.1) node [right] {${\scriptstyle (\mathit{getEnc}!y_{pw2},\mathit{getEnc}?x_{pw})}$};
\draw [black] (8.776,-16.667) arc (109.28921:70.71079:29.436);
\fill [black] (28.22,-16.67) -- (27.63,-15.93) -- (27.3,-16.87);
\draw (18.5,-14.51) node [above] {${\scriptstyle \mathit{read}?x_{pw}}$};
\draw [black] (2.5,-11.1) -- (4.61,-15.14);
\fill [black] (4.61,-15.14) -- (4.68,-14.2) -- (3.8,-14.66);
\draw [black] (28,-30.4) -- (7.9,-30.4);
\fill [black] (7.9,-30.4) -- (8.7,-30.9) -- (8.7,-29.9);
\draw (17.95,-29.9) node [above] {${\scriptstyle x_{pw}!=x_{pw2}}$};
\draw [black] (4.003,-27.544) arc (-169.14845:-200.83028:13.031);
\fill [black] (4.62,-20.46) -- (3.87,-21.03) -- (4.8,-21.38);
\draw (3.19,-23.89) node [left] {${\scriptstyle y_{pw}<2^{64}}$};
\draw [black] (29.493,-32.986) arc (-36.77303:-82.44661:13.142);
\fill [black] (20.69,-38.15) -- (21.55,-38.54) -- (21.42,-37.55);
\draw (34.09,-36.96) node [below] {${\scriptstyle x_{pw}==x_{pw2}}$};
\draw [black] (14.795,-37.461) arc (-108.43356:-134.28061:20.699);
\fill [black] (14.8,-37.46) -- (14.19,-36.73) -- (13.88,-37.68);
\draw (2.46,-36.01) node [below] {${\scriptstyle y_{pw}\geq 2^{64}}$};

\end{tikzpicture}

	\end{tabular}
	\caption{The program $M_1$ and the specification $P$ }\label{fig:two_systems_spec}
%\vspace{-0.8cm}
\end{figure}

The \lstar~algorithm aims at learning a regular language $U$. Its entities consist of a {\em teacher} -- an oracle who answers {\em membership queries} (``is the word $w$ in $U$?'') and {\em equivalence queries} (``is $\cal A$ an automaton whose language is $U$?''), and a {\em learner}, who iteratively constructs a finite deterministic automaton $\cal A$ for $U$ by submitting a sequence of membership and equivalence queries to the teacher. 

In using the \lstar~algorithm for learning an assumption $A$ for the AG-rule, membership queries are answered according to the satisfaction of the specification $P$: If $M_1||t$ satisfies $P$, then the trace $t$ in hand should be in $A$.  Otherwise, $t$ should not be in $A$. Once the learner constructs a stable system $A$, it submits an equivalence query. The teacher then checks whether $A$ is a suitable assumption, that is, whether $M_1||A$ satisfies $P$, and whether the language of $M_2$ is contained in the language of $A$. According to the results, the process either continues or halts with an answer to the verification problem. The learning procedure aims at learning the weakest assumption $A_{w}$, which contains all the traces that in parallel with $M_1$  satisfy $P$ . 
The key observation that guarantees termination in \cite{DBLP:journals/fmsd/PasareanuGBCB08} is that the components in this procedure -- $M_1$,$M_2$, $P$ and $A_{w}$ --  are all regular. 
\subsubsection*{Our setting}

Our setting is more complicated, since the traces in the components -- both the programs and the specification -- contain constraints, which are to be checked semantically and not syntactically. 
These constraints may cause some traces to become infeasible. For example, if a trace contains an assignment $x:=3$ followed by a constraint $x \geq 4$ (modeling an ``if'' statement), then this trace does not contribute any concrete runs, and therefore does not  affect the system behavior. Thus, we must add feasibility checks to the process.% (by applying an SMT solver, for example). 

Constraints in the specification also pose a difficulty, as satisfiability of a specification~is determined by the semantics of the constraints and not only by the language syntax, and so there is more here to check than standard language containment. 
Moreover, in our setting $A_{w}$ above may no longer be regular, see Lemma~\ref{lemma:weak_non_reg}. 
However, our method manages overcoming this problem in a way that still guarantees termination in case the verification succeeds, and progress, otherwise. In addition, we characterize special cases in which the weakest assumption is in-fact regular. 

As we have described above, not only do we construct a learning-based method for the AG-rule for communicating programs, but we also repair the programs in case the verification fails. 
An AG-rule 
%for finite-state systems (e.g. LTSs~\cite{DBLP:conf/tacas/CobleighGP03}) 
can either conclude that $M_1 || M_2$ satisfies $P$, or 
return a real, non-spurious counterexample of a computation $t$ of $M_1 || M_2$ that violates $P$. 
In our case, instead of returning $t$, we repair $M_2$ in a way that eliminates the counterexample $t$. Our repair is both syntactic and semantic, where for semantic repair we use {\em abduction}~\cite{peirce1974collected} to infer a new constraint which makes the counterexample $t$ infeasible. 
%and add it to the set of actions of $M_2$ through the learning process in a way which makes the counterexample infeasible. 

Consider again $M_1$ and $P$ of Figure~\ref{fig:two_systems_spec} and $M_2$ of Figure~\ref{model}.
The composition $M_1||M_2$ does not satisfy $P$. For example, if the initial value of $x_{pw}$ is $2^{63}$, then after encryption the value of $y_{pw}$ is $2^{64}$, violating $P$.  
Our algorithm finds a bad trace during the AG stage which captures this bad behavior, and the abduction in the repair stage finds a constraint that eliminates it: $x_{pw} < 2^{63}$, and inserts this constraint to $M_2$. 

Following this step we now have an updated $M_2$, and we continue with applying the AG-rule again, using information we have gathered in the previous steps. In addition to removing the error trace, we update the alphabet of $M_2$ with the new constraint.
Continuing our example, in the following iteration AGR will verify that the repaired $M_2$ together with $M_1$ satisfy $P$, and terminate.

Thus, AGR operates in a verify-repair loop, where each iteration runs a learning-based process to determine whether the (current) system satisfies $P$, and if not, eliminates bad behaviors from $M_2$ while enriching the set of constraints derived from these bad behaviors, which often leads to a quicker convergence. In case the current system does satisfy $P$, we return the repaired $M_2$ together with an assumption $A$ that abstracts $M_2$ and acts as a smaller proof for the correctness of the system. 
The original motivation for using the AG rule for verification is to find small proofs for the correctness of $M_1 || M_2 \vDash P$, without the need to compute the whole composition of $M_1 || M_2$, but using the smaller assumption $A$. In our case, we use the same reasoning, but we do not only prove correctness or provide a counterexample, we also repair $M_2$. Thus, we keep learning abstractions to the repaired $M_2$, and, as we later show, rely on information from previous iterations in order to learn an abstraction for the next iteration. Our algorithm produces both a repaired system $M_2'$ such that $M_1 || M_2' \vDash P$; and an assumption $A$ that serves as an abstraction for $M_2'$ and allows to prove the correctness of the repaired system with respect to $P$, without the need to compute the whole composition of $M_1 || M_2'$. 

The assumption $A$ can later be used for the verification of different components other than $M_2$ (or the repaired $M_2'$). If some $M_2''$ is such that $M_2''\subseteq A$ (that is, $A$ is an abstraction of $M_2''$), then from $M_1 || A \vDash P$ we conclude that $M_1 || M_2'' \vDash P$. That is, given $M_2''$ and the small $A$, all we need to check is that $A$ is indeed an abstraction of $M_2''$, instead of computing the composition of the whole system and specification.

\section{Related Work}

Assume-guarantee style compositional verification ~\cite{DBLP:journals/tse/MisraC81,Pnueli85} has been extensively studied.
The assumptions necessary for compositional verification were first produced manually, limiting the practicality of the method.
More recent works~\cite{DBLP:conf/tacas/CobleighGP03,DBLP:journals/ase/GiannakopoulouPB05,DBLP:conf/tacas/GheorghiuGP07,DBLP:conf/tacas/ChakiS07} proposed techniques for automatic assumption generation using learning and abstraction refinement techniques, making assume-guarantee verification more appealing. 
In~\cite{DBLP:journals/fmsd/PasareanuGBCB08,DBLP:conf/tacas/ChakiS07} alphabet refinement has been suggested as an optimization, to reduce the alphabet of the generated assumptions, and consequently their sizes. This optimization can easily  be incorporated into our framework as well.
Other learning-based approaches for automating assumption generation have been described in~\cite{DBLP:conf/cav/ChenCFTTW10,DBLP:journals/fmsd/GuptaMF08,DBLP:conf/tacas/ChenFCTW09}. All these works address non-circular rules and are limited to finite state systems.
%in the context of the verification of finite-state models.
Automatic assumption generation for circular rules is presented in~\cite{Karam1,Karam2}, using compositional rules similar to the ones studied in~\cite{DBLP:conf/charme/McMillan99circular,DBLP:conf/cav/NamjoshiT00}.
Our approach is based on a non-circular rule but it targets complex, infinite-state concurrent systems, and addresses not only verification but also repair.
The compositional framework presented in~\cite{DBLP:conf/fm/LinH14} addresses \lstar-based compositional verification and synthesis but it only targets finite state systems.

Several works use abduction. The work in~\cite{DBLP:conf/tacas/LiDDMS13} addresses automatic synthesis of circular compositional proofs based on logical abduction; however the focus of that work is sequential programs, while here we target concurrent programs.
%A key difference is that they refer to a decomposition of a sequential program, while we consider a parallel composition. 
A sequential setting is also considered in~\cite{POPL16}, where abduction is used for automatically generating a program environment. Our computation of abduction is similar to that of~\cite{POPL16}. However, we require our constraints to be over a predefined set of variables, while they look for a minimal~set. Moreover, our goal is to repair a faulty program.

The approach presented in \cite{CAVMayMust} aims to compute the {\em interface} of an infinite-state component. Similar to our work, the approach works with both over- and under- approximations but it only analyzes one component at a time. Furthermore, the component is restricted to be deterministic (necessary for the permissiveness check). In contrast we use both components of a system to compute the necessary assumptions, and as a result they can be much smaller than in \cite{CAVMayMust}. Furthermore, we do not restrict the components to be deterministic and, more importantly, we also address the system repair in case of dissatisfaction.

Many works study the learnability of symbolic automata, e.g.~\cite{DBLP:conf/cav/ArgyrosD18,DBLP:conf/csl/FismanFZ22,DBLP:conf/tacas/MalerM14,DBLP:conf/fm/Sheinvald19}. The work of~\cite{DBLP:conf/vmcai/HowarSM11} studies alphabet refinement for learning symbolic automata. However, all of these are restricted to the abstraction of single transitions, where valuations of variables (states) and their update via program statements are not considered. We, on the other hand, model program statements and program states.
%and they are not able to model general programs as we do here, but only regular languages over large alphabets. 

There is extensive research on automated program repair. Examples for testing-based approaches are~\cite{DBLP:journals/tse/GouesNFW12,DBLP:conf/icsm/QiML13,DBLP:conf/icse/MechtaevYR16,DBLP:conf/icse/MechtaevYR15,DBLP:conf/icse/NguyenQRC13}, whereas in this work we repair with respect to a formal specification rather than a test-suit. 
In the context of formal repair with respect to a specification, 
\cite{DBLP:journals/ai/Reiter87} lays foundations of diagnosis of errors in multi-component systems, with respect to logical specifications. Their setting relies on the logical representation of the system and the specification, and is concerned with finding all reasons to the violation of the property. In our work we are indeed looking for error traces that violate the specification, but our specifications allow us to reason also about the temporal behavior of the program due to its representation as an automaton. %In addition, the assume-guarantee approach results in an abstraction $A$ of $M_2$, with a more concise representation. 
\cite{DBLP:conf/cav/JobstmannGB05} presents an algorithm to repair finite-state programs with respect to a temporal specification, where we are concerned with infinite-state programs that admit some finite representation. \cite{DBLP:conf/fm/RothenbergG16,DBLP:conf/cav/RothenbergG20} present automated repair algorithms using SAT based techniques for specifications given as assertions. 
Note that while we are aiming at repairing the system, our work is not restricted only to repair, but also for finding a small proof of correctness, in the form of the assumption $A$, that abstracts (the repaired) $M_2$ and allows a smaller and more efficient composition of the multiple components. 
\label{sec:related}
\section{Preliminaries}\label{sec:prelim}
%\subsection{(Deterministic) Finite Automata}
\subsection{Regular Languages}
In this section we define the notions of finite automata and regular languages. We then give a high-level description of the automata learning algorithm \lstar.
\begin{definition}
A \emph{Finite Automaton} is $M = \langle Q, \alpha , \delta, q_0, F\rangle$, where all sets are finite and nonempty, and: 
	\begin{enumerate}
		\item $Q$ is the set of states and $q_0 \in Q$ is the initial state.
		\item $\alpha$  is the set of alphabet of $M$.
        \item $\delta \subseteq Q\times \alpha \times Q$ is the transition relation. We sometimes also refer to it as a function $\delta: Q\times \alpha \rightarrow 2^Q$. 
        \item $F\subseteq Q$ is the set of accepting states. 
	\end{enumerate}
\end{definition}
A Finite automaton is \emph{deterministic} (DFA) if for every $q \in Q$ and every $a \in \alpha$, $|\delta (q, a)| \leq 1$. It is also  \emph{complete} if for every $q \in Q$ and every $a \in \alpha$, $|\delta (q, a)| = 1$.

Let $w = a_1, \ldots, a_k$ be a word over alphabet $\alpha$. We say that $M$ accepts $w$ if there exists $r_0, \ldots, r_k \in Q$ such that $r_0 = q_0$, $r_k \in F$ and for every $i$, $0 \leq i \leq k-1$, $r_{i+1} \in \delta(r_i, a_{i+1})$. The language of $M$, $L(M)$, is the set of all words accepted by $M$. A set of words $W$ over $\alpha$ is called \emph{regular} if there exists a finite automaton $M$ such that $W= L(M)$.

\subsubsection{Learning Regular Languages}\label{sec:l*}

%\lstar~algorithm, constructing the table. 
In
\emph{active automata learning}  
we are interested in learning a DFA $\cal A$ for some regular language $U$, by issuing queries to some knowledgeable teacher.
%is an algorithm whose goal is to characterize a regular language $U$ by learning a DFA $\cal A$, such that $L(\cal A) = U$. 
\lstar~\cite{DBLP:journals/iandc/Angluin87} is an algorithm for active automata learning, using membership and equivalence queries. 
%It consists of a \emph{teacher}, which is an oracle that can answer two types of queries: \emph{membership queries} and \emph{equivalence queries}.
%
%The \lstar~algorithm aims at learning a regular language $U$. 
Its entities consist of a {\em teacher} -- an oracle who answers {\em membership queries} (\mq s) -- ``is the word $w$ in $U$?'', and {\em equivalence queries} (\eq s) -- ``is $\cal A$ an automaton whose language is $U$?''; and a {\em learner}, who iteratively constructs a finite deterministic automaton $\cal A$ for $U$ by submitting a sequence of membership and equivalence queries to the teacher. 

The learner presents its queries to the teacher following the instructions of the \lstar\ algorithm. In particular, membership queries present words in increasing length and in alphabetical order.
Based on its queries, the learner maintains an \emph{observation table} $T$, which accumulates all words learned so far, together with an indication of whether each word belongs to $U$ or not.
If $w$ is in $T$ and is indicated to be in $U$ we call $w$ a \emph{positive example}, denoted by +$w$, and if it is indicated to not be in $U$ we call it a \emph{negative example}, denoted by -$w$.
 
When the observation table is ``stable'', the learner generalizes the table and constructs a conjectured automaton $\cal A'$, which is sent via equivalence query to the teacher. The teacher either determines that $L(\cal A') = U$, in which case \lstar\  terminates; or returns a counterexample in the form of a word which is in the symmetric difference between $L(\cal A')$ and $U$.
The learner updates its observation table accordingly and continues with its construction.

If the teacher answers according to a regular language $U$, then
the \lstar~algorithm is guaranteed to terminate and learn a DFA for $U$, in polynomial time. In the rest of the paper we rely on the correctness of \lstar\ and the fact that the learner queries words in increasing order. 

For more details regarding the \lstar\ algorithm, we refer the reader to~\cite{DBLP:journals/iandc/Angluin87}.

\subsection{Assume-Guarantee Reasoning}
For large systems, composed of two components (or more), Assume-Guarantee (AG) proof rule~\cite{Pnueli85} is highly effective.
Given a system composed of components $M_1$ and $M_2$ and a property $P$, the AG-rule concludes that $M_1|| M_2 \models P$ provided that an assumption $A$ is found, such that $M_1 || A \models P$ and  the behaviors of $M_2$ are contained in the behaviors of $A$.
For $M_1$, $M_2$, $P$ and $A$ that are all Label Transition Systems (LTSs), an automated AG rule has been suggested in~\cite{DBLP:conf/tacas/CobleighGP03}. There, the assumption $A$ is learned using the \lstar\ algorithm. Often, the learned $A$ is much smaller (in terms of states and transitions) than $M_2$. Thus, an efficient composition proof rule is obtained, for proving that $M_1|| M_2 \models P$.

In using the \lstar~algorithm for learning an assumption $A$ for the AG-rule, membership queries are answered according to the satisfaction of the specification $P$: If $M_1||t$ satisfies $P$, then the trace $t$ in hand should be in $A$.  Otherwise, $t$ should not be in $A$. Once the learner constructs a stable system $A$, it submits an equivalence query. The teacher then checks whether $A$ is a suitable assumption, that is, whether $M_1||A$ satisfies $P$, and whether the language of $M_2$ is contained in the language of $A$. According to the results, the process either continues or halts with an answer to the verification problem. The learning procedure aims at learning the weakest assumption $A_{w}$, which contains all the traces that in parallel with $M_1$  satisfy $P$ . 
The key observation that guarantees termination in \cite{DBLP:journals/fmsd/PasareanuGBCB08} is that the components in this procedure -- $M_1$,$M_2$, $P$ and $A_{w}$ --  are all regular. 

It worth noticing that often $ A$, learned by \lstar, is much smaller than the weakest  assumption and than $M_2$.

\section{Communicating Programs and Regular Properties} \label{sec:Defin}
\label{sec:communicating}
	In this section we 
	present the notion of \emph{communicating programs}. These are C-like programs, extended with the ability to synchronously read and write messages over communication channels. We model such programs as automata over an \emph{action alphabet} that reflects the program statements. The alphabet includes \emph{constraints}, which are quantifier-free first-order formulas, representing the conditions in \emph{if} and \emph{while} statements. It also includes \emph{assignment statements} and $read$ and $write$ \emph{communication actions}.
The automata representation is similar in nature to that of control-flow graphs. %~\cite{CFG}.  
Its advantage, however, is in the ability to exploit an automata-learning algorithm such as \lstar\ for its verification~\cite{DBLP:journals/iandc/Angluin87}.

We first formally define the alphabet over which communicating programs are defined. 
Let $G$ be a finite set of communication channels. Let $X$ be a finite set of variables (whose ordered vector is $\bar{x}$) and $\dom{D}$ be a (possibly infinite) data domain. For simplicity, we assume that all variables are defined over $\dom{D}$. The elements of $\dom{D}$ are also used as constants in arithmetic expressions and constraints.

\begin{definition} \label{def:action_alpha}
{An {\em action alphabet} over $G$ and $X$ is }
$\alpha = \mathcal{G} \cup \mathcal{E} \cup \mathcal{C}$ where:
\begin{enumerate}
	\item $\mathcal{G} \subseteq \{\ g?x_1, g!x_1 , (g?x_1, g!x_2), (g!x_1, g?x_2)\  :~ g\in G, x_1, x_2\in X \} $ is a finite set of \emph{communication actions}.  %where $G$ is a finite set of communication channels.  
\begin{itemize}
    \item $g?x$ is a \emph{read} action of a value to the variable $x$ through channel $g$.
    \item $g!x$ is a \emph{write} action of the value of $x$ on channel $g$. We use $g*x$ to indicate some action, either \emph{read} or \emph{write}, through $g$.
    \item The pairs $(g?x_1, g!x_2)$ and $(g!x_1, g?x_2)$ represent a synchronization of two programs on read-write actions over channel $g$ (defined later).
\end{itemize}
	\item $\mathcal{E} \subseteq \{\ x:=e ~:~ e\in E, x\in X   \}$  is a finite set of {\em assignment statements}, where $E$ is a set of %arithmetic 
	expressions over $X\cup \dom{D}$. For an expression $e$, we denote by $e[\bar{x}\leftarrow \bar{d}]$ the expression over $\dom{D}$ in which the variables $\bar{x} \subseteq X$ are assigned with the values of $\bar{d}\subseteq\dom{D}$.  \label{item:assign_stat}
	\item $\mathcal{C}$ is a finite set of {constraints} over $X\cup \dom{D}$. 
\end{enumerate}
\end{definition}

\begin{definition} \label{def:comm}
A \emph{communicating program} (or, a program) is $M = \langle Q, X, \alpha M , \delta, q_0, F\rangle$, where: 
	\begin{enumerate}
		\item $Q$ is a finite set of states and $q_0 \in Q$ is the initial state.
        \item $X$ is a finite set of variables that range over $\dom{D}$. 
		\item $\alpha M = \mathcal{G} \cup \mathcal{E} \cup \mathcal{C}$ is the action alphabet of $M$.
        \item $\delta \subseteq Q\times \alpha \times Q$ is the transition relation. We sometimes also refer to it as a function $\delta: Q\times \alpha \rightarrow 2^Q$. 
        \item $F\subseteq Q$ is the set of accepting states. 
	\end{enumerate}
\end{definition}
%From now on we will refer to communicating programs simply as \emph{programs}.

The words that are read along a communicating program are a {\em symbolic representation} of the program behaviors. We refer to such a word as a {\em trace}. Each such trace induces {\em concrete executions} of the program, which are formed by concrete assignments to the program variables in a way that conforms with the actions along the word. 

Although communicating programs are an extension of finite automata, we investigate them from a different perspective. 
Usually, an
%While in Chapter~\ref{chap:NFM} the 
automaton takes as input a word $w$ and checks whether $w$ is in the language of the automaton. 
%the computation satisfies the specification by reading the computation against the specification automaton, 
In this work we like to think of the automaton as the generator of the behavior, as it describes the program. Therefore, we begin with a run of the program, and induce traces from the run, and not the other way around. 
We now formally define these notions.

\begin{definition}
%As defined in Chapter~\ref{sec:prelim_DFA}, 
A run in a program automaton ${M}$ is a finite sequence of states and actions $r = \langle q_0, a_1, q_1 \rangle \ldots \langle q_{n-1}, a_n, q_n \rangle$, starting with the initial state $q_0$,  such that $\forall 0\leq i < n$ we have $\langle q_i, a_{i+1}, q_{i+1}\rangle  \in \delta$.
The \emph{induced trace} of $r$ is the sequence $t = (a_1,\ldots, a_n)$ of the  actions in $r$.
If $q_n$ is accepting, then $t$ is an {\em accepted trace} of $M$.
\end{definition}

From now on we assume that every trace we discuss is induced by some run.
We turn to define the concrete executions of the program. 
% A trace $t$ induces {\em concrete runs}, in which the variables are assigned values from $D$ in accordance with the actions along $t$.

\begin{definition}\label{def:exec}
Let $t = (a_1, \ldots, a_{n})$ be a trace and let $(\beta_0,\ldots,$ $ \beta_{n})$ be a sequence of valuations (i.e., assignments to the program variables).\footnote{Such valuations are usually referred to as states. We do not use this terminology here in order not to confuse them with the states of the automaton.} Then a sequence $\pi = (\beta_0, a_1, \beta_1, a_2,\ldots, a_n, \beta_{n})$ is an \emph{execution} of $t$ if the following holds.
\begin{enumerate} 
		\item $\beta_0$ is an arbitrary valuation.
		\item If $a_i = g?x$, then $\beta_i(y) = \beta_{i-1}(y)$ for every $y \neq x$. 
        %and $\beta_i(x)$ is assigned arbitrarily. 
        Intuitively, $x$ is arbitrarily assigned by the read action, and the rest of the variables are unchanged.
		\item If $a_i$ is an assignment $x:=e$, then $\beta_i(x) = e[\bar{x}\leftarrow \beta_{i-1}(\bar{x})]$ and $\beta_i(y) = \beta_{i-1}(y)$ for every $y \neq x$.
        \item If $a_i = (g?x, g!y)$ then $\beta_i(x) = \beta_{i-1}(y)$ and  $\beta_i(z) = \beta_{i-1}(z)$ for every $z \neq x$. That is, the effect of a synchronous communication on a channel is that of an assignment. \label{item:sync}
		\item If $a_i$ does not involve a read or an assignment, then $\beta_i = \beta_{i-1}$.
		\item Finally, if $a_i$ is a constraint in $\mathcal{C}$, then  $\beta_i(\bar{x})\vDash a_i$ (and since $a_i$ does not change the variable assignments, then  $\beta_{i-1}(\bar{x})\vDash a_i$ holds as well).
\end{enumerate} 
We say that $t$ is {\em feasible} if there exists an execution of $t$. 
\end{definition}

 The \emph{symbolic language} of $M$, denoted $\mathcal{T}(M)$, is the set of all \emph{accepted} traces induced by runs of $M$. The \emph{concrete language} of $M$, denoted $\mathcal{L}(M)$, 
is the set of all executions of accepted traces in $\mathcal{T}(M)$. 
We will mostly be interested in feasible traces, which represent (concrete) executions of the program.

\begin{example}
\begin{itemize}
    \item[--] The trace $( x := 2 \cdot y,\, g?x, \, y := y+1,\,  g!y )$ is feasible, as it has an execution $(x=1, y=3), (x=6, y=3), (x=20, y=3), (x=20, y=4), (x=20, y=4)$.
    \item[--] The trace $(g?x,\, x:=x^2\, , x<0)$ is not feasible since no $\beta$ can satisfy the constraint $x<0$ if $x:=x^2$ is executed beforehand. 
\end{itemize}

\end{example}

\subsection{Parallel Composition} 
\label{sec:parallel}
    
We now describe and define the parallel composition of two communicating programs, and the way in which they communicate.
    
Let $M_1$ and $M_2$ be two programs, where $M_i = \langle Q_i, X_i,$ $ \alpha_i, \delta_i, {q_0}^i, F_i \rangle$ for $i=1,2$. Let $G_1, G_2$ be the sets of communication channels occurring in actions of $M_1, M_2$, respectively. We assume that $X_1 \cap X_2 = \emptyset$.

The \emph{interface alphabet} $\alpha I$ of $M_1$ and $M_2$ consists of all communication actions on channels that are common to both components. 
That is, $\alpha I = \{ \, g?x, \, g!x ~:~ g\in G_1 \cap G_2,\,  x\in X_1 \cup X_2  \}$.

In {\em parallel composition}, the two components synchronize on their communication interface only when one component writes data through a channel, and the other reads it through the same channel. The two components cannot synchronize if both are trying to read or both are trying to write.  
We distinguish between communication of the two components with each other (on their common channels), and their communication with their environment. In the former case, the components must ``wait'' for each other in order to progress together. In the latter case, the communication actions of the two components interleave asynchronously.

\begin{definition} \label{def:parallel}
The \emph{parallel composition} of $M_1$ and $M_2$, denoted $M_1 || M_2$,  is the program $M= \langle Q, x, \alpha, \delta, {q_0}, F\rangle $, defined as follows.

\begin{enumerate}
	\item $Q = (Q_1 \times Q_2) \cup (Q_1'\times Q_2')$, where $Q_1'$ and $Q_2'$ are new copies of $Q_1$ and $Q_2$, respectively. The initial state is $q_0= (q_0^1, q_0^2)$.
    \item $X = X_1 \cup X_2$.
    \item $\alpha = \{\, (g? x_1 , g! x_2), (g!x_1 , g? x_2) : g*x_1 \in (\alpha_1 \cap \alpha I) \ \mbox{and}\ g* x_2 \in (\alpha_2 \cap \alpha I) \} \  \cup \ ((\alpha_1\cup \alpha_2) \setminus \alpha I) $.\label{item:alpha}
    %\item $\alpha = \{\, (g* x_1, \, g* x_2) \, | \, g*x_1, g*x_2 \in \alpha G \, \}  \cup ((\alpha_1\cup \alpha_2) \setminus \alpha G) $. 
    
    That is, the alphabet includes pairs of read-write communication actions on channels that are common to $M_1$ and $M_2$. It also includes individual actions of $M_1$ and $M_2$, which are not communications on common channels.
	\item $\delta$ is defined as follows. 
    \begin{enumerate}
    \item[(a)] For $(g* x_1 , g* x_2) \in \alpha$\footnote{Note that according to item~\ref{item:alpha}, one of the actions must be a read action and the other one is a write action.}:
    \begin{enumerate} 
    \item $\delta ((q_1, q_2), (g* x_1 , g* x_2)) = \{(q_1', q_2')_g\} $.
	 \item $\delta((q_1',q_2' )_g, x_1 = x_2) \ =  
	 \{ (p_1, p_2) ~|~ p_1\in \delta_1 (q_1, g*x_1), p_2\in \delta_2(q_2, g*x_2)\} $.
    \end{enumerate}
    That is, when a communication is performed synchronously in both components, the data is transformed through the channel from the writing component to the reading component. As a result, the values of $x_1$ and $x_2$ equalize. This is enforced in $M$ by adding a transition labeled by the constraint $x_1 = x_2$ that immediately follows the synchronous communication.\footnote{Note that, equality is implied by Definition~\ref{def:exec} item~\ref{item:sync}. Here it is included syntactically to emphasize this fact. In-fact, when implementing feasibility checks this equality has to be included explicitly.}
    \item[(b)] For $a \in \alpha_1 \setminus \alpha I $ we set $\delta ((q_1, q_2),a) =  \{ (p_1, q_2) ~|~ p_1\in \delta_1 (q_1, a) \})$.
    \item[(c)] For $a \in \alpha_2 \setminus \alpha I $ we set $\delta ((q_1, q_2),a) =  \{ (q_1, p_2) ~|~ p_2\in \delta_2 (q_2, a) \}$.
    %\item For $c_1 \in \mathcal{C}_1, c_2 \in \mathcal{C}_2$, we define $\delta ((q_1, q_2),c_1 \wedge c_2) =  (\delta_1 (q_1, c_1), \delta_2 (q_2, c_2))$.
    \end{enumerate}
        That is, on actions that are not in the interface alphabet, the two components interleave.
    \item $F = F_1 \times F_2$
\end{enumerate}
\end{definition}

\begin{definition} \label{def:mt}
For a trace $t$, we define $M_t$ to be the communicating program which only follows the trace $t$, and has no other transitions. Then, we define $M||t$ to be the composition $M || M_t$ (and similarly for $t || M$). 
\end{definition}

Figure~\ref{Fig-composition} demonstrates the parallel composition of components $M_1$ and $M_2$. The program $M = M_1 || M_2$ reads a password from the environment through channel $\mathit{pass}$. The two components synchronize on channel $\mathit{verify}$. Assignments to $x$ are interleaved with reading the value of $y$ from the environment.

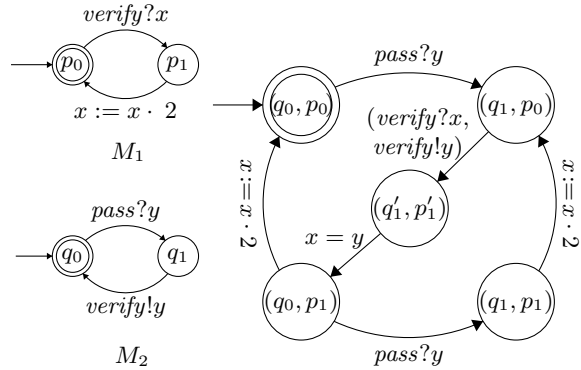
\begin{figure}[h]
	\centering
	\begin{tabular}{cc} %{>{\centering\arraybackslash}m{1in}>{\centering\arraybackslash}m{4in}}
 \begin{tikzpicture}[scale=0.09]
\tikzstyle{every node}+=[inner sep=0pt]
\draw [black] (24.8,-13.3) circle (3);
\draw (24.8,-13.3) node {$p_0$};
\draw [black] (24.8,-13.3) circle (2.4);
\draw [black] (40.4,-13.3) circle (3);
\draw (40.4,-13.3) node {$p_1$};
%\draw [black] (33,-26.2) circle (3);
\draw (33,-26.2) node {$M_1$};
\draw [black] (24.8,-41.6) circle (3);
\draw (24.8,-41.6) node {$q_0$};
\draw [black] (24.8,-41.6) circle (2.4);
\draw [black] (40.4,-41.6) circle (3);
\draw (40.4,-41.6) node {$q_1$};
%\draw [black] (33.6,-56.4) circle (3);
\draw (33.6,-56.4) node {$M_2$};
\draw [black] (15.7,-13.3) -- (21.8,-13.3);
\fill [black] (21.8,-13.3) -- (21,-12.8) -- (21,-13.8);
\draw [black] (26.444,-10.809) arc (136.4462:43.5538:8.494);
\fill [black] (38.76,-10.81) -- (38.57,-9.89) -- (37.84,-10.57);
\draw (32.6,-7.67) node [above] {$\mathit{verify}?x$};
\draw [black] (38.707,-15.758) arc (-44.58648:-135.41352:8.575);
\fill [black] (26.49,-15.76) -- (26.7,-16.68) -- (27.41,-15.98);
\draw (32.6,-18.81) node [below] {$x:=x\cdot\mbox{ }2$};
\draw [black] (26.699,-39.295) arc (130.97087:49.02913:9);
\fill [black] (38.5,-39.3) -- (38.23,-38.39) -- (37.57,-39.15);
\draw (32.6,-36.59) node [above] {$\mathit{pass}?y$};
\draw [black] (16.4,-41.6) -- (21.8,-41.6);
\fill [black] (21.8,-41.6) -- (21,-41.1) -- (21,-42.1);
\draw [black] (38.876,-44.165) arc (-41.05029:-138.94971:8.322);
\fill [black] (26.32,-44.17) -- (26.47,-45.1) -- (27.23,-44.44);
\draw (32.6,-47.52) node [below] {$\mathit{verify}!y$};
\end{tikzpicture}

	%\hspace{0.1cm} 
	
	&
	
{
\begin{tikzpicture}[scale=0.17]
\tikzstyle{every node}+=[inner sep=0pt]
\draw [black] (24.2,-24.7) circle (3);
\draw [black] (24.2,-24.7) circle (2.4);
\draw (24.2,-24.7) node {$(\!q_0 , p_0\!)$};
\draw [black] (24.2,-40.1) circle (3);
\draw (24.2,-40.1) node {$( q_0 , p_1  )$};
\draw [black] (41,-24.7) circle (3);
\draw (41,-24.7) node {$( q_1 , p_0  )$};
\draw [black] (41,-40.1) circle (3);
\draw (41,-40.1) node {$( q_1 , p_1  )$};
\draw [black] (32.7,-32.8) circle (3);
\draw (32.7,-32.8) node {$( q_1' , p_1'  )$};
\draw [black] (17.3,-24.7) -- (21.2,-24.7);
\fill [black] (21.2,-24.7) -- (20.4,-24.2) -- (20.4,-25.2);
\draw [black] (26.851,-23.307) arc (112.09981:67.90019:15.28);
\fill [black] (38.35,-23.31) -- (37.8,-22.54) -- (37.42,-23.47);
\draw (32.6,-21.68) node [above] {\small$\mathit{pass}?y$};
\draw [black] (38.428,-41.633) arc (-65.35757:-114.64243:13.977);
\fill [black] (38.43,-41.63) -- (37.49,-41.51) -- (37.91,-42.42);
\draw (32.6,-43.41) node [below] {\small$\mathit{pass}?y$};
\draw [black] (22.494,-37.643) arc (-152.73421:-207.26579:11.444);
\fill [black] (22.49,-27.16) -- (21.68,-27.64) -- (22.57,-28.1);
\draw (20.72,-32.4) node [left] {\small\begin{turn}{270} 
\small$x\!:=\!x\cdot 2$
\end{turn}
};
\draw [black] (42.781,-27.103) arc (28.74781:-28.74781:11.014);
\fill [black] (42.78,-27.1) -- (42.73,-28.04) -- (43.6,-27.56);
\draw (44.64,-32.4) node [right] {
\begin{turn}{270} 
\small$x\!:=\!x\cdot 2$
\end{turn}
};
\draw [black] (38.85,-26.8) -- (34.85,-30.7);
\fill [black] (34.85,-30.7) -- (35.77,-30.5) -- (35.07,-29.79);
\draw (33.33,-29.27) node [above] {
\begin{tabular}{c}
		$(\mathit{verify}?x,$\\ $\mathit{verify}!y)$
			\end{tabular}
%\tiny$(\mathit{verify}?x, \mathit{verify}!y)$
};
\draw [black] (30.42,-34.75) -- (26.48,-38.15);
\fill [black] (26.48,-38.15) -- (27.41,-38) -- (26.76,-37.24);
\draw (26.81,-35.96) node [above] {$x=y$};
\end{tikzpicture}}

	\end{tabular}
	\caption{Components $M_1$ and $M_2$  and their parallel composition $M_1 || M_2$.}\label{Fig-composition}
\end{figure}

\subsection{Regular Properties and their Satisfaction} \label{sec:safety} 
The specifications we consider are also given as some variation of communicating programs. We now define the syntax and semantics of the properties that we consider as specifications.
These are properties that can be represented as finite automata, hence the name \emph{regular}.
However, the alphabet of such automata  includes communication actions and  first-order constraints over program variables. Thus, such automata are suitable for specifying the desired and undesired behaviors of communicating programs over time.

In order to define our properties, we first need the notion of a {\em deterministic and complete} program. The definition is somewhat different from the standard definition for finite automata,
%given in Chapter~\ref{sec:prelim_DFA} 
since it takes the semantic meaning of constraints into account.
Intuitively, in a deterministic and complete program, every concrete execution has exactly one trace that induces it.

\begin{definition}\label{def:detAndComplete}
A communicating program over alphabet $\alpha$ is \emph{deterministic and complete} if for every state $q$ and for every action $a \in \alpha$ the following hold: 
\begin{enumerate}
\item \emph{Syntactic determinism and completeness.} There is exactly one state $q'$ such that $\langle q, a, q'\rangle$ is in $\delta$.\footnote{In our examples we sometimes omit the actions that lead to a rejecting sink for the sake of clarity.}
\item \emph{Semantic determinism.} If $\langle q, c_1, q'\rangle$ and $\langle q, c_2, q''\rangle$ are in $\delta$ for constraints $c_1, c_2 \in \mathcal{C}$ such that $c_1\neq c_2$ and $q' \neq q''$, then $c_1 \wedge c_2 \equiv false$. 
\item \emph{Semantic completeness.} Let $C_q$ be the set of all constraints on transitions leaving $q$. Then $(\bigvee_{c\in C_q} c) \equiv \true$.
%\item If $(q, a_1, q'), (q, a_2, q'')$ are in $\delta$ then either $a_1, a_2$ are both in $\mathcal{G}$, or $a_1, a_2$ are both in $\mathcal{C}$. 
\end{enumerate}
\end{definition}

 A {\em property} is a deterministic and complete program with no assignment actions, whose language defines the set of desired and undesired behaviors over the alphabet $\alpha P$. 

\stam{
\begin{definition}
A \emph{property} $P =  \langle Q, X, \alpha, \delta, {q_0}, F \rangle$ is a deterministic and complete  program that does not contain assignment actions. 
The set of rejecting (bad) states of $P$ is $B = Q \setminus F$. 
\end{definition}
}

A trace is accepted by a property $P$ if it reaches a state in $F$, the set of accepting states of $P$. Otherwise, it reaches a state in $Q\setminus F$, and is rejected by $P$.
%whose language defines the set of desired and undesired behaviors over the alphabet $\alpha P$.     
%
%Formally, $P =  \langle Q, X, \alpha, \delta, {q_0}, B \rangle$ where $Q, X, \alpha, \delta, q_0$ are as in the definition of programs. 
%The set $B\subseteq Q$ is the set of \emph{rejecting (bad) states} of $P$. 
%We use $\mathcal{T}_F(P)$ to denote the set of accepted traces of $P$ and $\mathcal{T}_B(P)$ to denote the erroneous, undesired traces of $P$. 

Next, we define the satisfaction relation $\vDash$ between a program and a property.
Intuitively, a program $M$ satisfies a property $P$ (denoted $M \vDash P$) if all executions induced by accepted traces of $M$ reach an accepting state in $P$. 
Thus, the accepted behaviors of $M$ are also accepted by $P$. 

A property $P$ specifies the behavior of a program $M$ by referring to communication actions of $M$ and imposing constraints over the variables of $M$.
Thus, the set of variables of $P$ is identical to that of $M$. Let $\mathcal{G}$ be the set of communication actions of $M$. Then, $\alpha P$ includes a subset of $\mathcal{G}$ as well as constraints over the variables of $M$. 
The {\em interface} of $M$ and $P$, which consists of the communication actions that occur in $P$, is defined as $\alpha I = \mathcal{G} \cap \alpha P$. 

In order to capture the satisfaction relation between $M$ and $P$, we define a {\em conjunctive composition} between $M$ and $P$, denoted $M \times P$.  
%restrictive -> conjunctive
In conjunctive composition, the two components synchronize on their common communication actions when both read or both write through the same communication channel. They interleave on constraints and on actions of $\alpha M$ that are not in $\alpha P$.

 \begin{definition}\label{def:conjunctive}
 Let $M = \langle Q_M, X_M, \alpha M, \delta_M, {q_0^M}, F_M \rangle$ be a program and $P = \langle Q_P,$ $ X_P,  \alpha P, \delta_P, {q_0^P}, F_P \rangle$  be a property, where $X_M \supseteq X_P$. The \emph{conjunctive composition}  of $M$ and $P$ is $M \times P = \langle Q, X, \alpha, \delta, {q_0}, F \rangle$, where:
 \begin{enumerate}
	\item $Q = Q_M \times Q_P$. The initial state is $q_0= (q_0^M, q_0^P)$.
    \item $X = X_M$.
    \item $\alpha = \{\, g*x , ~g*y,  , ~(g?x,g!y), ~(g!x, g?y) ~:~ g*x, g*y,~ (g*x, g*y) \in  \alpha I \}  \cup ((\alpha M\cup \alpha P) \setminus \alpha I))$.
    
    %\item $\alpha = \{\, (g* x_1, \, g* x_2) \, | \, g*x_1, g*x_2 \in \alpha G \, \}  \cup ((\alpha_1\cup \alpha_2) \setminus \alpha G) $. 
    That is, the alphabet includes %pairs of identical read-read and write-write 
    communication actions on channels common to $M$ and $P$. 
    %It also includes assignment actions that are common to $M_1$ and $M_2$. Finally,
    It also includes individual actions of $M$ and $P$.
    
    {Note that communication actions of the form $(g*x, g*y)$ can only appear if $M$ is itself a parallel composition of two programs.}
	\item $\delta$ is defined as follows.
    \begin{itemize}
    \item[(a)] For $a = (g* x , g* y)$ in $\alpha I$, or $a = g*x $ in $\alpha I$, we define 
    $\delta ((q_1, q_2), a) = $ \\
    $\{ (q_M, q_P) ~|~ q_M\in\delta_M (q_1, a), q_P = \delta_P(q_2,a)\} $.\footnote{Recall that $P$ is deterministic thus its transition relation only corresponds to one state, for each letter.}
    %\item for $a \in (\mathcal{E}_1 \cap \mathcal{E}_2)$ we define $\delta ((q_1, q_2),a) =  (\delta_1 (q_1, a), \delta_2(q_2,a ))$.
   % That is, assignments that are common to both components (they are in their interface) are performed synchronously. 
    \item[(b)] For $a \in \alpha M \setminus \alpha I $ we define  $\delta ((q_1, q_2),a) =  \{ (q_M, q_2) ~|~ q_M \in \delta_M (q_1,a) \}$.
    \item[(c)] For $a \in \alpha P \setminus \alpha I $ we define  $\delta ((q_1, q_2),a) =  \{ (q_1, \delta_P (q_2, a)) \}$.
    \end{itemize}
    That is, on actions that are not common communication actions to $M$ and $P$, the two components interleave.
    \item $F = F_M \times B_P$, where $B_P = Q_P\setminus F_P$.
\end{enumerate}
\end{definition}
Note that accepted traces in $M \times P$ are those that are accepted in $M$ and rejected in $P$.
Such traces are called \emph{error traces} and their corresponding executions are called \emph{error executions}. 
Intuitively, an error execution is an execution along $M$ which violates the properties modeled by $P$. Such an execution either fails to synchronize on the communication actions, or reaches a point in the computation in which its assignments violate some constraint described by $P$. These executions are manifested in the traces that are accepted in $M$ but are composed with matching traces that are rejected in $P$.
%An error run is actually an evidence to a violating behavior of $M$.
%
We can now formally define when a program satisfies a property.
\begin{definition}
For a program $M$ and a property $P$, we define $M \vDash P$ iff $M \times P$ contains no feasible accepted traces.
\end{definition}
Thus, a feasible error trace in $M \times P$ is an evidence to $M \not\vDash P$, since it indicates the existence of an execution that violates $P$.

\begin{example}

Consider the program $M$, the property $P$ and the partial construction of $M\times P$ presented in Figure~\ref{Fig-restric}. The property $P$ requires every verified password $y$ to be of length at least 4. It is easy to see that $M\nvDash P$, since the trace $t=(\mathit{password}?y,$  $y>0, \mathit{verify}!y, y<1000)$ is a feasible error trace in $M\times P$.

\end{example}

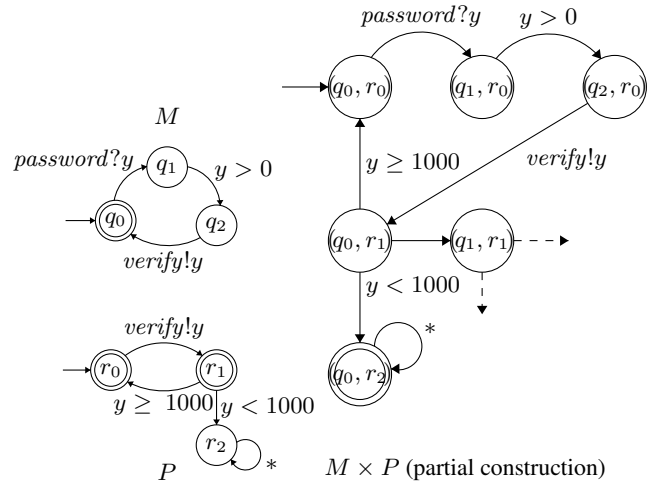
\begin{figure}[h]
	\centering
	\begin{tabular}
	%{m{2.5in}m{4in}}
	{cc} 
	%{>{\centering\arraybackslash}m{1in}>{\centering\arraybackslash}m{4in}}

\begin{tikzpicture}[scale=0.09]
\tikzstyle{every node}+=[inner sep=0pt]
%\draw [black] (5.8,-13.5) circle (3);
%\draw (5.8,-13.5) node {$M$};
%\draw [black] (5.8,-35.6) circle (3);
\draw (28.3,-50.6) node {$P$};
\draw [black] (20.7,-13.5) circle (3);
\draw (20.7,-13.5) node {$q_0$};
\draw [black] (20.7,-13.5) circle (2.4);
\draw [black] (28.3,-5.6) circle (3);
\draw (28.3,2) node {$M$};
\draw (28.3,-5.6) node {$q_1$};
\draw [black] (35.6,-14.2) circle (3);
\draw (35.6,-14.2) node {$q_2$};
\draw [black] (20,-35.6) circle (3);
\draw (20,-35.6) node {$r_0$};
\draw [black] (20,-35.6) circle (2.4);
\draw [black] (35.6,-35.6) circle (3);
\draw (35.6,-35.6) node {$r_1$};
\draw [black] (35.6,-35.6) circle (2.4);
\draw [black] (35.6,-46.7) circle (3);
\draw (35.6,-46.7) node {$r_2$};
\draw [black] (13.3,-13.5) -- (17.7,-13.5);
\fill [black] (17.7,-13.5) -- (16.9,-13) -- (16.9,-14);
\draw [black] (20.48,-10.537) arc (-189.6015:-258.18088:6.209);
\fill [black] (25.33,-5.5) -- (24.45,-5.17) -- (24.65,-6.15);
\draw (23.5,-4.5) node [left] {$\mathit{password}?y$};
\draw [black] (31.27,-5.531) arc (77.54058:3.1111:6.234);
\fill [black] (36.15,-11.28) -- (36.61,-10.45) -- (35.61,-10.51);
\draw (35.23,-6.14) node [right] {$y>0$};
\draw [black] (33.213,-15.998) arc (-61.61275:-123.76679:9.992);
\fill [black] (22.91,-15.51) -- (23.3,-16.37) -- (23.85,-15.54);
\draw (27.65,-18.03) node [below] {$\mathit{verify}!y$};
\draw [black] (12.9,-35.6) -- (17,-35.6);
\fill [black] (17,-35.6) -- (16.2,-35.1) -- (16.2,-36.1);
\draw [black] (22.091,-33.466) arc (126.59775:53.40225:9.575);
\fill [black] (33.51,-33.47) -- (33.16,-32.59) -- (32.57,-33.39);
\draw (27.8,-31.08) node [above] {$\mathit{verify}!y$};
\draw [black] (33.165,-37.338) arc (-62.00467:-117.99533:11.43);
\fill [black] (22.43,-37.34) -- (22.91,-38.15) -- (23.38,-37.27);
\draw (27.8,-39.18) node [below] {$y\geq\mbox{ }1000$};
\draw [black] (35.6,-38.6) -- (35.6,-43.7);
\fill [black] (35.6,-43.7) -- (36.1,-42.9) -- (35.1,-42.9);
\draw (36.1,-41.15) node [right] {$y<1000$};
\draw [black] (38.56,-46.293) arc (125.56505:-162.43495:2.25);
\draw (42.77,-49.54) node [right] {$*$};
\fill [black] (37.72,-48.8) -- (37.78,-49.74) -- (38.6,-49.16);
\end{tikzpicture}

	\hspace{-0.7cm} 
	&
	
{
\begin{tikzpicture}[scale=0.14]
\tikzstyle{every node}+=[inner sep=0pt]
\draw [black] (15,-12.6) circle (3);
\draw (15,-12.6) node {$(\!q_0,r_0\!)$};
\draw [black] (26.6,-12.6) circle (3);
\draw (26.6,-12.6) node {$(\!q_1,r_0\!)$};
\draw [black] (39.2,-12.6) circle (3);
\draw (39.2,-12.6) node {$(\!q_2,r_0\!)$};
\draw [black] (15,-27.2) circle (3);
\draw (15,-27.2) node {$(\!q_0,r_1\!)$};
\draw [black] (15,-39.9) circle (3);
\draw (15,-39.9) node {$(\!q_0,r_2\!)$};
\draw [black] (15,-39.9) circle (2.4);
\draw [black] (26.6,-27.2) circle (3);
\draw (26.6,-27.2) node {$(\!q_1,r_1\!)$};
%\draw [black] (37.8,-27.2) circle (3);
%\draw [black] (26.6,-37.2) circle (3);
%\draw [black] (29.3,-48.9) circle (3);
\draw (25,-48.9) node {$M\times P\text{ (partial construction)}$};
\draw [black] (7.6,-12.6) -- (12,-12.6);
\fill [black] (12,-12.6) -- (11.2,-12.1) -- (11.2,-13.1);
\draw [black] (16.051,-9.826) arc (144.50913:35.49087:5.832);
\fill [black] (25.55,-9.83) -- (25.49,-8.88) -- (24.68,-9.46);
\draw (20.8,-6.88) node [above] {$\mathit{password}?y$};
\draw [black] (27.904,-9.928) arc (140.67168:39.32832:6.458);
\fill [black] (37.9,-9.93) -- (37.78,-8.99) -- (37,-9.63);
\draw (32.9,-7.06) node [above] {$y>0$};
\draw [black] (15,-30.2) -- (15,-36.9);
\fill [black] (15,-36.9) -- (15.5,-36.1) -- (14.5,-36.1);
\draw (15.5,-31.55) node [right] {$y<1000$};
\draw [black] (16.38,-37.249) arc (180.23337:-107.76663:2.25);
\draw (21.03,-35.52) node [right] {$*$};
\fill [black] (17.94,-39.38) -- (18.75,-39.88) -- (18.74,-38.88);
\draw [black] (36.63,-14.15) -- (17.57,-25.65);
\fill [black] (17.57,-25.65) -- (18.51,-25.67) -- (18,-24.81);
\draw (34.6,-18.41) node [below] {$\mathit{verify}!y$};
\draw [black] (15,-24.2) -- (15,-15.6);
\fill [black] (15,-15.6) -- (14.5,-16.4) -- (15.5,-16.4);
\draw (15.5,-19.9) node [right] {$y\geq 1000$};
\draw [black] (18,-27.2) -- (23.6,-27.2);
\fill [black] (23.6,-27.2) -- (22.8,-26.7) -- (22.8,-27.7);
\draw [black , dashed] (29.6,-27.2) -- (34.8,-27.2);
\fill [black] (34.8,-27.2) -- (34,-26.7) -- (34,-27.7);
\draw [black, dashed] (26.6,-30.2) -- (26.6,-34.2);
\fill [black] (26.6,-34.2) -- (27.1,-33.4) -- (26.1,-33.4);
\end{tikzpicture}
}

	\end{tabular}
	\caption{Partial conjunctive composition of $M$ and $P$.% Error traces are all traces that end at state $(q_0, r_2)$. 
	\label{Fig-restric}}%\vspace{-0.5cm}
\end{figure}
\section{Traces in the Composed System}\label{sec:tarces}

Before we discuss our framework for compositional verification and repair of communicating systems, we prove some properties of traces in the composed system. We later use these properties in order to prove that our framework is sound and complete (Section~\ref{sec:soundAGR}), 
and to prove the correctness and termination of our algorithm (Section~\ref{sec:abduction} and Section~\ref{Sec:termination}).

\begin{definition}
Let $t$ be a trace over alphabet $\alpha$, and let $\alpha' \subseteq \alpha$. We denote by $t\downarrow_{\alpha'}$ the \emph{restriction of $t$ on $\alpha'$}, which is the
trace obtained from $t$ by omitting all letters in $t$ that are not in $\alpha'$.
If $\alpha$ contains a communication action $a = (g*x, g*y)$ and we have $g*x\in \alpha'$ %or $g*y\in \alpha'$ 
then the restriction $t\downarrow_{\alpha'}$ includes the corresponding communication, $g*x$, and similarly for $g*y\in\alpha'$. 
\end{definition}

\begin{example}
Let $\alpha = \{g_1!x, x:=x+1, x<10, (g_2?x, g_2!y)\}$, and $\alpha' = \{g_1!x, x:=x+1, g_2?x\}$. Then $(g_2?x, g_2!y)\downarrow_{\alpha'} = g_2?x$, and for  $$t=((g_2?x, g_2!y), x:=x+1,x:=x+1, x<10, g_1!x)$$ 
we have $$t\downarrow_{\alpha'} = (g_2?x, x:=x+1, x:=x+1, g_1!x).$$
\end{example}

For traces in the conjunctive and parallel compositions, we have the following two lemmas. 

\begin{lemma} \label{lemma:traces}
Let  $M$ be a communicating program and $P$ be a property, and let $t$ be a trace of $M\times P$. Then $t\downarrow_{\alpha M}$ is a trace of $M$.
\end{lemma}

\begin{lemma} \label{lemma:traces12}
Let  $M_1, M_2$ be two programs, and let $t$ be a trace of $M_1 || M_2$. Then $t\downarrow_{\alpha M_1}$ is a trace of $M_1$ and $t\downarrow_{\alpha M_2}$ is a trace of $M_2$.
\end{lemma}

We prove Lemma~\ref{lemma:traces}. The proof of Lemma~\ref{lemma:traces12} is similar, with the special care of communications actions, and is provided in Appendix~\ref{app:traces}.

\begin{proof}[Lemma~\ref{lemma:traces}]
Let $M = \langle Q_M, X_M, \alpha M, \delta_M, q_0^M, F_M\rangle $ and $P = \langle Q_P, X_P, \alpha P, \delta_P, q_0^P, F_P\rangle $, and denote $M\times P = \langle Q, X_M, \alpha , \delta, q_0, F\rangle$.\footnote{Recall that the set of variables $X_P$ is a subset of $X_M$.} 

Let $r = \langle q_0, c_1, q_1\rangle \ldots \langle q_{m-1}, c_m, q_m \rangle$ be the run in $M\times P $ such that $t$ is induced from $r$. 
Denote by $t_M = t\downarrow_ {\alpha M}$  the trace $t_M = (c_{i_1}, \cdots , c_{i_n})$. 

We first observe the following. If $(a_1, \cdots, a_k)$ is a trace of $M\times P$ such that $\forall i: a_i\notin \alpha M$, and $q = (q_M, q_P^0)$ is the state in $M\times P$ before reading $a_1$,  then $\forall i\geq 1: \exists q_P^i: \delta((q_M, q_P^{i-1}), a_i) = (q_M, q_P^{i})$, that is, when reading a trace that does not contain letters from $\alpha M$, the program $M\times P$ only advances on the $P$ component. This is true since by the definition of $\delta$, if $a_i$ is not in $\alpha M$, then $\delta((q_M, q_P), a_i) = (q_M, \delta_P(q_P, a_i))$. %Informally, it means that $\delta$ only advances in the $P$ component and not in the $M$ component.  

We now inductively prove that $(c_{i_1}, \cdots , c_{i_j})$ is a trace of $M$ for every $1\leq j \leq n$. In particular, for $j=n$ this means that $t_M$ is a trace of $M$.

Let $j:=1$ and denote $k:= i_1$. Then $c_1, \ldots,c_{k-1} \notin \alpha M $ since $k$ is the first index of $t$ for which $c_k \in \alpha M$. Thus, $\forall 1 < i<k: \exists q^P_i: \delta((q_0^M, q_{i-1}^P), c_i) =(q_0^M, q_i^P)$.
For $c_{i_1} = c_k \in\alpha M$, by the definition of $\delta$, we have $$\delta((q_0^M, q_{k-1}^P), c_{i_1}) = (\delta_M(q_0^M, c_{i_1}), q')$$ for some $q'\in Q_P$. Then indeed, $\langle q_0^M, c_{i_1}, \delta_M(q_0^M, c_{i_1}) \rangle$ is a run in $M$, making $(c_{i_1})$ a trace of $M$. 

Let $1< j \leq n$, and assume $t_{j-1} = (c_{i_1}, \cdots, c_{i_{j-1}})$ is a trace of $M$. Let $\langle q_0, c_{i_1}, q_1 \rangle \ldots $ $\langle q_{j-2}, c_{i_{j-1}}, q_{j-1} \rangle$ be a run that induces $t_{j-1}$. Denote $i_{j-1} = k, i_j = k+m $ for some $m>0$.
Then, as before, $c_{k+1}, \ldots , c_{k+m-1} \notin \alpha M$, thus
$\forall k<l<k+m:\exists q^P_l:    
\delta(q_{j-1}, q^P_{l-1}), c_l) = (q_{j-1}, q^P_l)$. For $c_{i_j}$ it holds that  
$\delta((q_{j-1}, q^P_{k+m-1}), c_{i_j}) = (\delta_M(q_{j-1}, c_{i_j}), q'))$ for some $q'\in Q_P$. Thus $(c_{i_1},\cdots, c_{i_j})$ is a trace of $M$, as needed. 
\qed
\end{proof}

We now discuss the \emph{feasibility} of traces in the composed system.

\begin{lemma}\label{lemma:feasible_traces}
Let  $M$ be a program and $P$ be a property, and let $t$ be a \textbf{feasible} trace of $M\times P$. Then $t\downarrow_{\alpha M}$ is a \textbf{feasible} trace of $M$.
%Let $t_M$ be a trace of $M$, and $t_p$ be a trace of $P$. Then if $t \in t_{M} \times t_{p}$ is a feasible error trace of $M\times P$, then $t_M$ is a feasible accepting trace of $M$. 
\end{lemma}
\begin{proof}%[Proof of Lemma~\ref{lemma:feasible_traces}]
Let $t\in\T( M\times P)$ be a feasible trace. Then, there exists an execution $u$ on $t$. Denote $t = (b_1, \cdots, b_n)$ and $u = (\beta_0, b_1, \beta_1, \cdots, b_n, \beta_n)$. We inductively construct an execution $e$ on $t\downarrow _{\alpha M}$. The existence of such an execution $e$  proves that $t\downarrow _{\alpha M}$ is feasible.

Let $t\downarrow _{\alpha M} = (c_1, \cdots, c_k)$.
We set $e=(\gamma_0, c_1, \gamma_1, \cdots ,$ $ c_k, \gamma_k)$ where $\gamma_0, \ldots, \gamma_k$ are defined as follows. 
\begin{enumerate}
    \item Set $j:=0, i:=0$.
    \item Define $\gamma_0:= \beta_0$ and set $j:=j+1$. 
    \item Repeat until $j = k$ : 
    \begin{itemize}
        \item[--] Let $i'> i$ be the minimal index such that $b_{i'} = c_j$.
        \item[--] Define $\gamma_j : = \beta_{i'}$ and set $j:= j+1, i:= i'+1$.
    \end{itemize}  
\end{enumerate}
    Note that for each $i<l<i'$ is holds that $b_l$ is a constraint. Indeed, by the definition of conjunctive composition (Definition~\ref{def:conjunctive}), if $b_l$ is not a constraint, then $b_l \in \alpha M$. But in that case, $b_l$ must synchronize with some alphabet letter in $t\downarrow _{\alpha M}$, contradicting the fact that $i'$ is the minimal index for which $b_{i'} = c_j$. 
    Thus, since $u$ is an execution, and for all $i<l<i':$ $b_l$ is a constraint, 
    it holds that $\forall i\leq l <i':\beta_i = \beta_l$. In particular, $\beta_{i'-1} = \beta_i = \gamma_{j-1}$. Now, since $b_{i'} = c_j$, we can assign $\gamma_j$ to be the same as $\beta_{i'}$ and result in a valid assignment. 
    Thus, $e$ is a valid execution on $t\downarrow _{\alpha M}$, making $t\downarrow _{\alpha M}$ feasible, as needed.
\qed
\end{proof}

\begin{lemma}\label{lemma:feasible_traces12}
Let  $M_1, M_2$ be  two programs, and let $t$ be a \textbf{feasible} trace of $M_1|| M_2$. Then $t\downarrow_{\alpha M_i}$ is a \textbf{feasible} trace of $M_i$ for $i\in{1,2}$.
\end{lemma}

The proof of Lemma~\ref{lemma:feasible_traces12} is different from the proof of Lemma~\ref{lemma:feasible_traces}, since here we can no longer use the exact same assignments as the ones of the run on $M_1 || M_2$. In the case of $M\times P$, the set of  variables of $M\times P$ is equal to that of $M$, and the two runs only differ on the constraints that are added to the trace of $M\times P$. 
In $M_1||M_2$, on the other hand, 
$M_1$ and $M_2$ are defined over two disjoint sets of variables. Nevertheless, The proof is similar to the proof of Lemma~\ref{lemma:feasible_traces}, and is provided in Appendix~\ref{app:traces}. 

\label{sec:traces}
\section{The Assume-Guarantee Rule for Communicating Systems} \label{AGR:AR_rule_reg}\label{sec:rule}

Let $M_1$ and $M_2$ be two programs, and let $P$ be a property. The classical Assume-Guarantee (AG) proof rule~\cite{Pnueli85} assures that if we find an assumption $A$ (in our case, a communicating program) such that $M_1||A \vDash P$ and $M_2 \vDash A$ both hold, then $M_1 || M_2 \vDash P$ holds as well. For labeled transition systems over a finite alphabet (LTSs)~\cite{DBLP:conf/tacas/CobleighGP03}, the AG-rule is guaranteed to either prove correctness or return a real (non-spurious) counterexample.
The work in~\cite{DBLP:conf/tacas/CobleighGP03} relies on the \lstar\ algorithm~\cite{DBLP:journals/iandc/Angluin87} for learning an assumption $A$ for the AG-rule. In particular, \lstar\ aims at learning $A_w$, the weakest assumption for which $M_1 || A_w \vDash P$. A crucial point of this method is the fact that $A_w$ is \emph{regular}~\cite{DBLP:conf/kbse/GiannakopoulouPB02}, and thus can be learned by \lstar. For communicating programs, this is not the case, as we show in Lemma~\ref{lemma:weak_non_reg}. 

\begin{definition}[Weakest Assumption]
Let $P$ be a property and let $M_1$ and $M_2$ be two programs. The weakest assumption $A_w$ with respect to $M_1, M_2$ and $P$
has the language $\aut{L}(A_w) = \{ w\in (\alpha M_2)^* :~ M_1 || w\vDash P \}$. That is, $A_w$ is the set of all words over the alphabet of $M_2$ that together with $M_1$ satisfy $P$. 
\end{definition}

\begin{lemma}\label{lemma:weak_non_reg}
For infinite-state communicating programs, the weakest assumption $A_w$ is not always regular. 
\end{lemma}

\begin{proof}
Consider the programs $M_1$ and $M_2$, and the property $P$ of Figure~\ref{fig:not_regular}. Let $\alpha M_2=\{ x:=0,~ y:=0, ~ x:= x+1, ~ y:= y+1, ~ \mathit{sync} \}$. 
Note that in order to satisfy $P$, after the \emph{sync} action, a trace $t$ must pass the test $x=y$.
Also note that the weakest assumption $A_w$ does not depend on the behavior of $M_2$, but only on its alphabet. 
Assume by way of contradiction that $\aut{L}(A_w)$ is a regular language, and consider the language
$$L = \{ x:=0 \} \cdot \{ y:=0 \} \cdot \{ x:=x+1, y:=y+1 \}^* \cdot \{ \mathit{sync}  \}$$
%$$L = (\alpha M_2)^* \cdot \{ x:=0 \} \cdot \{ y:=0 \} \cdot \{ x:=x+1, y:=y+1 \}^* \cdot \{ \mathit{sync}  \}$$ 
By closure properties of regular languages, it holds that 
$L$ is a regular language, and thus following our assumption, we have that $L \cap \aut{L}(A_w)$ is regular as an intersection of two regular languages. However $L \cap \aut{L}(A_w)$ is the set of all words that after the initialization $\{ x:=0\} \{y:=0\}$, contain equally many actions of the form $x:=x+1$ and $y := y+1$. That is 
%$$   L \cap \aut{L}(A_w) =  (\alpha M_2)^* \cdot \{ x:=0 \} \cdot \{ y:=0 \} \cdot L_{\mathit{eq}} \cdot \{ \mathit{sync} \}$$
$$   L \cap \aut{L}(A_w) =   \{ x:=0 \} \cdot \{ y:=0 \} \cdot L_{\mathit{eq}} \cdot \{ \mathit{sync} \}
$$
 where 
\begin{align*}
    L_{\mathit{eq}} = \{ u\in \{ x:=x+1, y:=y+1  \}^* ~:~~~~~~~~~~~~~ \\\text{num of } x:=x+1 \text{ in } u  \text{ is equal to num of } y:=y+1 \text{ in } u  \}
\end{align*}

$L_{\mathit{eq}}$ is not regular since the pumping lemma does not hold for it. For the same reason $ L \cap \aut{L}(A_w)$ is also not regular, contradicting our assumption that $ \aut{L}(A_w)$ is regular. \qed
\end{proof}

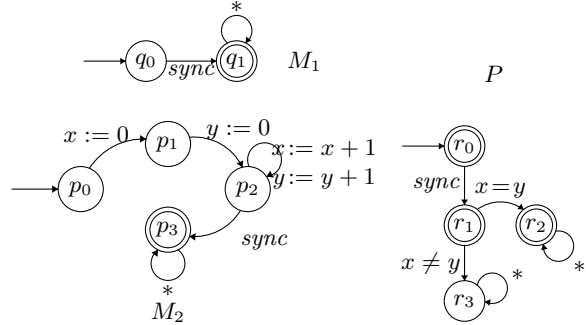
\begin{figure}[h] 
	\centering
	\begin{tabular}{cc}
    
    \begin{tikzpicture}[scale=0.09]
    
\tikzstyle{every node}+=[inner sep=0pt]
\draw [black] (17.2,-18.5) circle (3);
\draw (17.2,-18.5) node {$q_0$};
\draw [black] (30.6,-18.5) circle (3);
\draw (30.6,-18.5) node {$q_1$};
\draw (40.6,-18.5) node {$M_1$};
\draw [black] (30.6,-18.5) circle (2.4);
\draw [black] (8,-18.5) -- (14.2,-18.5);
\fill [black] (14.2,-18.5) -- (13.4,-18) -- (13.4,-19);
\draw [black] (20.2,-18.5) -- (27.6,-18.5);
\fill [black] (27.6,-18.5) -- (26.8,-18) -- (26.8,-19);
\draw (23.9,-19) node [below] {$\mathit{sync}$};
\draw [black] (29.277,-15.82) arc (234:-54:2.25);
\draw (30.6,-11.25) node [above] {$*$};
\fill [black] (31.92,-15.82) -- (32.8,-15.47) -- (31.99,-14.88);
\end{tikzpicture}

&

$P$

\\
\\

	%\hspace{0.3cm} &
		%	\begin{minipage}{.2\textwidth}
\begin{tikzpicture}[scale=0.1]\tikzstyle{every node}+=[inner sep=0pt]
\draw [black] (17.2,-18.5) circle (3);
\draw (17.2,-18.5) node {$p_0$};
\draw [black] (28.8,-12.5) circle (3);
\draw (28.8,-12.5) node {$p_1$};
\draw [black] (39.4,-18.5) circle (3);
\draw (39.4,-18.5) node {$p_2$};
\draw [black] (28.8,-23.9) circle (3);
\draw (28.8,-23.9) node {$p_3$};
\draw (28.8,-34.9) node {$M_2$};
\draw [black] (28.8,-23.9) circle (2.4);
\draw [black] (8,-18.5) -- (14.2,-18.5);
\fill [black] (14.2,-18.5) -- (13.4,-18) -- (13.4,-19);
\draw [black] (18.347,-15.745) arc (147.28808:87.41167:8.506);
\fill [black] (25.89,-11.84) -- (25.11,-11.31) -- (25.07,-12.31);
\draw (19.27,-12.28) node [above] {$x:=0$};
\draw [black] (31.64,-11.606) arc (95.36446:25.61255:7.098);
\fill [black] (38.7,-15.6) -- (38.81,-14.67) -- (37.91,-15.1);
\draw (38.14,-11.99) node [above] {$y:=0$};
\draw [black] (39.581,-15.517) arc (204.25512:-83.74488:2.25);
\draw (42,-15.13) node [right] {
\begin{tabular}{c}
			$x\!:=x+1$\\ $y\!:=y+1$
			\end{tabular}
};
\fill [black] (41.88,-16.83) -- (42.81,-16.96) -- (42.4,-16.05);
\draw [black] (38.444,-21.321) arc (-30.73866:-95.26966:7.149);
\fill [black] (31.64,-24.79) -- (32.39,-25.36) -- (32.49,-24.36);
\draw (41.62,-24.56) node [below] {$\mathit{sync}$};
\draw [black] (30.123,-26.58) arc (54:-234:2.25);
\draw (28.8,-31.15) node [below] {$*$};
\fill [black] (27.48,-26.58) -- (26.6,-26.93) -- (27.41,-27.52);

\end{tikzpicture}
		
		\hspace{0.1cm} &
		%		\end{minipage}
		%		\begin{minipage}{.2\textwidth}

\begin{tikzpicture}[scale=0.09]

\tikzstyle{every node}+=[inner sep=0pt]
\draw [black] (28.8,-6.8) circle (3);
\draw [black] (28.8,-6.8) circle (2.4);

\draw (28.8,-6.8) node {$r_0$};
\draw [black] (28.8,-18.5) circle (2.4);
\draw [black] (28.8,-18.5) circle (3);
\draw (28.8,-18.5) node {$r_1$};
\draw [black] (39.4,-18.5) circle (3);
\draw (39.4,-18.5) node {$r_2$};
\draw [black] (39.4,-18.5) circle (2.4);
\draw [black] (28.8,-29.7) circle (3);
\draw (28.8,-29.7) node {$r_3$};
\draw [black] (19.6,-6.8) -- (25.8,-6.8);
\fill [black] (25.8,-6.8) -- (25,-6.3) -- (25,-7.3);
\draw [black] (28.8,-9.8) -- (28.8,-15.5);
\fill [black] (28.8,-15.5) -- (29.3,-14.7) -- (28.3,-14.7);
\draw (28.3,-12.65) node [left] {$\mathit{sync}$};
\draw [black] (30.657,-16.187) arc (126.42204:53.57796:5.799);
\fill [black] (37.54,-16.19) -- (37.2,-15.31) -- (36.6,-16.11);
\draw (34.1,-14.55) node [above] {$x\!=\!y$};
\draw [black] (42.256,-19.378) arc (100.63658:-187.36342:2.25);
\draw (44.84,-24.6) node [right] {$*$};
\fill [black] (40.44,-21.3) -- (40.1,-22.18) -- (41.08,-22);
\draw [black] (28.8,-21.5) -- (28.8,-26.7);
\fill [black] (28.8,-26.7) -- (29.3,-25.9) -- (28.3,-25.9);
\draw (28.3,-24.1) node [left] {$x\neq y$};
\draw [black] (30.679,-27.376) arc (168.77514:-119.22486:2.25);
\draw (35.65,-25.84) node [right] {$*$};
\fill [black] (31.79,-29.78) -- (32.47,-30.42) -- (32.67,-29.44);

\end{tikzpicture}

	%	\\
	%	\\
	%	$M_1$ & $\quad M2$& $\quad P$
		
	\end{tabular}
\caption{A system for which the weakest assumption is not regular.}
\label{fig:not_regular}%\vspace{-0.9cm}
\end{figure}

To cope with this difficulty, we change the focus of learning. Instead of learning the (possibly) non-regular language of $A_w$, we learn ${\cal T}(M_2)$, the set of accepted traces of $M_2$. This language is guaranteed to be regular, as it is represented by the automaton $M_2$.

\subsection{Soundness and Completeness of the Assume-Guarantee Rule for Communicating Systems}
\label{sec:soundAGR}
Since we have changed the goal of learning, we first show that in the setting of communicating systems, the assume-guarantee rule is sound and complete.

\begin{theorem}
For communicating programs, the Assume- Guarantee rule is sound and complete. That is, 
\begin{itemize}
    \item Soundness: 
    for every communicating program $A$ such that $\alpha A \subseteq \alpha M_2$,
    if $M_1 || A \vDash P$ and  $\T(M_2) \subseteq \T(A)$ then $ M_1 || M_2 \vDash P$.
    \item Completeness: If $M_1 || M_2 \vDash P$ then there exists an assumption $A$ such that $M_1 || A\vDash P$ and $\T(M_2) \subseteq \T(A)$. 
\end{itemize} 
\label{thrm:soundness}
\end{theorem}

\begin{proof}
\emph{Soundness}. Assume by way of contradiction that there exists an assumption $A$ such that $M_1||A \vDash P$ and $\T(M_2) \subseteq \T(A)$, but $M_1||M_2 \nvDash P$. 
Therefore, there exists an error trace $t\in (M_1 || M_2) \times P$. 
By Lemma~\ref{lemma:traces} and Lemma~\ref{lemma:traces12}, it holds that $t_2 = t\downarrow _{\alpha M_2}\in \T(M_2)$ and 
by Lemma~\ref{lemma:feasible_traces} and Lemma~\ref{lemma:feasible_traces12} it 
holds that $t_2$
is feasible. 
Since $\T(M_2) \subseteq \T(A)$, it holds that 
$t_2\in \T(A)$ and thus $t$
is an error trace in $(M_1 || A) \times P$, contradicting $M_1 || A \vDash P$. 

\emph{Completeness}. If $M_1||M_2 \vDash P$, then for $A = M_2$ it holds that $M_1 || A\vDash P$ and $\T(M_2) \subseteq \T(A)$.
\qed
\end{proof}

\subsection{Weakest Assumption - Special Cases} \label{sec:weak_reg}

As we have proven in Lemma~\ref{lemma:weak_non_reg}, in the context of communicating systems, the weakest assumption is not always regular, and so we changed our learning goal to the language of $M_2$, that is, to $\mathcal{T}(M_2)$. 
We now consider special cases for which the weakest assumption is guaranteed to be regular, and can therefore be used as a target for the learning process. This may result in the generation of a more general assumption.

We first show that if all components are constraints-free, then the weakest assumption is guaranteed to be regular. 
Intuitively, this is since when there are no constraints in the composed system (that includes also the specification), all traces are feasible. Thus we can reduce the problem of finding a weakest assumption for communicating programs to finding the weakest assumption for finite state automata.

\begin{lemma}\label{lemma:weakest_reg}
Let $P$ be a property and $M_1$ and $M_2$ be communicating programs such that $\alpha P$, $\alpha M_1$ and $\alpha M_2$ 
do not contain constraints. 
%Let $M_2$ be some communicating program. 
Then, the weakest assumption $A_w$ with respect to $M_1, M_2$ and $P$, is regular.
\end{lemma}

\begin{proof}
We construct a communicating system $A$ over $\alpha M_2$, based on $M_1 \times P$, as follows. A state in $A$ is accepting if its $M_1$ component is rejecting, or its $P$ component is accepting. 
We add self loops to all states, labeled by the alphabet of $M_2$ that does not synchronize with $M_1$ and $P$. We replace the alphabet of $M_1$ and $P$ that is not also in $M_2$ with $\varepsilon$-transitions (and leave the rest unchanged). Finally, we determinize the result. We therefore have that $t$ is an accepting trace of $A$ 
iff all the states that $A$ reaches when reading $t$ are either not accepting in $M_1$, or accepting (hence, not rejecting) in~$P$. For the complete details, see Appendix~\ref{app:traces}. 
Let $\aut{L}(A)$ be the set of all traces of $A$.

\textbf{Claim 1}
$\aut{L}(A) = \aut{L}(A_w)$.

Note that in case that there are no constraints, all traces are feasible, that is, $\aut{L}(A) = \aut{T}(A)$. To show the correctness of {Claim 1}, 
we state the following claim, which can be proved by induction on the trace $t$. 
See Appendix~\ref{app:traces} for the full proof. 

\textbf{Claim 2}
For a trace $t\in(\alpha M_2)^*$, the program $A$ reaches the same set of states $S$ when reading $t$, as the set of states that $M_1 \times P$ reaches given $t$.

{Claim $1$} then follows from the definition of the set of accepting states of $A$. 
Then, to conclude the proof, we have that $\mathcal{L}(A_w) = \mathcal{T}(A)$, and therefore is regular. 
\qed
\end{proof}

In~\cite{DBLP:conf/kbse/GiannakopoulouPB02} the authors prove that the weakest assumption is regular for the setting of LTSs, which are  prefix-closed finite-state automata. Their proof relies on the fact that LTSs are prefix closed, and they construct the weakest assumption accordingly. 
Our proof holds for general communicating programs, and not only for prefix-closed ones. 
We can therefore extend the result of~\cite{DBLP:conf/kbse/GiannakopoulouPB02} to the case of general finite-state automata. Corollary~\ref{lemma:weakest_reg_automata} follows from the fact that we can refer to finite-state automata as communicating programs without constraints; 
when applying the composition operators $||$ and $\times$, we relax the requirements of synchronization on read-write channels, to requiring synchronization on mutual alphabet letters. 

\begin{corollary}\label{lemma:weakest_reg_automata}
Let $\aut{A}_1$, $\aut{A}_2$ and $\aut{A}_p$ be finite-state automata
such that $\alpha \aut{A}_p \subseteq (\alpha \aut{A}_1 \cup \alpha \aut{A}_2)$. Then, the weakest assumption $A_w$ for $\aut{A}_1, \aut{A}_2$ and $\aut{A}_p$ is regular.   

\end{corollary}

We now consider a more general case for which we can find regular weakest assumptions. 
To this end, we define a refined notion of the weakest assumption, in which we do not care what the behavior of  $A_w$ is, for traces that are not feasible. We call it the \emph{semantic weakest assumption}. In the following, for a trace $t\in(\alpha M_2)^*$
we denote $S_t := \aut{T}((M_1 || t) \times P) $.

\begin{definition}[Semantic weakest assumption]
Let $P$ be a property and let $M_1$ and $M_2$ be two communicating programs. A semantic weakest assumption $A_s$ with respect to $M_1, M_2$ and $P$
has the following property:
\emph{
$\forall t\in (\alpha M_2)^*$ such that $S_t$ contains feasible traces, it holds that $t\in \aut{T}(A_s)$ iff $M_1 || t \vDash P$.  }

\end{definition}

That is, a semantic weakest assumption $A_s$ contains all words over the alphabet of $M_2$ that together with $M_1$ satisfy $P$ non-vacuously. 
 Note that there can be more than one such assumption, as we do not define its behavior for sets $S_t$ that include no infeasible traces. 

\begin{lemma}\label{lemma:one_set}
For a property $P$ and communicating programs $M_1$ and $M_2$ we have the following: 
    If for every $t\in (\alpha M_2)^*$ it holds that
either all traces in $S_t$ are feasible, or all of them are infeasible,
then
there exists a regular semantic weakest assumption. 
\end{lemma}

\begin{proof}
We show that the assumption $A$ from the proof of Lemma~\ref{lemma:weakest_reg} is such a semantic weakest assumption. We prove that under the terms of Lemma~\ref{lemma:one_set}, for every trace $t\in (\alpha M_2)^*$ such that $S_t$ contains feasible traces, it holds that $t\in \aut{T}(A)$ iff $M_1 || t \vDash P$.

Let $t\in \aut{T}(A)$. 
Then, according to the construction of $A$, and similar to the proof of Lemma~\ref{lemma:weakest_reg}, all paths in $(M_1 || t)\times P$ reach an accepting state, and thus $M_1 || t \vDash P$. 

For the other direction, let $t\in (\alpha M_2)^*$ such that $M_1 || t\vDash P$. According to the terms of the lemma, it holds that either all traces in $S_t$ are not feasible, or that all of them are feasible. 
If the former holds, then we have no requirement on~$t$. 

If the latter holds, then, since $M_1 || t\vDash P$, then all paths in $(M_1 || t) \times P$ reach an accepting state, and due to the construction of $A$, it holds that $t\in \aut{T}(A)$.  
\qed
\end{proof}

Note that
we restrict the traces in $S_t$, since
if $S_t$ contains both feasible and infeasible traces, then we can no longer guarantee that all paths of $(M_1||t)\times P$ reach an accepting state, even if $M_1 || t\vDash P$. 
For general communicating programs, for which this restriction does not hold, we can use the proof as in Lemma~\ref{lemma:weak_non_reg} to show that there are cases in which there is no regular semantic weakest assumption, that is, every semantic weakest assumption is not regular. 

As a special case of Lemma~\ref{lemma:one_set} we have the following. 
\begin{corollary} \label{cor:one_trace}
In case the behavior of $(M_1 || M_2) \times P $ is deterministic in the sense that for every trace $t\in (\alpha M_2) ^*$ there is only one corresponding trace $t' \in (M_1 || M_2) \times P $ such that $t' \downarrow_{\alpha M_2} = t$, then there exists a regular semantic weakest assumption. 
 \end{corollary}

\section{The Assume-Guarantee-Repair (AGR) Framework}\label{sec:AGR}

In this section we discuss our Assume-Guarantee-Repair (AGR) framework for communicating programs. The framework consists of a learning-based Assume-Guarantee algorithm, called $\agl$, and a \repair~ procedure, which are tightly joined.

Recall that the goal of \lstar\ in our case is to learn $\aut{T}(M_2)$, though we might terminate earlier if we find a suitable assumption for the AG rule: The nature of $\agl$ is such that the assumptions it learns before it reaches $M_2$ may contain the traces of $M_2$ and more, but still be represented by a smaller automaton. Therefore,
similarly to~\cite{DBLP:conf/tacas/CobleighGP03}, $\agl$ often terminates with an assumption $A$ that is much smaller than $M_2$. Indeed, our tool often produces very small assumptions (see Section~\ref{results}).

As mentioned before, not only do we
determine whether $M_1 || M_2 \vDash P$, but we also repair the program in case it violates the specification.
When $M_1 || M_2 \nvDash P$, the $\agl$ algorithm finds an error trace $t$ as a witness for the violation. In this case, we initiate the \repair~ procedure, which
eliminates $t$ from $M_2$. \repair~ applies abduction in order to learn a new constraint which, when added to $t$, makes the counterexample infeasible.%creates an infeasible trace.
\footnote{There are also cases in which we do not use abduction, as discussed in Section~\ref{sec:methods}.} The new constraint enriches the alphabet in a way which may eliminate additional counterexamples from $M_2$, by making them infeasible. 
%make similar traces infeasible as well. 
We elaborate on our use of abduction in Section~\ref{sec:abduction}. The removal of~$t$ and the addition of the new constraint result in a new goal $M_2'$ for $\agl$ to learn. Then, $\agl$ continues to search for a new assumption $A'$ that allows to verify $M_1 || M_2' \vDash P$. 

An important feature of our AGR algorithm is its \emph{incrementality}. When learning an assumption $A'$ for $M_2'$ we can use the 
information gathered in previous \mq s and \eq s, since the answer for them has not been changed (see Theorem~\ref{thm:incremental2} in  Section~\ref{sec:AGR_algorithm}). This allows the learning of $M_2'$ to start from the point where the previous learning has left off, resulting in a more efficient algorithm.

%membership queries previously asked for $M_2$, since the answer for them has not been changed.  As we show later (Theorem~\ref{thm:incremental2} in  Section~\ref{sec:AGR_algorithm}), the difference between the languages of $M_2$ and $M_2'$ lies in words (traces) whose membership has not yet been queried on $M_2$. 

As opposed to the case where $M_1 || M_2 \vDash P$, we cannot guarantee the termination of the repair process in case $M_1 || M_2 \nvDash P$.
This is because
we are only guaranteed to remove one (bad) trace and add one (infeasible) trace in every iteration (although in practice, every iteration may remove a larger set of traces). Thus, we may never converge to a repaired system. Nevertheless, in case of property violation, our algorithm always finds an error trace, thus a progress towards a ``less erroneous'' program is guaranteed. 

It should be noted that the $\agl$ part of our AGR algorithm deviates from the AG-rule of~\cite{DBLP:conf/tacas/CobleighGP03} in two important ways. First, since the goal of our learning is $\mathcal{T}(M_2)$ rather than $\mathcal{L}(A_w)$, our membership queries are different in type and order. Second, in order to identify real error traces and send them to \repair~ as early as possible, we add additional queries to the membership phase that reveal such traces. We then send them to \repair~ without ever passing through equivalence queries, which improves the overall efficiency.
Indeed, our experiments include several cases in which all repairs were invoked from the membership phase. In these cases, AGR ran an equivalence query only when 
it has already successfully repaired $M_2$, and terminated.
%the repaired component $M_2'$ already guaranteed $M_1 || M_2' \vDash P$, and terminated successfully.

{
\subsection{\mq s \& \eq s and the Implementation of the Teacher}\label{sec:queries}
As our algorithm heavily relies on the \lstar algorithm, we begin with a description of membership and equivalence queries, both from the side of the learner and the teacher. When using \lstar\ in verification, the implementation of the teacher plays a major role. We therefore describe the queries issued by the learner, and for each query, we elaborate on how exactly the query is answered by the teacher. 
The pseudo-code for the teacher is given in Algorithm~\ref{alg:queries}.

\subsubsection{Membership Queries}
%Algorithm~\ref{alg:AGR}, lines~\ref{membership}-\ref{endmembership}. 
As a \mq, the learner asks whether a given trace $t$ is in the language of some suitable assumption $A$. The order of \mq s is as in the \lstar\ algorithm, that is, we traverse traces in alphabetical order and increasing length. 
%Given a trace $t$, the teacher answers ``is $t\in U$?" as follows. 
The teacher answers the \mq~as follows. 
\begin{itemize}
    \item[--] If $t\notin \aut{T}(M_2)$, answer \ansno\ . 
    \item[--] If $t\in \aut{T}(M_2)$, check if $t$ is an error trace, and if so, turn directly to repair. That is - 
    \begin{itemize}
        \item If  $M_1 || t \nvDash P$, pause learning and turn to repair.
        \item If $M_1 || t \vDash P$, answer the \mq~ with \ansyes\ . 
    \end{itemize}
\end{itemize}
Note that not every membership query is answered immediately, since the learning process may pause for repairing the system. 
If a repair was issued on a trace $t$, then after repair, the teacher answers \ansno~ on the \mq~on $t$, but it may provide additional information as we describe in Section~\ref{sec:AGR_algorithm}.

\subsubsection{Equivalence Queries} %Algorithm~\ref{alg:AGR} lines~\ref{equivalence}-\ref{endequiv}. 
As an \eq, the learner asks, given assumption $A$, whether $A$ is a suitable assumption to verify the correctness of the system. 
An \eq~is issued once the learner was able to construct an automaton that is consistent with all previous \mq s, as done in \lstar algorithm. 
The teacher answers an \eq~as follows.
%Given a candidate $A$, the teacher answers the \eq~ ``$A\equiv U$" as follows. 
\begin{itemize}
    \item[--] If $M_1 || A \vDash P$, then $A$ is a suitable candidate according to the AG-rule. Then, we check if the second condition of the AG-rule holds, that is, if $\aut{T}(M_2) \subseteq \aut{T}(A)$, and answer accordingly.
    \item[--] If $M_1 || A \nvDash P$, check if the error trace $t$ that violates $P$ is also a trace of $M_2$. If so, pause learning and turn to repair. If $t$ is not a real error trace, return  \ansno~ to the \eq~, along with a trace to be eliminated from $A$. 
\end{itemize}
Note that, as with \mq s, not every \eq~is immediately answered. In case of a violation, the algorithm first repairs the system, and only then the teacher returns \ansno~ to the \eq~(as the queried $A$ contained a real error trace). Also as in \mq s, the repair stage might provide the learner with additional information (Section~\ref{sec:AGR_algorithm}). 
}%end revision changes

Algorithm~\ref{alg:queries} below presents the pseudo-code of the teacher when answering the two types of queries. While $M_1$ and $P$ remain unchanged,  $M_2$ is iteratively repaired (if needed). The teacher therefore considers at iteration $i$, the repaired program $M_2^i$.

\begin{algorithm}[]
	\begin{algorithmic}[1]

	\Function{\mq-oracle }{trace $t_2$}
	%\While{\lstar learner did not converge to a candidate assumption}
	%	\State   let $t_2\in (\alpha M_2^i)^*$
		\If{$t_2 \in \mathcal{T}(M_2^i)$} \label{line:ifin}
		\If{$M_1 || t_2 \nvDash P$} \label{terminate_cex} 
		\State  let $t \in (M_1 || t_2) \times P $ be an error trace \Comment{$t$ is a cex\\} \Comment{proving $M_1 || M_2^i \nvDash P$}\label{error1}
		%\State REPAIR($M_2^i, t$)\label{repair1}
		\State\Return {\ansrep~ - $t$} \label{line:repairmq}
		\Else {~~}
		\Return{\ansyes\ } \Comment{$M_1 || t_2 \vDash P$}
	%	return to $\agl$ in Line~\ref{membership} with  $t_2\in \mathcal{T}(A_j^i)$ 
	\label{backL1}  
		\EndIf
		\Else {~~}
		 \Return{\ansno\ } \Comment{$t_2 \notin \mathcal{T}(M_2^i)$}
	%	return to $\agl$ in Line~\ref{membership} with $t_2\notin \mathcal{T}(A_j^i)$  
	\label{endmembership}
	    \EndIf
%	\EndWhile
	\EndFunction
	
	\Function{\eq-oracle }{candidate assumption $A^i_j$}
%	\State  let $A_j^i$ be the candidate assumption generated by the learner  
	    \If{$M_1 || A_j^i \vDash P$} \label{equivcheck1}
	    \If{$\mathcal{T} (M_2^i ) \subseteq \mathcal{T} (A_j^i)$}\label{equivcheck2}
	   \Return{\ansyes\ }
	    %terminate and return $M_1 || M_2^i \vDash~P$ 
	 
	    \Else 
	    \State  let $t_2 \in \mathcal{T} (M_2^i ) \setminus \mathcal{T} (A_j^i)$  
	    %\State set $j:= j+1$
	    \State \Return{\ansno~ + $t_2$} 
	    %return to $\agl$ in Line~\ref{membership} with $t_2\in \mathcal{T}(A_j^i)$  
	    %Return $t$ as a counter example to the equivalence query, requiring $t \in A$, and go back to~\ref{membership}  
	    \label{L*_cex_inA}
        \EndIf
        \Else \Comment{$M_1 || A_j^i \nvDash P$} \label{line:else}
        \State let $t \in (M_1 || A_j^i) \times P$ be an error trace \label{line:erroreq}
        \State denote $t = (t_1 || t_A )\times t_P$  \label{error2}
        \If{$t_A\in \mathcal{T} (M_2^i)$} \label{error3}
        %\State REPAIR($M_2^i, t_A$) \Comment{$t_A$ is a cex proving $M_1 || M_2^i \nvDash P$} \label{terminate_cex_eq}\label{repair2}
        \Return {\ansrep~ - $t$}
        \Else {~~} \Return{\ansno~ - $t_A$}
        %\State set $j:= j+1 $ 
        %\State  return to $\agl$ in Line~\ref{membership} with $t_A\notin \mathcal{T}(A_{j}^i)$ 
        \label{endequiv} 
        %Return to~\ref{membership} with $t_A$ as a counterexample to the equivalence query, requiring $t_A \notin A$   
     
        \EndIf
        \EndIf
	\EndFunction

	\end{algorithmic}

	\caption{The $\agl$ Teacher }\label{alg:queries}
\end{algorithm}

\begin{algorithm*}[]
	\begin{algorithmic}[1]
\Require{$M_1, M_2, P$}
\Ensure{A repaired $M_2^i$ (if needed) and an assumption $A^i$ that proofs that $M_1 || M_2^i \vDash P$}
	\State set $i=0$, $j=0$, $M_2^i = M_2$
	\Function{$\agl$}{}

	\While{true} \Comment{continue running until repairing $M_2$, might not terminate}
		\While{\lstar learner did not converge to a candidate assumption} \Comment{issue \mq s}	 \label{membership}
		\State   let $t_2\in (\alpha M_2^i)^*$ chosen according to the \lstar learner
		\If{\Call{\mq-oracle }{$t_2$} = \ansrep~ - $t$} 
		\State  $t_2, t_2'
		 := $ \Call{repair}{$M_2^i, t$} \label{repair1}
		\State add -$t_2$ to the \lstar~table and in case of abduction, add  +$t_2'$ to the \lstar~table  \label{endmembership}
		\Else{ add $t_2$ to the \lstar\ table according to the answer of the \mq }
		\EndIf
		\label{endmemphase}
	    \EndWhile
	   % \State // equivalence Queries\label{equivalence}
	    \State  let $A_j^i$ be the candidate assumption generated by the \lstar~learner  \Comment{issue \eq\ } 	    \label{equivalence}
	    \If{ \Call{\eq-oracle }{$A_j^i$} = \ansyes\ } 
	    \label{line:ifeq}
	    \State \Return {$M_1 || M_2^i \vDash~P$ together with the assumption $A_j^i$}   

	    \label{terminate_proof} 
	    \Else \If{\Call{\eq-oracle }{$A_j^i$} = \ansrep~ - t}
	    	\State   $t_2, t_2'
		 := $ \Call{repair}{$M_2^i, t$} \label{repair2}
		\State add -$t_2$ to the \lstar~table and in case of abduction, add  +$t_2'$ to the \lstar~table \label{line:repaireq}
	     \Else \Comment{\Call{\eq-oracle }{$A_j^i$} = \ansno~ \textpm~t}
	    \State let $t_2$ be the cex from the \eq
	    \State add $t_2$ to the \lstar table according to the answer of the \eq \Comment{after issuing an \eq, the algorithm returns to the \mq s phase}
	    \State set $j:=j+1$ \label{endequivalence}
	    \EndIf
	    \EndIf
	    \EndWhile
        \EndFunction
		\Function{repair}{$M_2^i, t$}
		\State  let $t_1 \in M_1, t_2\in M_2^i, t_p \in P$ such that $t = (t_1 || t_2 )\times t_p$  
		\If{$t$ does not contain constraints}
		\State set $M_2^{i+1} := \T(M_2^i)\setminus \{t_2 \}$ \label{line:repremovesyn}
		\State set $i:= i+1, j:=0$ \Comment{$i$ is set to $0$, start a new iteration of learning $M_2^{i+1}$} \label{line:newit1}
		\State\Return{ - $t_2$}
		
		%\State  return to $\agl$ in Line~\ref{membership} with $M_2^{i+1} = \T(M_2^i)\setminus \{t_2 \}$ \phantom{-------------------------------------------}and $t_2\notin \mathcal{T}(A_0^{i+1})$   
		%return to~\ref{membership} with $M_2' = \T(M_2)\setminus \{t \}$,  requiring $t\notin \T(A)$  
		\label{repair:syntax}
		\Else \Comment{$t$ contains constraints, use abduction to eliminate $t$  }\label{repair:abduction}
		%\State use \emph{abduction} to eliminate $t$  
		\State  let $c$ be the new constraint learned during abduction and let $t_2' = t_2 \cdot c$
		\State update $\alpha M_2^{i+1} := \alpha M_2^{i} \cup \{ c \}$ 
		\label{line:up_alpha}
		\State set $M_2^{i+1} := (\T(M_2^i) \setminus \{t_2 \})\cup \{t_2' \}$
		\label{up_alpha}
		\State set $i:=i+1, j:=0$ \Comment{$i$ is set to $0$, start a new iteration of learning $M_2^{i+1}$} \label{line:newit2}
		%\State  let $t_2' = t_2 \cdot c$  % be the output of the abduction
		\State \Return{ - $t_2$, + $t_2'$}
		%\State  return to $\agl$ in Line~\ref{membership} with\phantom{------------------------} \phantom{------------------------------------} $M_2^{i+1} = (\T(M_2^i) \setminus \{t_2 \})\cup \{t_2' \}$
		%\phantom{-----------------------------------} and $t_2 \not\in \T(A_0^{i+1}), t_2'\in \T(A_0^{i+1})$
		\EndIf
		
		\EndFunction
	\end{algorithmic}

	\caption{AGR }\label{alg:AGR}
\end{algorithm*}
 
\subsection{The Assume-Guarantee-Repair (AGR) Algorithm }\label{sec:AGR_algorithm}

We now describe our AGR algorithm in more detail (see Algorithm~\ref{alg:AGR}).
Figure~\ref{fig:agr} describes the flow of the algorithm. AGR comprises two main parts, namely $\agl$ and \repair . 

The input to AGR are the components $M_1$ and $M_2$, and the property $P$. While
$M_1$ and $P$ stay unchanged during AGR, $M_2$ keeps being updated as long as the algorithm recognizes that it needs repair.% (we can bound the number of iterations, as we discuss in Section~\ref{Sec:termination}). 

The algorithm works in iterations, where in every iteration the next updated $M_2^i$ is calculated, starting with iteration $i = 0$, where $M_2^0 = M_2$. An iteration starts with the membership phase in line~\ref{membership} of Algorithm~\ref{alg:AGR}, and ends either when $\agl$ successfully terminates (Alg.~\ref{alg:AGR} line~\ref{terminate_proof}) or when procedure \repair~ is called (Alg.~\ref{alg:AGR} lines~\ref{repair1} and~\ref{repair2}). 
When a new system $M_2^i$ is constructed, %the execution of $\agl$ continues from the point it has left off in the previous iteration, and 
$\agl$ does not start from scratch. The information that has been learned in previous iterations is still valid for $M_2^i$. The new iteration is given additional new trace(s) that have been added or removed from the previous $M_2^i$ (Alg.~\ref{alg:AGR} lines~\ref{endmembership},~\ref{line:repaireq}). %(lines~\ref{backL1},\ref{endmembership},\ref{L*_cex_inA}, \ref{L*_cex_notinA}).

$\agl$ consists of two phases: membership, and equivalence. In the membership phase (lines~\ref{membership}-\ref{endmemphase} of Alg.~\ref{alg:AGR}), the algorithm issues \mq s as calls to the function \mqor\ of Algorithm~\ref{alg:queries}. If, during the membership phase, we encounter a trace $t_2\in M_2^i$ that in parallel with $M_1$ does not satisfy $P$, then $t_2$ is a bad behavior of $M_2$, and \repair~is invoked. To this end we enhance the answers of the teacher not only to \ansyes~ and \ansno\ , but also to \ansrep\ . This holds also for \eq s.

Once the learner reaches a candidate assumption $A_j^i$, it issues an \eq\ (Alg.~\ref{alg:AGR} lines~\ref{equivalence}-\ref{endequivalence}). 
$A_j^i$ is a suitable assumption if both $M_1 ||A_j^i \vDash P$ and   $\mathcal{T}(M_2^i) \subseteq \mathcal{T}(A_j^i)$ hold. 
In this case, AGR terminates and returns $M_2^i$ as a successful repair of $M_2$.
If $M_1 ||A_j^i \nvDash P$, then a counterexample $t$ is returned, that is composed of bad traces in $M_1, A_j^i$, and $P$. If the bad trace $t_2$, the restriction of $t$ to the alphabet of $A_j^i$, is also in $M_2^i$, then $t_2$ is a bad behavior of $M_2^i$, and here too \repair~is invoked. Otherwise, AGR updated the \lstar~table, returns to the membership phase, and continues to learn $A_j^i$.

As we have described, \repair~is called when a bad trace $t$ is found in $(M_1|| M_2^i)\times P$ and should be removed. If $t$ contains no constraints then its sequence of actions is illegal and its restriction $t_2\in \aut{T}(M_2^i)$ should be removed from $M_2^i$. In this case, \repair~returns to $\agl$ and updates the learning goal to be  ${\cal T}(M_2^{i+1}) := {\cal T}(M_2^i) \setminus \{t_2\}$, along with the answer ``- $t_2$" that indicates that $t_2$ should not be a part of the learned assumption. 
In Section~\ref{sec:methods} we discuss different methods for removing $t_2$ from~$M_2^i$.\footnote{For the different methods for removing $t_2$, we actually end up with a learning goal ${\cal T}(M_2^{i+1}) \subseteq {\cal T}(M_2^i) \setminus \{t_2\}$ and not necessarily ${\cal T}(M_2^i) \setminus \{t_2\}$ itself. We discuss this further in Section~\ref{sec:methods}.}

The more interesting case is when $t$ contains constraints. In this case, we not only remove the matching $t_2$ from $M_2^i$, but also add a new constraint $c$ to the alphabet, which causes $t_2$ to be infeasible. This way we eliminate $t_2$, and may also eliminate a family of bad traces that violate the property in the same manner -- 
adding a new constraint can only cause the removal of additional error traces, and cannot add traces to the system. 
We deduce $c$ using abduction, as we describe in  Section~\ref{sec:abduction}. As before, \repair~returns to $\agl$ 
with a new goal to be learned, but now also with an extended alphabet. 
In addition, we are provided with information about two traces: 
$t_2$ that should \emph{not} be included in the new assumption, and $(t_2 \cdot c)$ that should be included.

\subsubsection{Incremental learning}\label{AGR_sec:incremental}

One of the advantages of AGR is that it is {\em incremental}, in the sense that answers to membership queries from previous iterations remain unchanged for the repaired system. 
Formally, we have the following.

\begin{theorem}\label{thm:incremental2}
Assume that $T^i$ is the \lstar~ table (see Section~\ref{sec:l*}) at iteration $i$, and let $M_2^{i+1}$ be the repaired component after that iteration. Then, $T^i$ is consistent with $M_2^{i+1}$. 
\end{theorem}

From Theorem~\ref{thm:incremental2}
it follows that, in particular, $T^i$ is consistent with (1) traces that are removed between $M_2^i$ and $M_2^{i+1}$; and (2) traces that are learned using abduction and added to $M_2^{i+1}$.

\begin{proof}
Traces are added to the \lstar~ table in three scenarios: during \mq s; as counterexamples to \eq s; and while repairing the system, removing error traces and adding traces learned by abduction. The difference between $M_2^i$ and $M_2^{i+1}$ is only due to the repair part, that is, in error traces that are removed and traces that were added due to abduction. The two components agree on all other traces, and so $T^i$ is consistent with all traces in $M_2^{i+1}$ that are not part of  \repair . 

Let $t_2 \in \mathcal{T}(M_2^i)\setminus\mathcal{T}(M_2^{i+1})$
be an error trace that was removed during \repair . We consider the two cases - in which \repair\ is invoked during a \mq , or during an \eq .

In the former case, this is the first time the \lstar~learner queries $t_2$, as the \lstar~learner does not issue \mq s on traces that are already in the table. Therefore, $t_2$ is not part of $T^i$ yet, and there is no inconsistency. 

If $t_2$ is an error trace that is found during an \eq, then  $M_1 || t_2 \nvDash P$ (Alg.~\ref{alg:queries}, line~\ref{line:erroreq}). Since whether $M_1 || t_2 \nvDash P$ only depends on $M_1, P$ and the specific trace $t_2$, and not on $M_2^i$, and since $M_1$ and $P$ remain constant, it holds that $t_2$ was an error trace also in previous iterations. Therefore, $t_2$ cannot appear in $T^i$ as a positive trace. This concludes the case of traces that are removed from $M_2^i$ during repair. 

Now, let $t_2'$ be a positive trace added to $M_2^{i+1}$ using abduction, that is $t_2' = t_2\cdot c$ for $t_2$ that was removed by \repair . If the constraint $c$ is a new alphabet letter, then in particular $t_2'$ is not over the alphabet of $M_2^i$ and cannot appear in $T^i$. 

We now consider the case in which the added constraint $c$ is already a part of $\alpha M^i_2$. 
If the error trace $t_2$ was found during a \mq , then, since the \lstar~learner issues queries in an increasing-length order, it holds that $t_2'$ was not queried in a \mq\ in previous iterations as $|t_2|<|t_2'|$. Moreover, it cannot be the case that $t_2'$ was a negative counterexample to an \eq\ since these are only derived from error traces (Alg.\ref{alg:queries}, lines~\ref{line:erroreq}-\ref{endequiv}), and $t_2'$ is not an error trace. The same argument holds also if $t_2$ was found during an \eq . 

We are left to show that 
if $t_2$ was found during an \eq , it cannot be the case that $t_2'$ was previously queried during a \mq .
In this case it holds that the \lstar~learner did not issue a \mq\ on $t_2$. Since the \lstar~learner issues queries in increasing length, and since $|t_2| < |t_2'|$, it holds that $t_2'$ was not previously queried. 
\qed
\end{proof}

\begin{figure*}
    \centering
    \includegraphics[width=1\textwidth]{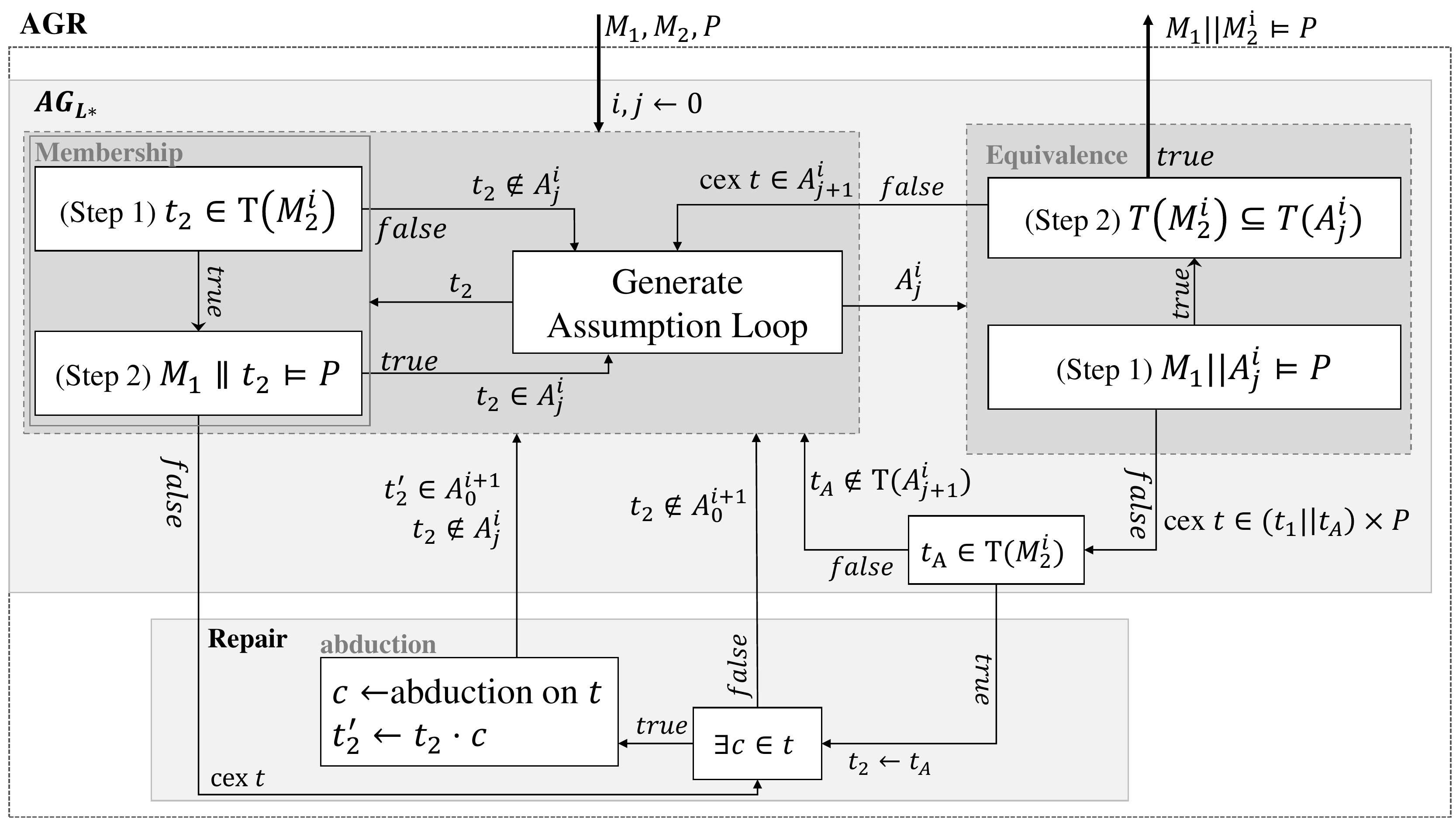}
    %[scale=0.4]
    \caption{The flow of AGR}
    \label{fig:agr}
\end{figure*}

\subsection{Semantic Repair by Abduction} \label{sec:abduction}
We now describe our repair of $M_2^i$, in case the error trace $t$ contains constraints (Alg.~\ref{alg:AGR} line~\ref{repair:abduction}). Error traces with no constraints are removed from $M_2^i$ syntactically (Alg.~\ref{alg:AGR} lines\ref{line:repremovesyn} -\ref{repair:syntax}), while in abduction we \emph{semantically} eliminate $t$ by making it infeasible. 
The new constraint is then
added to the alphabet of $M_2^i$ and may eliminate additional error traces.
Note that the constraints added by abduction can only restrict the behavior of $M_2$, making more traces infeasible. Therefore, we do not add counterexamples to $M_2$.

The process of inferring new constraints from known facts about the program is called \emph{abduction}~\cite{DBLP:conf/cav/DilligD13}. 
We now describe how we apply it. 
Given a trace $t$, let $\varphi_t$ be the first-order formula (a conjunction of constraints), which constitutes the SSA representation of $t$~\cite{DBLP:conf/popl/AlpernWZ88}.
%
%We first construct a formula $\varphi_t$ by transforming each of the alphabet letters along $t$ into a constraint, using SSA representation~\cite{DBLP:conf/popl/AlpernWZ88}. 
In order to make $t$ infeasible, we look for a formula $\psi$ such that $\psi \wedge \varphi_t \rightarrow \mathit{false}$.\footnote{Usually, in abduction, we look for $\psi$ such that $\psi\wedge \varphi_t$ is not a contradiction. In our case, however, since $\varphi_t$ is a violation of the specification, we want to infer a formula that makes $\varphi_t$ unsatisfiable.}

\begin{example} \label{ex:abduction}
Consider the component $M_2$ of Figure~\ref{model} and the component $M_1$ and specification $P$ of Figure~\ref{fig:two_systems_spec} from Section~\ref{Intro}.
The following $t$ is an error trace in $(M_1 || M_2) \times P $: 
\begin{align*}
 & t=   (\mathit{read}? x_{pw}, 999 < x_{pw}, (\mathit{enc}? y_{pw}, \mathit{enc}! x_{pw}), \\
& y_{pw} = 2\cdot y_{pw}, (\mathit{getEnc}! y_{pw}, \mathit{getEnc}? x_{pw2}), x_{pw}\neq x_{pw2}, \\ & y_{pw} \geq 2^{64}
)
\end{align*}

due to the execution 
\begin{align*}
 &   (\mathit{read}? 2^{63}, 999 < 2^{63}, (\mathit{enc}? 2^{63}, \mathit{enc}! 2^{63}),
y_{pw} = 2\cdot 2^{63},\\ & (\mathit{getEnc}! 2^{64}, \mathit{getEnc}? 2^{64}), 2^{63}\neq2^{64}, 2^{64} \geq 2^{64}
).
\end{align*}

We then look for a constraint $\psi$ that will make the sequence $t$, and in particular the violation of the last constraint $y_{pw}\geq 2^{64}$, infeasible. 
\end{example}

Note that $t \in \T(M_1 || M_2^i) \times P$, and so it includes variables both from $X_1$, the set of variables of $M_1$, and from $X_2$, the set of variables of $M_2^i$. Since we wish to repair $M_2^i$, the learned $\psi$ is over the variables of $X_2$ only. In Example~\ref{ex:abduction}, this corresponds to learning a formula over $x_{pw}$.

The formula $\psi \wedge \varphi_t \rightarrow \mathit{false}$ is equivalent to $\psi \rightarrow (\varphi_t \rightarrow \mathit{false})$.
Then, $\psi = \forall x\in X_1: (\varphi_t \rightarrow \mathit{false})\  = \ \forall x\in X_1 (\neg \varphi_t)$, is such a desired constraint: $\psi$ makes $t$ infeasible and is defined only over the variables of $X_2$. 
We now use quantifier elimination~\cite{quantifier} to produce a quantifier-free formula over $X_2$.
Computing $\psi$ is similar to the abduction suggested in~\cite{DBLP:conf/cav/DilligD13}, but the focus here is on finding a formula over $X_2$ rather than over any minimal set of variables as in~\cite{DBLP:conf/cav/DilligD13}; in addition, in~\cite{DBLP:conf/cav/DilligD13} they look for $\psi$ such that $\varphi_t \wedge \psi$ is not a contradiction, while we specifically look for $\psi$ that blocks $\varphi_t$. We use Z3~\cite{de2008z3} to apply quantifier elimination and to generate the new constraint. 

\begin{example} \label{ex:abduction2}
For $t$ of Example~\ref{ex:abduction}, the process described above results in the constraint $\psi = x_{pw} < 2^{63}$. Note that while $\psi$ blocks the erroneous behaviour of $t$, it allows all executions of $t$ in which $x_{pw}$ is assigned with smaller values than $2^{63}$. In addition, it does not only block the one execution in which $x_{pw} = 2^{63}$, but the set of all erroneous executions in $(M_1 || M_2 ) \times P$ of the example. 
\end{example}

After generating $\psi(X_2)$, we add it to the alphabet of $M_2^i$ (line~\ref{line:up_alpha} of Algorithm~\ref{alg:AGR}). In addition, we produce a new trace $t_2' = t_2 \cdot \psi(X_2)$.
The trace $t_2'$ is returned as the output of the abduction.

We now turn to prove that by making $t_2$ infeasible, we eliminate the error trace $t$.

%which will replace $t_2$ in $M_2^i$.  
\begin{lemma}\label{lemma:AGR_abduction_feas}
Let $t = (t_1 || t_2) \times t_P$. If $t_2$ is infeasible, then $t$ is infeasible as well. 
\end{lemma}
%Recall that $t$ is formed from $(t_1 || t_2) \times t_P$.  By concatenating $\psi(X_2)$ to $t_2$ and making it infeasible we in fact made $t$ infeasible as well. 
\begin{proof}
This is due to the fact that $t_P$ can only restrict the behaviors of $t_1$ and $t_2$, thus if $t_2$ is infeasible, $t$ cannot be made feasible. 
Formally, Lemma~\ref{lemma:AGR_abduction_feas} follows from Lemma~\ref{lemma:feasible_traces} and Lemma~\ref{lemma:feasible_traces12} given in Section~\ref{sec:tarces}. 
By Lemma~\ref{lemma:feasible_traces}, if $t = (t_1 || t_2) \times t_P$ is feasible, then $t_1 || t_2$ is a feasible trace of $M_1 || M_2$. By Lemma~\ref{lemma:feasible_traces12}, if $t_1 || t_2$ is feasible,  then $t_2$ is feasible as well. Therefore, if $t_2$ is infeasible, then $t$ is infeasible, proving Lemma~\ref{lemma:AGR_abduction_feas}. 
\qed
%We now turn to prove Lemma~\ref{lemma:feasible_traces}. 
\end{proof}

\begin{figure}[]
    \centering
\begin{tikzpicture}[scale=0.12]
\tikzstyle{every node}+=[inner sep=0pt]
\draw [black] (28.4,-21.7) circle (3);
\draw (28.4,-21.7) node {$q$};
\draw [black] (43.9,-21.7) circle (3);
\draw (43.9,-21.7) node {$q'$};
\draw [black] (43.9,-21.7) circle (2.4);
%\draw [black] (14.3,-21.7) circle (3);
%\draw [black] (16.6,-11) circle (3);
\draw [black] (31.4,-21.7) -- (40.9,-21.7);
\fill [black] (40.9,-21.7) -- (40.1,-21.2) -- (40.1,-22.2);
\draw (36.15,-22.2) node [below] {$\psi(X_2)$};
\draw [black, dashed] (17.3,-21.7) -- (25.4,-21.7);
\fill [black] (25.4,-21.7) -- (24.6,-21.2) -- (24.6,-22.2);
\draw (21.35,-22.2) node [below] {$t_2$};
\draw [black, dashed] (18.82,-13.02) -- (26.18,-19.68);
\fill [black] (26.18,-19.68) -- (25.92,-18.78) -- (25.25,-19.52);
\draw (23.47,-15.86) node [above] {$t'$};
\end{tikzpicture}

\caption{Adding the constraint $\psi(X_2)$ to block the error trace~$t_2$.
Note that $\psi$ is also added to traces other than $t_2$, for example to $t'$. Then, $\psi(X_2)$ blocks assignments of $t'$ that violate $P$ in the same way as $t_2$, but it allows for other assignments of $t'$ to hold.
}\label{fig:adding_constraint}
\end{figure}
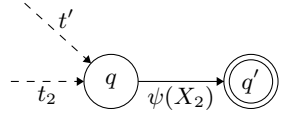

In order to add $t_2\cdot \psi(X_2)$ to $M_2^i$ while removing $t_2$, we split the state $q$ that $t_2$ reaches in $M_2^i$ into two states $q$ and $q'$, and add a transition labeled $\psi(X_2)$ from $q$ to $q'$, where only $q'$ is now accepting (see Figure~\ref{fig:adding_constraint}). 
Thus, we eliminate the violating trace from $M_1 || M_2^i$.
\repair~ now returns to $\agl$ 
with the negative example -$t_2$ and the positive example +$t_2'$ to add to the \lstar~ table $T^{i+1}$ of the next iteration, 
in order to learn an assumption for the repaired component $M_2^{i+1}$ (which includes $t_2'$ but not $t_2$). %Figure~\ref{fig:adding_constraint} demonstrates this construction. 

\begin{figure}[t]

\begin{center}
\begin{tikzpicture}[scale=0.12]
\tikzstyle{every node}+=[inner sep=0pt]
\draw [black] (20,-16.6) circle (3);
\draw (20,-16.6) node {$q_0$};
\draw [black] (20,-16.6) circle (2.4);
\draw [black] (26.6,-7.8) circle (3);
\draw (26.6,-7.8) node {$q_5$};
\draw [black] (33.9,-16.6) circle (3);
\draw (33.9,-16.6) node {$q_4$};
\draw [black] (33.9,-28.3) circle (3);
\draw (33.9,-28.3) node {$q_3$};
\draw [black] (8.9,-34.5) circle (3);
\draw (8.9,-34.5) node {$q_2$};
\draw [black] (20,-28.3) circle (3);
\draw (20,-28.3) node {$q_1$};
\draw [black] (18.517,-25.707) arc (-159.47742:-200.52258:9.291);
\fill [black] (18.52,-25.71) -- (18.7,-24.78) -- (17.77,-25.13);
\draw (17.43,-22.45) node [left] {${\scriptstyle \mathit{read}?x_{pw}}$};
\draw [black] (9.476,-31.578) arc (-203.04103:-278.58713:7.231);
\fill [black] (9.48,-31.58) -- (10.25,-31.04) -- (9.33,-30.65);
\draw (11.66,-27.59) node [above] {${\scriptstyle x_{pw}\leq 999}$};
\draw [black] (19.751,-31.265) arc (-17.39143:-104.23674:6.828);
\fill [black] (19.75,-31.27) -- (19.03,-31.88) -- (19.99,-32.18);
\draw (17.56,-35.69) node [below] {${\scriptstyle \mathit{read}?x_{pw}}$};
\draw [black] (31.138,-29.456) arc (-73.2127:-106.7873:14.499);
\fill [black] (31.14,-29.46) -- (30.23,-29.21) -- (30.52,-30.17);
\draw (26.95,-30.57) node [below] {${\scriptstyle 999<x_{pw}}$};
\draw [black] (35.572,-19.071) arc (23.81791:-23.81791:8.367);
\fill [black] (35.57,-19.07) -- (35.44,-20) -- (36.35,-19.6);
\draw (36.79,-22.45) node [right] {${\scriptstyle \mathit{enc}!x_{pw}}$};
\draw [black] (29.556,-7.515) arc (81.37918:-2.02468:6.086);
\fill [black] (29.56,-7.52) -- (30.27,-8.13) -- (30.42,-7.14);
\draw (33.88,-8.21) node [right] {${\scriptstyle \mathit{getEnc}?x_{pw2}}$};
\draw [black] (19.342,-13.704) arc (-181.20588:-252.53392:6.136);
\fill [black] (19.34,-13.7) -- (19.86,-12.91) -- (18.86,-12.89);
\draw (19.99,-8.75) node [left] {${\scriptstyle x_{pw} < 2^{63}}$};
\end{tikzpicture}
\end{center}

    \caption{The repaired component $M_2^1$.}
    \label{fig:repairabd}
\end{figure}
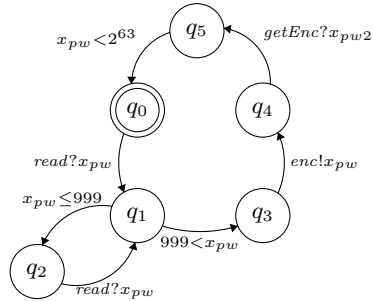

\begin{example}\label{ex:repairandassumption}
Figure~\ref{fig:repairabd} presents the repaired component $M_2^1$ we generate given $M_2$ of Figure~\ref{model} and $M_1$ and $P$ of Figure~\ref{fig:two_systems_spec}. As there is only one error trace (that induces many erroneous executions), the repaired component is achieved after one iteration. The new constraint, $\psi = x_{pw} < 2^{63}$, is added at the end of the trace $t_2 = t\downarrow_{\alpha M_2^i}$ for $t$ of Example~\ref{ex:abduction}. Intuitively, this constraint is equivalent to adding an \emph{assume} statement in the program. 
\end{example}

\subsection{Syntactic Removal of Error Traces} \label{sec:methods}
Recall that the goal of \repair~ is to remove a bad trace $t_2$ from $M_2$ once it is found by $\agl$. If $t_2$ contains constraints, we remove it by using abduction as described in Section~\ref{sec:abduction}. Otherwise, we can remove $t_2$ by constructing a system whose language is ${\cal T}(M_2)\setminus\{t_2\}$. We call this the {\em exact} method for repair. 
However, removing a single trace at a time may lead to slow convergence, and to an exponential blow-up in the sizes of the repaired systems. Moreover, as we have discussed, in some cases there are infinitely many such traces, in which case AGR may never terminate. 

For faster convergence, we have implemented two additional heuristics, namely {\em approximate} and {\em aggressive}. These heuristics may remove more than a single trace at a time, while keeping the size of the systems small. While ``good'' traces may be removed as well, the correctness of the repair is maintained, since no bad traces are added. Moreover, an error trace is likely to be in an erroneous part of the system, and in these cases our heuristics manage to remove a set of error traces in a single step.

We survey the three methods.

\begin{itemize}
\item\emph{Exact.} To eliminate only $t_2$ from $M_2$, we construct the program $M_{t_2}$ that accepts only $t_2$, and complement it to construct $\Bar{M_{t_2}}$ that accepts all traces except for $t_2$. Finally, we intersect $\Bar{M_{t_2}}$ with $M_2$. This way we only eliminate $t_2$, and not other (possibly good) traces. On the other hand, this method converges slowly in case there are many error traces, or does not converge at all if there are infinitely many error traces.

\item\emph{Approximate.} 
Similarly to our repair via abduction in Section~\ref{sec:abduction}, we prevent the last transition that $t_2$ takes from reaching an accepting state. Let $q$ be the state that $M_2$ reaches when reading 
$t_2$. We mark $q$ as a non-accepting state, and add an accepting state $q'$, to which all in-going transitions to $q$ are diverted, except for the last transition on $t_2$. 
This way, some traces that lead to $q$ are preserved by reaching $q'$ instead, and the traces that share the last transition of $t_2$ are eliminated along with $t_2$. As we have argued, these transitions may also be erroneous.

\item\emph{Aggressive.} 
In this simple method, we remove $q$, the state that $M_2$ reaches when reading $t_2$, from the set of accepting states.
This way we eliminate $t_2$ along with all other traces that lead to $q$. In case that every accepting state is reached by some error trace, this repair might result in an empty language, creating a trivial repair. However, our experiments show that in most cases, this method quickly leads to a non-trivial repair. 
\end{itemize}

\subsubsection{Towards convergence of syntactic repair}\label{sec:syntactic_convg}

As discussed above, the exact repair may not terminate in case of infinitely many traces that introduce the same error. Indeed, we now claim that when provided with long-enough counterexamples, we can conclude that the exact repair will not converge, and justifiably turn to the other repair types. 

In the following, we claim that once a long-enough error trace is found, it is induced from some run in the underlying automaton, that contains a cycle. Then, all similar traces that follow the same sequence of states are also error traces, and the cycle induces infinitely many error traces that cannot be removed together using the exact repair method. We now formalise this intuition. 

We make use of the pumping lemma for regular languages, as stated in the following claim. 

\begin{claim}\label{obs}
Let $L$ be a regular language and let $\aut{A}$ be a finite state automaton for $L$, with $n$ states. Let $z\in L$ such that $|z| > n$. Then, we can write $z = uvw$ such that for every $i\geq 0$, it holds that $uv^iw \in L$. 
\end{claim}

\begin{lemma}\label{lemma:syntactic_long}
Let $t \in \aut{T} ( (M_1 || M_2)\times P)$ be an error trace without constraints. Let $N$ be the number of states in $(M_1 || M_2)$ $\times P$. Then, if $|t| > N$, it induces an error trace $t_2 \in \aut{T} (M_2)$, such that we can write $t = uvw$ for $|v| >0$ and for every $i$ it holds that $t_{2^i}:= uv^i w$ also corresponds to an error trace. 
\end{lemma}
\begin{proof}
Denote $M = (M_1 || M_2)\times P$. 
Let $t\in \aut{T}(M) $ be an error trace such that $|t| >N$. Then, the run of $M$ on $t$ contains a cycle. Since $M$ is the composition of three components, there is a state $( p, q, r ) $ of $M$ that appears more than once on the run of $M$ on $t$. Let us denote $$t = t[1] t[2] \cdots t[j] \cdots t[k] \cdots t[m]$$ such that the run of $M$ on $t$ visits the state $( p, q, r ) $ when reading $t[j]$ and $t[k]$. Consider the partition
$t = u\cdot v\cdot w$ where
$$u = t[1]\cdots t[j], v=t[j+1]\cdots t[k], w=t[k+1]\cdots t[m]$$
Then, using the argument of the pumping lemma, it holds that $t_i:=uv^i w$ is a trace of $M$, that reaches the same state as $t$, for every $i\geq 0$. 
Now, note that for every trace $t'$ that does not contain constraints, and for every system $M'$, it holds that $t'$ is a trace of $M'$ iff $t'$ is a \emph{feasible trace} of $M'$. In particular,  
since $t$ does not contain constraints, then $t_i$ is a feasible trace of $M$ for every $i$, and therefore is an error trace for every $i$. In particular, it holds that $t_{2^i} = t_i \downarrow_{\alpha M_2}$ is an error trace in $M_2$ and should be eliminated, for every $i$.
\qed
\end{proof}

Lemma~\ref{lemma:syntactic_long} proves that once a long-enough error trace $t$ is detected, in case that $t$ does not contain constraints, then it induces infinitely many error traces. Thus, the exact repair process will never terminate. 
Note that in this scenario, the approximate repair can fix the system, as it diverts all traces of the same nature to a non-accepting state. 

We remark that the same reasoning cannot be applied to traces with constraints. 
\begin{example}
Consider the programs $M_1$ and $M_2$ and the property $P$ of Figure~\ref{fig:long_semantic}, and consider the trace
$$
t =( x:=0, (x:=x+1)^{90}, \mathit{sync}, x < 100 )$$
where $(x:=x+1)^{90}$ means repeating the letter $(x:=x+1)$ for $90$ times. 

The trace $t$ is of length $93$, and it is an error trace in the system $(M_1 || M_2) \times P$, which is of size at most $9$.\footnote{In fact, when composing the systems, the resulting system is of size $5$.} Then, 
even-though
$|t| > 9$, we cannot decompose it as stated in the claim. This, since after applying $x:=x+1$ more than $100$ times, it will no longer be an error trace of the system. 

\end{example}

\begin{figure}[h]
	\centering

%\begin{center}
\begin{tikzpicture}[scale=0.12]
\tikzstyle{every node}+=[inner sep=0pt]
\draw [black] (26.2,-40.4) circle (3);
\draw (26.2,-40.4) node {$p_0$};
\draw [black] (26.2,-40.4) circle (2.4);
%\draw [black] (13.9,-40.4) circle (3);
\draw (13.9,-40.4) node {$M_1$};
\draw [black] (11.6,-24.8) circle (3);
\draw (11.6,-24.8) node {$q_0$};
\draw [black] (28.4,-24.8) circle (3);
\draw (28.4,-24.8) node {$q_1$};
\draw [black] (44.5,-24.8) circle (3);
\draw (44.5,-24.8) node {$q_2$};
\draw [black] (44.5,-24.8) circle (2.4);
%\draw [black] (0.3,-24.8) circle (3);
\draw (0.3,-24.8) node {$M_2$};
\draw [black] (11.6,-14) circle (3);
\draw (11.6,-14) node {$r_0$};
\draw [black] (11.6,-14) circle (2.4);
\draw [black] (28.4,-14) circle (3);
\draw (28.4,-14) node {$r_1$};
\draw [black] (28.4,-14) circle (2.4);
\draw [black] (44.5,-14) circle (3);
\draw (44.5,-14) node {$r_2$};
%\draw [black] (0.3,-14) circle (3);
\draw (0.3,-14) node {$P$};
\draw [black] (18.9,-40.4) -- (23.2,-40.4);
\fill [black] (23.2,-40.4) -- (22.4,-39.9) -- (22.4,-40.9);
\draw [black] (27.523,-43.08) arc (54:-234:2.25);
\draw (26.2,-47.65) node [below] {$\mathit{sync}$};
\fill [black] (24.88,-43.08) -- (24,-43.43) -- (24.81,-44.02);
\draw [black] (5.7,-24.8) -- (8.6,-24.8);
\fill [black] (8.6,-24.8) -- (7.8,-24.3) -- (7.8,-25.3);
\draw [black] (14.6,-24.8) -- (25.4,-24.8);
\fill [black] (25.4,-24.8) -- (24.6,-24.3) -- (24.6,-25.3);
\draw (20,-25.3) node [below] {$x:=0$};
\draw [black] (29.723,-27.48) arc (54:-234:2.25);
\draw (28.4,-32.05) node [below] {$x:=x+1$};
\fill [black] (27.08,-27.48) -- (26.2,-27.83) -- (27.01,-28.42);
\draw [black] (31.4,-24.8) -- (41.5,-24.8);
\fill [black] (41.5,-24.8) -- (40.7,-24.3) -- (40.7,-25.3);
\draw (36.45,-25.3) node [below] {$\mathit{sync}$};
\draw [black] (44.5,-21.8) -- (44.5,-21.8);
%\fill [black] (44.5,-21.8) -- (45,-21) -- (44,-21);
\draw [black] (45.823,-27.48) arc (54:-234:2.25);
\draw (44.5,-32.05) node [below] {$\mathit{sync}$};
\fill [black] (43.18,-27.48) -- (42.3,-27.83) -- (43.11,-28.42);
\draw [black] (6,-14) -- (8.6,-14);
\fill [black] (8.6,-14) -- (7.8,-13.5) -- (7.8,-14.5);
%\draw [black] (8.852,-12.826) arc (274.60129:-13.39871:2.25);
%\fill [black] (10.86,-11.1) -- (11.29,-10.27) -- (10.3,-10.35);
%\draw [black] (8.852,-12.826) arc (274.60129:-13.39871:2.25);
%\fill [black] (10.86,-11.1) -- (11.29,-10.27) -- (10.3,-10.35);
\draw [black] (14.6,-14) -- (25.4,-14);
\fill [black] (25.4,-14) -- (24.6,-13.5) -- (24.6,-14.5);
\draw (20,-14.5) node [below] {$\mathit{sync}$};
\draw [black] (31.4,-14) -- (41.5,-14);
\fill [black] (41.5,-14) -- (40.7,-13.5) -- (40.7,-14.5);
\draw (36.45,-14.5) node [below] {$x<100$};
\draw [black] (27.077,-11.32) arc (234:-54:2.25);
\draw (28.4,-6.75) node [above] {$x\geq 100$};
\fill [black] (29.72,-11.32) -- (30.6,-10.97) -- (29.79,-10.38);
%\draw [black] (11.6,-17) -- (11.6,-21.8);
%\fill [black] (11.6,-21.8) -- (12.1,-21) -- (11.1,-21);
%\draw [black] (28.4,-17) -- (28.4,-21.8);
%\fill [black] (28.4,-21.8) -- (28.9,-21) -- (27.9,-21);
%\draw [black] (44.5,-17) -- (44.5,-21.8);
%\fill [black] (44.5,-21.8) -- (45,-21) -- (44,-21);
%\draw [black] (0.3,-17) -- (0.3,-21.8);
%\fill [black] (0.3,-21.8) -- (0.8,-21) -- (-0.2,-21);
\end{tikzpicture}
%\end{center}
\caption{}
\label{fig:long_semantic}
\end{figure}

Following the discussion in Section~\ref{sec:weak_reg}, note that if all the three components -- $M_1, M_2$, and $P$ -- do not contain constraints, then only syntactic queries are needed, and only syntactic repair is applied.

\subsection{Correctness and Termination} \label{Sec:termination}
For this discussion, we assume a sound and complete teacher who can answer the membership and equivalence queries in $\agl$, which require verifying communicating programs and properties with first-order constraints. Our implementation uses Z3~\cite{de2008z3} in order to answer satisfiability queries issued in the learning process. The soundness and completeness of Z3 depend on the underlying theory (induced by the program statements we allow). For first-order linear arithmetic over the reals, as we consider in this work, this is indeed the case. However, our method can be applied to all theories for which there exists a sound solver.\footnote{In case we use an incomplete solver, then termination of \lstar~iterations is not guaranteed.}  

As we have discussed earlier, AGR is not guaranteed to terminate, due to its repair part. There are indeed cases for which the \repair~stage may be called infinitely many times. 
%Since communicating programs use first-order constraints, the problem is undecidable, and one cannot hope for a sound and complete algorithm. 
However, in case that no repair is needed, or if a repaired system is obtained after finitely many calls to \repair~, then AGR is guaranteed to terminate with a correct answer. 

To see why, consider a repaired system $M_2^i$ for which $M_1 || M_2^i \vDash P$. Since the goal of $\agl$ is to learn $\mathcal{T}(M_2^i)$, which is (syntactically) regular, this stage will terminate at the latest when $\agl$ learns exactly $\mathcal{T}(M_2^i)$ (it may terminate sooner if a smaller appropriate assumption is found). 
Notice that, in particular, if $M_1 || M_2 \vDash P$, then AGR terminates with a correct answer in the first iteration of the verify-repair loop. 

\repair~is only invoked when a (real) error trace $t_2$ is found in $M_2^i$, in which case a new system $M_2^{i+1}$, that does not include $t_2$, is produced by \repair . If $M_1 || M_2^i\nvDash P$, then an error trace is guaranteed to be found by $\agl$ either in the membership or equivalence phase. 
Therefore, also in case that $M_1 || M_2^i$ violates $P$, the iteration is guaranteed to terminate. %To conclude, we have the following.

In particular, since every iteration of AGR finds and removes an error trace $t_2$, and no new erroneous traces are introduced in the updated system, then in case that $M_2$ has finitely many error traces, AGR is guaranteed to terminate with a repaired system, which is correct with respect to $P$.\footnote{Note that finitely many error traces might induce infinitely many erroneous executions, that are all eliminated together when we eliminate $t_2$.}

To conclude the above discussion, Theorem~\ref{theorem} formally states the correctness and termination of the AGR algorithm. Recall that in Algorithm~\ref{alg:AGR} we set $M_2^0 := M_2$ and that $M_2^i$ is the repaired component after $i$ iterations of repair.

\begin{theorem}\label{theorem} \quad
\begin{enumerate}
    \item \label{item:thm_sat}
    If $M_1 || M_2 \vDash P$ then AGR terminates with the correct answer. That is, the output of AGR is an assumption $A_0$ such that $M_1 || A_0 \vDash P$ and $M_2 \subseteq A_0$.
        \item \label{item:thm_sat_iteration}
    If, after $i$ iterations, a repaired program $M_2^i$ is such that $M_1 || M_2^i\vDash P$, then AGR terminates with the correct answer. That is, AGR outputs an assumption $A^i$ for the AG rule (this is a generalization of item~\ref{item:thm_sat}).
    \item \label{item:thm_error_trace}
    If an iteration $i$ of AGR ends with an error trace $t$, then $M_1 || M_2^i \nvDash P$. %where $M_2^i$ is the updated system at iteration $i$.
    \item \label{item:thm_not_sat}
    If $M_1 || M_2^i \nvDash P$ then AGR finds an error trace. In addition, $M^{i+1}_2$, the system post \repair , contains fewer error traces than $M_2^i$.

\end{enumerate}
\end{theorem}

The proof of Theorem~\ref{theorem} follows from Lemmas~\ref{theorem:T_M2}, \ref{lemma:termination_cex}, and~\ref{lemma:Lstar_terminates_correctly}, given below. 

\begin{lemma}\label{theorem:T_M2}
%If no counterexample is returned, the language $L^*$ learns during the process above is $\T(M_2)$. 
Every iteration $i$ of the 
AGR algorithm terminates. 
In addition, unless \repair~ is invoked, answers to \mq s and \eq s are
consistent with $\T(M_2^i)$.\footnote{If \repair~ is invoked then, as we remove the trace from $M_2^i$, the answer will not be consistent with $M_2^i$, but it will be consistent with $M_2^{i+1}$.}
That is, whenever 
the $\agl$ teacher (Alg.~\ref{alg:queries}) returns \ansyes~ for a \mq~on $t_2$ or +$t_2$ as a counterexample for an \eq , then indeed $t_2 \in \T(M_2^i)$; and whenever the $\agl$ teacher returns \ansno~ for a \mq~or -$t_2$ as a counterexample for an \eq , then indeed $t_2 \notin \T(M_2^i)$.
\end{lemma}

\begin{proof}%[Proof of Lemma~\ref{theorem:T_M2}]
Consistency of \mq s with $M_2^i$ is straight forward as the \mqor~ (Alg.~\ref{alg:queries}) answers \mq s according to membership in $\mathcal{T}(M_2^i)$, or invokes \repair . The same holds for \eq s -- in case \repair~is not invoked, \eqor~ returns the counterexample +$t_2$ (Alg.~\ref{alg:queries} line~\ref{L*_cex_inA}) if it is in $M_2^i$ but not part of the assumption, and it returns -$t_2$ (Alg.~\ref{alg:queries} line~\ref{endequiv}) if it is in the assumption but not in $M_2^i$.

If \repair~was not invoked, we have that 
since both types of queries are consistent with $\mathcal{T}(M_2^i)$, which is a regular language, the current iteration terminates due to the correctness of \lstar. If \repair~was invoked, then obviously the iteration terminates, by calling \repair . Also note that \repair~is only called when a real error is found, that is, a trace $t_2 \in \mathcal{T}(M_2^i)$ such that $M_1 || t_2 \nvDash P $. 
\qed
\end{proof}

\begin{lemma}\label{lemma:termination_cex}
If $M_1 || M_2^i \vDash P$, the AGR algorithm terminates with an assumption $A^i$ for the AG rule. 
If $M_1 || M_2^i \nvDash P$, 
AGR finds a counterexample witnessing the violation (and continues to repair $M_2^i$).
\end{lemma}

\begin{proof}
Assume that $M_1 || M_2^i \vDash P$. By Lemma~\ref{theorem:T_M2}, the answers to \mq s and \eq s are consistent with $\T(M_2^i)$, and from the correctness of \lstar\ algorithm we conclude that the algorithm will eventually learn $\T(M_2^i)$.
Note that in case that $M_1 || M_2 \vDash P$ and that AGR learned $\T(M_2^i)$, that is $\T(A^i)=\T(M_2^i)$, then the 
the \eqor\ returns \ansyes~ (Alg.~\ref{alg:queries} line~\ref{equivcheck2})
and the algorithm terminates with the assumption $A_i$ as a proof of correctness (Alg.~\ref{alg:AGR} line~\ref{terminate_proof}).

Assume that $M_1 || M_2^i \nvDash P$. 
Then there exists an error trace $t\in (M_1||M_2^i)\times P$. From Lemmas~\ref{lemma:feasible_traces},~\ref{lemma:feasible_traces12} it holds that $t_2 = t\downarrow_{\alpha  M_2^i}$ is feasible in $M_2$. In particular, it holds that $t$ is an error trace of $(M_1 || t_2) \times P$. Thus, $M_1||t_2\nvDash P$. 
Since AGR converges towards $\T(M_2^i)$ (by Lemma~\ref{theorem:T_M2}),  
either $t_2$ shows up as an \mq , and the \mqor\ indicates that repair is needed (Alg~\ref{alg:queries} lines~\ref{line:ifin}-\ref{line:repairmq}); Or AGR comes up with a candidate assumption and issues an \eq\ on it. There, again, $t_2$ (or some other error trace) will come up as an error trace $t_2\in \T(M_2^i)$, 
and the \eqor\ will indicate that repair is needed (Alg~\ref{alg:queries} lines~\ref{line:else}-\ref{error3}). 
\qed
\end{proof}

Note that although each phase converges towards $\T(M_2^i)$, it may terminate earlier. We show that in case that the algorithm terminates before finding $\T(M_2^i)$, it returns the correct answer.

\begin{lemma}
\label{lemma:Lstar_terminates_correctly}
%\lstar\ terminates and returns the correct answer. That is: 
\quad 
\begin{enumerate}
    \item If AGR outputs an assumption $A$, then $M_1 || A \vDash P$ and there exists $i$ such that $\mathcal{T}(M_2^i) \subseteq \mathcal{T}(A)$, thus we can conclude $M_1 || M_2^i \vDash P$.
    \item If a phase $i$ of AGR ends with finding an error trace $t$, then $M_1 || M_2^i \nvDash P$. 
    %a trace $t_2\in \mathcal{T}(M_2^i)$ then we can use $t_2$ to generate an error trace, and thus $M_1 || M_2 \nvDash P$.
\end{enumerate}
\end{lemma}

\begin{proof}
\emph{Item 1.} Assume AGR returns an assumption $A$. 
%We can then conclude that the following holds for $A$: 
Then there exists $i$ such that
$\T(M_2^i) \subseteq \T(A)$ and $M_1 || A \vDash P$, since this is the only scenario in which an assumption $A$ is returned (Alg.~\ref{alg:queries} lines~\ref{equivcheck1}-\ref{equivcheck2}
and Alg~\ref{alg:AGR} lines~\ref{line:ifeq}-\ref{terminate_proof}). 
From the soundness of the AG rule for communicating systems (Theorem~\ref{thrm:soundness}) it holds that $M_1 || M_2^i \vDash P$. 

\emph{Item 2.}
Assume now that a phase $i$ of AGR ends with finding an error trace $t$ (and a call to \repair ). We prove that $M_1 || M_2^i \nvDash P$. 
First note that AGR may output such a trace both while making a \mq\ and while making an \eq . If $t$ was found during a \mq\  (Alg.~\ref{alg:queries} lines~\ref{terminate_cex}-\ref{line:repairmq}), then there exists $t_2\in T(M_2^i)$ such that $M_1 || t_2 \nvDash P$, and $t\in (M_1 || t_2) \times P$. Since $t_2\in \T(M_2^i)$, it holds that $t$ is also an error trace of $(M_1 || M_2^i) \times P$, proving $M_1 || M_2^i \nvDash P$. 

%Thus, there is $t_1 \in M_1$ and a run $r$ on $t_1 || t_2$ such that $r\nvDash P$. Since $t_2\in \T(M_2^i)$, is holds that $r$ is a run of $M_1 || M_2^i$, proving $M_1 || M_2^i \nvDash P$. 

If $t$ was found during an \eq\ (Alg~\ref{alg:queries} lines~\ref{line:else}-\ref{error3}), then $t$ is an error trace in $(M_1 || A_j^i) \times P$. Moreover, $t\downarrow_{\alpha A_j^i} \in \T(M_2^i)$. This makes $t$ an error trace of $(M_1 || M_2^i)\times P$ as well, thus  $M_1 || M_2^i \nvDash P$ . This concludes the proof.
\qed
\end{proof}

The proof of Theorem~\ref{theorem} follows almost directly from the lemmas above. 
\begin{proof}[Proof of Theorem~\ref{theorem}]
Lemma~\ref{lemma:termination_cex} states that if $M_1 || M_2^i \vDash P$ then AGR terminates with the correct answer. This implies item~\ref{item:thm_sat} and item~\ref{item:thm_sat_iteration}.

In addition, Lemma~\ref{lemma:termination_cex} states that if $M_1 || M_2^i \nvDash P$ then AGR finds an error trace witnessing the violation. Once such an error trace is found, \repair~is invoked. Since \repair~eliminates at least one error trace, the system post \repair~contains fewer error traces, and item~\ref{item:thm_not_sat} follows.

Lemma~\ref{lemma:Lstar_terminates_correctly} states that if an iteration $i$ of AGR ends with an error trace, then $M_1 || M_2^i \nvDash P$. This implies item~\ref{item:thm_error_trace}.
\qed
\end{proof}

\section{Experimental Results}\label{results}

\begin{table*}[t]
\centering
\caption{AGR algorithm results on various examples}
\begin{tabular}{@  {}     |c|c|c|c|c|c|c|c|c|@{}}  
%\begin{tabular}{@{}lllllllll@{}}
\toprule
\textbf{Example}               & \textbf{$M_1$ Size} & \textbf{$M_2$ Size} & \textbf{$P$ Size} & \textbf{Time (sec.)} & \textbf{$A$ size} & \textbf{Repair Size} & \textbf{Repair Method} & \textbf{\#Iterations} \\ \hline
% \#1    & 4                & 4                & 3               & 0.2                  & 3               & \multicolumn{3}{c|}{verification}                                    \\
% \hline
% \#2   & 16               & 16               & 3               & 1.8                  & 4               & \multicolumn{3}{c|}{verification}                           \\
% \hline
% \#3   & 32               & 32               & 3               & 11.1                 & 6               & \multicolumn{3}{c|}{verification}                    \\
\hline
\#4   & 64               & 64               & 3               & 95                   & 7               & \multicolumn{3}{c|}{verification}                          \\
% \hline
% \multirow{3}{*}{\#5}                 & \multirow{3}{*}{2}                & \multirow{3}{*}{3}                 & \multirow{3}{*}{2}              & 0.08                 & 3               & 3                    & aggress.            & 2                          \\

 %                    &  &    &        & 0.09                 & 4               & 4                    & approx.          & 2                          \\
%                      &  & &       & 0.108                & 6               & 9                    & exact                & 2                          \\
\hline

\multirow{3}{*}{\#6}                 & \multirow{3}{*}{2}                & \multirow{3}{*}{27}                 & \multirow{3}{*}{2}        & 0.106                & 5               & 27                   & aggress.            & 2                          \\
                 &                 &                &                & 0.126                & 6               & 28                   & approx.          & 2                          \\
                 &                 &                &                & 0.132                & 8               & 81                   & exact                & 2                          \\
                 \hline
\multirow{3}{*}{\#7}                 & \multirow{3}{*}{2}                & \multirow{3}{*}{81}                 & \multirow{3}{*}{2}      & 0.13                 & 6               & 81                   & aggress.            & 2                          \\
                 &                 &                &                & 0.138                & 7               & 82                   & approx.          & 2                          \\

                 &                 &                &                & 0.165                & 9               & 243                  & exact                & 2                          \\
                 \hline
\multirow{3}{*}{\#8}                 & \multirow{3}{*}{2}                & \multirow{3}{*}{243}                 & \multirow{3}{*}{2}                & 0.15                 & 8               & 243                  & aggress.            & 2                          \\

                 &                 &               &                & 0.17                 & 8               & 244                  & approx.          & 2                          \\
                 &                 &               &                & 0.223                & 10              & 729                  & exact                & 2                          \\
%                \hline
% \#9                        & 2                & 4                & 3               & 0.093                & 3               & \multicolumn{3}{c|}{verification}                                  \\
% \hline
%\#10                        & 3                & 16               & 4               & 0.29                 & 13              & \multicolumn{3}{c|}{verification}              \\
\hline
\#11                        & 5                & 256              & 6               & 4.88                 & 92              & \multicolumn{3}{c|}{verification}              \\
%\hline
%\#12                    & 2                & 4                & 3               & 0.08                 & 3               & \multicolumn{3}{c|}{verification}      \\
%\hline
%\#13                     & 3                & 16               & 4               & 0.22                 & 10              & \multicolumn{3}{c|}{verification}                              \\
\hline
\#14                     & 5                & 256              & 6               & 4.44                 & 109             & \multicolumn{3}{c|}{verification}                      \\
\hline
\multirow{3}{*}{\#15}                 & \multirow{3}{*}{3}                & \multirow{3}{*}{16}                 & \multirow{3}{*}{5}        & 0.69                 & 12              & 16                   & aggress.            & 5                                       \\

            &                 &                &                & 0.28                 & 13              & 18                    & approx.          & 3                             \\
          &                &                &                & 4.27                 & 44              & 864                  & exact                & 5                          \\
          \hline
\multirow{3}{*}{\#16}&             \multirow{3}{*}{4}    &     \multirow{3}{*}{256}                   &        \multirow{3}{*}{8}        & 6.63                 & 113             & 256                  & aggress.            & 2                          \\

                 &                 &         &                & 5.94                 & 113             & 257                  & approx.          & 2                          \\

              &                 &               &                & 12.87                & 155             & 1280                 & exact                & 2                          \\
%              \hline
%\#17        & 2                & 3                & 4               & 0.075                & 3               & \multicolumn{3}{c|}{verification}              \\
%\hline
%\multirow{3}{*}{\#18}&\multirow{3}{*}{2}&\multirow{3}{*}{3}& \multirow{3}{*}{4}& 0.34& 5& 4& aggress.& 2\\
% &&&& 0.37& 5& 4& approx.& 2\\
%&&&& 0.488 &5&4& exact & 2   \\
 \hline
\multirow{3}{*}{\#19}                 & \multirow{3}{*}{3}                & \multirow{3}{*}{16}                 & \multirow{3}{*}{5}            & 1.07                 & 18              & 18                   & aggress.            & 3                          \\
 &                 &                &                & 1.12                 & 18              & 18                   & approx.          & 3                          \\
 &                 &                &                & 1.26                 & 18              & 18                   & exact                & 3                          \\
% \hline
%\#20       & 9                & 6                & 15              & 0.1                  & 6               & \multicolumn{3}{c|}{verification}                          \\
%\hline
%\#21     & 11               & 13               & 17              & 0.18                 & 11              & \multicolumn{3}{c|}{verification}                             \\
\hline
\multirow{3}{*}{\#22}& \multirow{3}{*}{2}& \multirow{3}{*}{4}& \multirow{3}{*}{2}& 0.09 & 1 & 4 (trivial)& aggress.  &    4\\
     &&&& 0.21 & 6 & 8& approx.  &    5\\

&&&& \multicolumn{3}{c|}{timeout}& exact  &timeout\\

%\hline
%\multirow{3}{*}{\#23}&\multirow{3}{*}{11}&\multirow{3}{*}{12}& \multirow{3}{*}{17}& 0.24& 1& 12 (trivial)& aggress.& 2\\
 %&&&& 0.25& 1& 13 (trivial) & approx.& 2\\
%&&&& 0.26 &1&144 (trivial)& exact & 2   \\
% \hline     
%\multirow{3}{*}{\#24}&\multirow{3}{*}{4}&\multirow{3}{*}{8}& \multirow{3}{*}{5}& 0.35& 6& 9 & aggress.& 2\\
% &&&& 0.34& 6& 9 & approx.& 2\\
%&&&& 0.36 &6&9& exact & 2   \\
% \hline     

%\#25     & 3               & 5               & 5             & 0.13                 & 4              & \multicolumn{3}{c|}{verification}                             \\
% \hline     

%\#26     & 4               & 5               & 5             & 0.11                 & 4              & \multicolumn{3}{c|}{verification}                             \\
%      \hline     

%\#27     & 4               & 8               & 5             & 0.35                 & 5              & \multicolumn{3}{c|}{verification}                        \\     

\bottomrule

\end{tabular}

\label{table:results}
\end{table*}

We implemented our AGR framework in Java, integrating the \lstar\ learner implementation from the LTSA tool
\cite{DBLP:books/daglib/0097696}. We used Z3~\cite{de2008z3} to implement 
calls to the teacher while answering 
the satisfaction queries in Algorithm~\ref{alg:queries}, and for abduction in \repair . 

Table~\ref{table:results} displays results of running AGR on various examples, varying in their sizes, types of errors -- semantic and syntactic, and the number of errors. Additional results are available in~\cite{hadarphd}. 
The examples are available on~\cite{examples}. 
The {\em iterations} column indicates the number of iterations of the verify-repair loop, until a repaired $M_2$ is achieved. 
Examples with no errors were verified in the first iteration, and are indicated by {\em verification}.
We experimented with the three repair methods described in Section~\ref{sec:methods}. 
Figure~\ref{fig:graphs} presents comparisons between the three methods in terms of run-time and the size of the repair and assumptions. The graphs are given in logarithmic scale. 

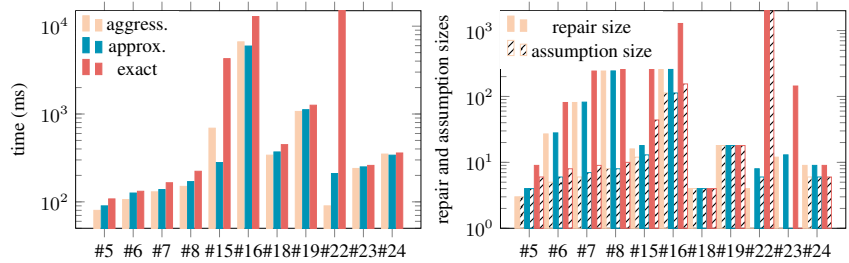
\begin{figure}

%\resizebox{\textwidth}{!}{

    \centering
    \begin{tabular}{cc}

\begin{tikzpicture}[scale=0.8]
\begin{axis}[
    ybar=0.5pt,
    ymin=50,
    ymax = 15000,
    bar width=3pt,
    legend style={font=\small, at={(0,1)},anchor=north west, draw=none},
    ylabel={time (ms)},
     ymode=log,
        log basis y={10},
    symbolic x coords={\#5,\#6,\#7,\#8,\#15, \#16, \#18, \#19, \#22, \#23, \#24},
    xtick=data,
   %nodes near coords,
    %every node near coord/.style={font=\tiny},
    %nodes near coords style={font=\tiny},
    %nodes near coords align={vertical},
    x tick label style={font=\small, %rotate=30,anchor=east
    },
    point meta=rawy,
    width = 0.6\textwidth,
    ]
\addplot[apricot, fill]  coordinates {(\#5,80) (\#6,106) (\#7,130) (\#8, 150) (\#15, 690) (\#16,6630) (\#18, 340) (\#19, 1070)(\#22,90)(\#23,240)(\#24,350)};
\addplot[bondiblue,fill] coordinates {(\#5,90) (\#6,126) (\#7,138) (\#8, 170) (\#15, 280) (\#16, 5940) (\#18, 370) (\#19, 1120) (\#22,210)(\#23,250)(\#24,340)};
\addplot[lightcarminepink, fill] coordinates {(\#5,108) (\#6,132) (\#7,165) (\#8, 223) (\#15, 4270) (\#16,12870) (\#18, 448) (\#19, 1260)(\#22,10000000)(\#23,260)(\#24,360)};
\legend{aggress.  ,approx. , exact}
\end{axis}

\end{tikzpicture}

         &
        % \hspace{0.2cm}
    
\begin{tikzpicture}[scale=0.8]

\begin{axis}[
    ybar=0.2pt,
    ymin=1,
    ymax = 2000,
    bar width=2.1pt,
    legend style={font=\small, at={(0,1)},anchor=north west, draw=none},
    ylabel={repair and assumption sizes},
     ymode=log,
        log basis y={10},
    symbolic x coords={\#5,\#6,\#7,\#8,\#15, \#16, \#18, \#19, \#22, \#23, \#24},
    xtick=data,
   %nodes near coords,
    %every node near coord/.style={font=\tiny},
    %nodes near coords style={font=\tiny},
    %nodes near coords align={vertical},
    x tick label style={font=\small, %rotate=30,anchor=east
    },
    point meta=rawy,
    width = 0.6\textwidth
    ]
\addplot[apricot,fill] coordinates {(\#5,3) (\#6,27) (\#7,81) (\#8, 243) (\#15, 16) (\#16,256) (\#18, 4) (\#19, 18)(\#22,4)(\#23,12)(\#24,9)}; %aggressive repair aize
\addplot[apricot, pattern=north east lines] coordinates {(\#5,3) (\#6,5) (\#7,6) (\#8, 8) (\#15, 12) (\#16, 113) (\#18, 4) (\#19, 18) (\#22,1)(\#23,1)(\#24,6)}; %agressive assumption size
\addplot+[bondiblue,fill] coordinates {(\#5,4) (\#6,28) (\#7,82) (\#8, 244) (\#15, 18) (\#16, 257) (\#18, 4) (\#19, 18) (\#22,8)(\#23,13)(\#24,9)}; %approximate repair size
\addplot+[bondiblue,pattern=north east lines] coordinates {(\#5,4) (\#6,6) (\#7,7) (\#8, 8) (\#15, 13) (\#16,113) (\#18, 4) (\#19, 18)(\#22,6)(\#23,1)(\#24,6)}; % approximate assumption size
\addplot+[lightcarminepink, fill] coordinates {(\#5,9) (\#6,81) (\#7,243) (\#8, 729) (\#15, 864) (\#16,1280) (\#18, 4) (\#19, 18)(\#22,1000000)(\#23,144)(\#24,9)}; %exact repair size
\addplot+[lightcarminepink, pattern=north east lines] coordinates {(\#5,6) (\#6,8) (\#7,9) (\#8, 10) (\#15, 44) (\#16, 155) (\#18, 4) (\#19, 18) (\#22,100000)(\#23,1)(\#24,6)}; %exact assumption size
\legend{repair size, assumption size}
\end{axis}
%\label{fig:graphs}
\end{tikzpicture}

\end{tabular}
    \caption{Comparing repair methods: time and repair size (logarithmic scale).}
    \label{fig:graphs}
    
%    }%end resize
\end{figure}
Most of our examples model multi-client-server communication protocols, with varying sizes. Our tool managed to repair all those examples
that were flawed.

As can be seen in Table~\ref{table:results}, our tool successfully generates assumptions that are significantly smaller than the repaired and the original $M_2$. 

For the examples that needed repair, in most cases our tool needed 2-5 iterations of verify-repair in order to successfully construct a repaired component. 
Interestingly, in example $\#15$ the {\em aggressive} method converged slower than the {\em approximate} method. This is due to the structure of $M_2$, in which different error traces lead to different states. Marking these states as non-accepting removed each trace separately. However, some of these traces have a common transition, and preventing this transition from reaching an accepting state, as done in the {\em approximate} method, managed removing several error traces in a single repair. 

Example $\#22$ models a simple structure in which, due to a loop in $M_2$, the same alphabet sequence can generate infinitely many error traces. The \emph{exact} repair method timed out, since it attempted removing one error trace at a time. On the other hand, the \emph{aggressive} method removed all accepting states, creating an empty program -- a trivial (yet valid) repair. In contrast, the {\em approximate} method created a valid, non-trivial repair. 

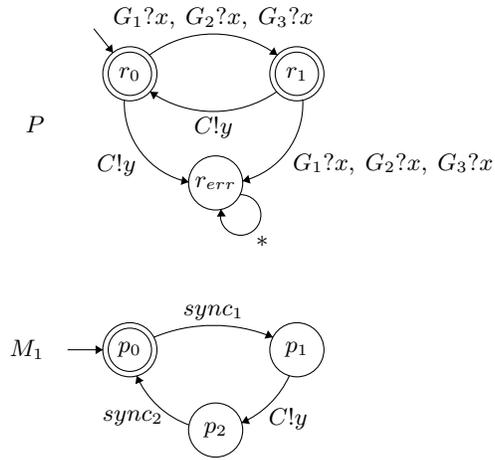
\begin{figure}
    \centering

\begin{center}
\begin{tikzpicture}[scale=0.12]
\tikzstyle{every node}+=[inner sep=0pt]
\draw [black] (21.6,-13.6) circle (3);
\draw (21.6,-13.6) node {$r_0$};
\draw [black] (21.6,-13.6) circle (2.4);
\draw [black] (40.1,-13.6) circle (3);
\draw (40.1,-13.6) node {$r_1$};
\draw [black] (40.1,-13.6) circle (2.4);
\draw [black] (31.1,-25.7) circle (3);
\draw (31.1,-25.7) node {$r_{err}$};
%\draw [black] (11.1,-19.2) circle (3);
\draw (11.1,-19.2) node {$P$};
\draw [black] (21.6,-44.2) circle (3);
\draw (21.6,-44.2) node {$p_0$};
\draw [black] (21.6,-44.2) circle (2.4);
\draw [black] (40.1,-44.2) circle (3);
\draw (40.1,-44.2) node {$p_1$};
%\draw [black] (10.2,-44.2) circle (3);
\draw (10.2,-44.2) node {$M_1$};
\draw [black] (31.1,-53.1) circle (3);
\draw (31.1,-53.1) node {$p_2$};
\draw [black] (17.6,-8.6) -- (19.73,-11.26);
\fill [black] (19.73,-11.26) -- (19.62,-10.32) -- (18.84,-10.95);
\draw [black] (23.761,-11.53) arc (126.54534:53.45466:11.905);
\fill [black] (37.94,-11.53) -- (37.59,-10.65) -- (37,-11.46);
\draw (30.85,-8.69) node [above] {$G_1?x,\mbox{ }G_2?x,\mbox{ }G_3?x$};
\draw [black] (37.864,-15.589) arc (-55.3204:-124.6796:12.327);
\fill [black] (23.84,-15.59) -- (24.21,-16.46) -- (24.78,-15.63);
\draw (30.85,-18.28) node [below] {$C!y$};
\draw [black] (40.752,-16.513) arc (2.83651:-76.12038:8.794);
\fill [black] (34.08,-25.49) -- (34.97,-25.78) -- (34.73,-24.81);
\draw (39.6,-23.59) node [right] {$G_1?x,\mbox{ }G_2?x,\mbox{ }G_3?x$};
\draw [black] (28.116,-25.601) arc (-101.54069:-182.18669:8.918);
\fill [black] (28.12,-25.6) -- (27.43,-24.95) -- (27.23,-25.93);
\draw (22.32,-23.79) node [left] {$C!y$};
\draw [black] (33.789,-27.003) arc (91.87498:-196.12502:2.25);
\draw (36.25,-31.52) node [below] {$*$};
\fill [black] (31.7,-28.63) -- (31.23,-29.44) -- (32.23,-29.41);
\draw [black] (14.7,-44.2) -- (18.6,-44.2);
\fill [black] (18.6,-44.2) -- (17.8,-43.7) -- (17.8,-44.7);
\draw [black] (24.256,-42.813) arc (112.57915:67.42085:17.175);
\fill [black] (37.44,-42.81) -- (36.9,-42.04) -- (36.51,-42.97);
\draw (30.85,-41) node [above] {$sync_1$};
\draw [black] (39.269,-47.072) arc (-24.39086:-66.24931:10.41);
\fill [black] (33.98,-52.3) -- (34.91,-52.44) -- (34.51,-51.52);
\draw (39.13,-50.65) node [below] {$C!y$};
\draw [black] (28.165,-52.536) arc (-109.51627:-156.74837:9.939);
\fill [black] (22.35,-47.09) -- (22.21,-48.02) -- (23.13,-47.63);
\draw (21.92,-50.91) node [below] {$sync_2$};
\end{tikzpicture}
\end{center}

    \caption{Program $M_1$ and specification $P$ of Example~\ref{ex:results}}
    \label{fig:resMP}
\end{figure}

\begin{figure}
    \centering
    
\begin{center}
\begin{tikzpicture}[scale=0.12]
\tikzstyle{every node}+=[inner sep=0pt]
\draw [black] (21.2,-25.2) circle (3);
\draw (21.2,-25.2) node {$q_0$};
\draw [black] (21.2,-25.2) circle (2.4);
\draw [black] (31.8,-8.6) circle (3);
\draw (31.8,-8.6) node {$q_1$};
%\draw [black] (7,-25.2) circle (3);
\draw (7,-25.2) node {$M_2$};
\draw [black] (40.3,-4.5) circle (3);
\draw (40.3,-4.5) node {$q_2$};
\draw [black] (31.8,-25.2) circle (3);
\draw (31.8,-25.2) node {$q_4$};
\draw [black] (40.3,-19.8) circle (3);
\draw (40.3,-19.8) node {$q_5$};
\draw [black] (31.8,-41.4) circle (3);
\draw (31.8,-41.4) node {$q_7$};
\draw [black] (40.3,-36) circle (3);
\draw (40.3,-36) node {$q_8$};
\draw [black] (48.1,-8.6) circle (3);
\draw (48.1,-8.6) node {$q_3$};
\draw [black] (48.1,-8.6) circle (2.4);
\draw [black] (48.1,-25.2) circle (3);
\draw (48.1,-25.2) node {$q_6$};
\draw [black] (48.1,-25.2) circle (2.4);
\draw [black] (48.1,-41.4) circle (3);
\draw (48.1,-41.4) node {$q_9$};
\draw [black] (48.1,-41.4) circle (2.4);
\draw [black] (14.3,-25.2) -- (18.2,-25.2);
\fill [black] (18.2,-25.2) -- (17.4,-24.7) -- (17.4,-25.7);
\draw [black] (22.81,-22.67) -- (30.19,-11.13);
\fill [black] (30.19,-11.13) -- (29.33,-11.53) -- (30.18,-12.07);
\draw (25.88,-15.59) node [left] {$G_1?x$};
\draw [black] (33.122,-5.942) arc (139.16905:92.3319:5.971);
\fill [black] (37.4,-3.88) -- (36.58,-3.41) -- (36.62,-4.41);
\draw (32.32,-3.95) node [above] {$sync_1$};
\draw [black] (24.2,-25.2) -- (28.8,-25.2);
\fill [black] (28.8,-25.2) -- (28,-24.7) -- (28,-25.7);
\draw [black] (33.327,-22.636) arc (138.932:105.92311:8.349);
\fill [black] (37.33,-20.09) -- (36.42,-19.83) -- (36.7,-20.79);
\draw (32.39,-20.57) node [above] {$sync_1$};
\draw [black] (22.84,-27.71) -- (30.16,-38.89);
\fill [black] (30.16,-38.89) -- (30.14,-37.95) -- (29.3,-38.49);
\draw [black] (33.162,-38.749) arc (141.43766:103.41745:7.564);
\fill [black] (37.32,-36.11) -- (36.43,-35.81) -- (36.66,-36.78);
\draw (32.27,-36.58) node [above] {$sync_1$};
\draw [black] (43.166,-3.77) arc (87.63139:36.9121:5.157);
\fill [black] (47.08,-5.82) -- (47,-4.88) -- (46.2,-5.49);
\draw (48.09,-3.85) node [above] {$sync_2$};
\draw [black] (45.468,-10.027) arc (-67.49407:-112.50593:14.414);
\fill [black] (34.43,-10.03) -- (34.98,-10.8) -- (35.36,-9.87);
\draw (39.95,-11.63) node [below] {$G_1?x$};
\draw [black] (43.234,-20.347) arc (69.21071:41.39898:8.406);
\fill [black] (46.56,-22.65) -- (46.4,-21.72) -- (45.65,-22.38);
\draw (47.79,-20.8) node [above] {$sync_2$};
\draw [black] (45.276,-26.205) arc (-74.66685:-105.33315:20.143);
\fill [black] (34.62,-26.21) -- (35.26,-26.9) -- (35.53,-25.93);
\draw (39.95,-27.42) node [below] {$G_2?x$};
\draw [black] (43.243,-36.487) arc (69.97582:40.63388:8.082);
\fill [black] (46.61,-38.82) -- (46.47,-37.88) -- (45.71,-38.54);
\draw (47.83,-36.94) node [above] {$sync_2$};
\draw [black] (45.336,-42.557) arc (-72.17163:-107.82837:17.592);
\fill [black] (34.56,-42.56) -- (35.17,-43.28) -- (35.48,-42.33);
\draw (39.95,-43.9) node [below] {$G_3?x$};
\end{tikzpicture}
\end{center}
    \caption{Program $M_2$ of Example~\ref{ex:results}}
    \label{fig:resm2}
\end{figure}
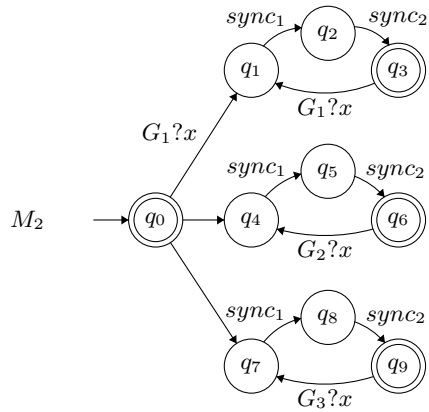
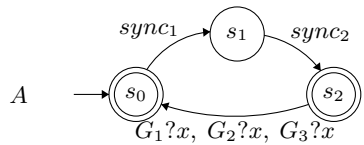
\begin{figure}
    \centering
    
\begin{center}
\begin{tikzpicture}[scale=0.12]
\tikzstyle{every node}+=[inner sep=0pt]
\draw [black] (21.6,-44.2) circle (3);
\draw (21.6,-44.2) node {$s_0$};
\draw [black] (21.6,-44.2) circle (2.4);
\draw [black] (32.8,-37.5) circle (3);
\draw (32.8,-37.5) node {$s_1$};
%\draw [black] (8.6,-44.2) circle (3);
\draw (8.6,-44.2) node {$A$};
\draw [black] (43.5,-44.2) circle (3);
\draw (43.5,-44.2) node {$s_2$};
\draw [black] (43.5,-44.2) circle (2.4);
\draw [black] (14.7,-44.2) -- (18.6,-44.2);
\fill [black] (18.6,-44.2) -- (17.8,-43.7) -- (17.8,-44.7);
\draw [black] (22.787,-41.459) arc (147.26821:94.50872:9.227);
\fill [black] (29.82,-37.25) -- (28.99,-36.81) -- (29.07,-37.81);
\draw (23.06,-38.03) node [above] {$sync_1$};
\draw [black] (35.786,-37.647) arc (78.88854:37.00448:10.362);
\fill [black] (42.06,-41.58) -- (41.98,-40.64) -- (41.18,-41.24);
\draw (42.04,-38.53) node [above] {$sync_2$};
\draw [black] (40.714,-45.308) arc (-71.62703:-108.37297:25.901);
\fill [black] (24.39,-45.31) -- (24.99,-46.03) -- (25.3,-45.09);
\draw (32.55,-47.13) node [below] {$G_1?x,\mbox{ }G_2?x,\mbox{ }G_3?x$};
\end{tikzpicture}
\end{center}

    \caption{Assumption $A$ of Example~\ref{ex:results}}
    \label{fig:resass}
\end{figure}
\begin{example}\label{ex:results}
As long as the system needs repair, no assumption can be learned. When we reach a correct $M_2$, we are usually able to find a smaller assumption that proves the correctness of $M_1 || M_2$ with respect to $P$. 
Our tool preforms the best on examples in the spirit of $M_1$ and $P$ of Figure~\ref{fig:resMP} and $M_2$ of Figure~\ref{fig:resm2}. There, the specification $P$ allows all traces in which first $M_2$ acts on one of its channels ($G_1, G_2$ or $G_3$), and then $M_1$ acts on its channel ($C$). The program $M_2$ in Figure~\ref{fig:resm2} is more restrictive than $P$ requires -- once the variable $x$ is read through some channel, $M_2$ continues to use only this channel. $M_1 || M_2 \vDash P$ due to the restriction on their synchronization using $sync_1$ and $sync_2$. We are then able to learn
the assumption $A$ of Figure~\ref{fig:resass}, which is much smaller than $M_2$, and allows proving the correctness of 
 $M_1 || M_2$ with respect to $P$. 
 
 This example demonstrates a similar behavior to that of examples $\#4, \#6, \#7$ and $\#8$ of Table~\ref{table:results}. 
\end{example}

\begin{example}\label{ex:res2}
Consider now $M_1$ and $P$ of Figure~\ref{fig:two_systems_spec} and the repaired $M_2^1$ of Figure~\ref{fig:repairabd}. In this case, we learn an assumption with $5$ states, that is the result of merging states $q_1$ and $q_2$ in $M_2^1$. This is since we answer queries according to $M_2^1$, which
has a more unique structure than the structure of $M_2$ of Figure~\ref{fig:resm2}. This demonstrates the behaviour of example $\#19$ in Table~\ref{table:results}.
\end{example}

\label{sec:experiments}
\section{Conclusion and 
Future Work}\label{sec:future}
We have presented the model of communicating programs, which is able to capture program behavior and synchronization between the system components, while exploiting a finite automata representation in order to apply automata learning. 
We then presented
AGR, which offers a new take on the learning-based approach to assume-guarantee verification, by managing to cope with complex properties and by also repairing infinite-state programs.

Our experimental results show that %using existing semantic tools, 
AGR can produce very succinct proofs, and can repair flawed communicating programs efficiently. 
AGR leverages the finite automata-like representation of the systems in order to apply the \lstar\ algorithm and to learn small proofs of correctness. 

We prove that in general, the weakest assumption that is often used for compositional verification, is not regular for the case of communicating programs, and we come up with a new goal for the learning process. In addition, we find types of communicating programs for which the weakest assumption is regular. 
We leave finding the full characterization of programs for which the weakest assumption is regular for future work. %As a future work, it will be interesting to consider the number of variables in the different components, or the type of constraints, and their effect on the regularity of the weakest assumption, and on the learning process. 

In this work, we repair the system by eliminating error traces, and locate new constraints learned using abduction, at the end of the error trace, in order to make it infeasible. A possible extension of this process is to wisely locate constraints over the error trace. Intuitively, we would like the constraint to be as "close'' as possible to the error location. However, this is not a trivial task, as the error can be the result of multiple actions of the two communicating programs. 
Another extension to the repair process is \emph{changing} the program behavior, rather than blocking it. Examples for such a mutation-based approach to program repair are~\cite{DBLP:conf/fm/RothenbergG16,DBLP:conf/hvc/BloemDFFHKRRS12,DBLP:conf/vmcai/NguyenTC19}.

For syntactic repair, we characterize cases in which the repair process does not converge. This may happen also in the case of semantic repair, in which infinitely many new constraints are learned and the repair process does not terminate. As future work, we intend to incorporate invariant generation, according to reoccurring error traces, in order to help the convergence of the semantic repair process. 

\label{sec:conclusion}

\begin{comment}

\section{additions from the conference version}
we should write those in the paragraph at the end of the intro. 
\begin{itemize}
    \item full proofs
    \item section~\ref{sec:teacher} discusses the implementation of the teacher
    \item formal discussion regarding soundness, completeness and termination
    \item Sec.~\ref{AGR:AR_rule_reg} - formal definition of weakest assumption, lemma about regular and its proof
    \item traces in the composed system - Sec. \ref{sec:tarces}
    \item Sections~\ref{sec:weak_reg} and ~\ref{sec:syntactic_convg}
    \item future work \ref{sec:future} 
\end{itemize}

\end{comment}


\bibliographystyle{plainurl}      % basic style, author-year citations
\bibliography{ref}   % name your BibTeX data base

%\begin{thebibliography}{}
% etc
%\end{thebibliography}

\newpage
\appendix
\section{Appendix}
\subsection{Full Proofs}\label{app:traces}

\paragraph{Lemma~\ref{lemma:traces12} (restated)}
Let  $M_1, M_2$ be two programs, and let $t$ be a trace of $M_1 || M_2$. Then $t\downarrow_{\alpha M_1}$ is a trace of $M_1$ and $t\downarrow_{\alpha M_2}$ is a trace of $M_2$.

\begin{proof}
Let ${M_i} = \langle Q_i, X_i, \alpha M_i, \delta_i, q_0^i, F_i\rangle $
for $i=1,2$, and denote $M_1 || M_2 = \langle (Q_1\times Q_2)\cup(Q'_1\times Q'_2) , X_1\cup X_2, \alpha M , \delta, (q_0^1, q_0^2), F_1 \times F_2\rangle$, as defined in Section~\ref{sec:parallel}.

Let $r = \langle q_0, c_1, q_1\rangle \ldots \langle q_{m-1}, c_m, q_m \rangle$ be the run in $M_1 || M_2 $ such that $t$ is induced from $r$. 
Denote by $t_1 = t\downarrow_ {\alpha M_1}$  the trace $t_1 = (c_{i_1}, \cdots , c_{i_n})$. The proof for $t_2 =  t\downarrow_ {\alpha M_2}$ is the same.  

We first observe the following. If $(a_1, \cdots, a_k)$ is a trace of $M_1|| M_2$ such that $\forall i: a_i\notin \alpha M_1$,
and $a_1\in\alpha M_2$, 
and $q = (q_1, q_2^0)$ is the state in $M_1 || M_2$ before reading $a_1$,  then $\forall i\geq 1: \exists q_2^i: \delta((q_1, q_2^{i-1}), a_i) = (q_1, q_2^{i})$, that is, when reading a trace that does not contain letters from $\alpha M_1$, the program $M_1||M_2$ only advances on the $M_2$ component. This is true since by the definition of $\delta$, if $a_i$ is not in $\alpha M_1$, then $\delta((q_1, q_2), a_i) = (q_1, \delta_2(q_2, a_i))$. %Informally, it means that $\delta$ only advances in the $P$ component and not in the $M$ component.  
The requirement of $a_1\in\alpha M_2$ is since if $a_1\notin \alpha m_2$ and $a_1\notin \alpha m_1$
then $a_1$ is an equality constraint of the form $x=y$ after a communication over a common channel and the transition relation is defined differently.

We now inductively prove that $\forall 1\leq j \leq n$ it holds that $(c_{i_1}, \cdots , c_{i_j})$ is a trace of $M_1$. In particular, for $j=n$ this means that $t_1 \in M_1$.
For some parts of the proof we abuse notations where $c$ is a communication action over $M_1$ and $M_2$, and we use it also to denote the restriction to the first component of $M_1$.

Let $j:=1$ and denote $k:= i_1$. Then $c_1, \ldots,c_{k-1} \notin \alpha M_1 $ since $k$ is the first index of $t$ for which $c_k \in \alpha M_1$
or $c_k = (c_k^1, c_k^2)$ such that $c_k$ is a communication action and $c_k^1\in \alpha M_1$.
In particular, this means that no common communication action had occurred until $c_k$.
Thus, $\forall 1\leq i<k: \exists q^2_i: \delta((q_0^1, q_{i-1}^2), c_i) =(q_0^1, q_i^2)$.
For $c_{i_1} = c_k \in\alpha M_1$, 
one of the following holds:
\begin{enumerate}
    \item  If $c_k\in \alpha M_1$ is not a common communication action, then 
by the definition of $\delta$, we have $\delta((q_0^1, q_{k-1}^2), c_{i_1}) = (\delta_1(q_0^1, c_{i_1}), q_{k-1}^2)$.  Then indeed, $\langle q_0^1, c_{i_1}, \delta_1(q_0^1, c_{i_1}) \rangle$ is a run in $M_1$, making $(c_{i_1})$ a trace of $M_1$. 

\item If $c_k = (c_k^1, c_k^2)$ is a common communication channel, then it holds that $c_{k+1}$ is an equality constraint.
Then, by the definition of $\delta$ we have that 
$\delta((q_0^1, q_{k-1}^2), c_{i_1}) = ({q_0^1}', {q_{k-1}^2}')$ and $\delta(({q_0^1}', {q_{k-1}^2}'), c_{i_1+1}) = (\delta_1(q_0^1, c_{i_1}^1), \delta_2(q_{k-1}^2, c_{i_1}^2))$. Then, again we have that 
 $\langle q_0^1, c_{i_1}, \delta_1(q_0^1, c_{i_1}) \rangle$ is a run in $M_1$, making $(c_{i_1})$ a trace of $M_1$. 
\end{enumerate}

Let $1< j \leq n$, and assume $t_{j-1} = (c_{i_1}, \cdots, c_{i_{j-1}})$ is a trace of $M_1$. Let $\langle q_0, c_{i_1}, q_1 \rangle \ldots $ $\langle q_{j-2}, c_{i_{j-1}}, q_{j-1} \rangle$ be a run that induces $t_{j-1}$. Denote $i_{j-1} = k, i_j = k+m $ for some $m>0$.
Then, as before, $c_{k+1}, \ldots , c_{k+m-1} \notin \alpha M_1$ and is not a communication action as well, thus
$\forall k<l<k+m:\exists q^2_l:    
\delta(q_{j-1}, q^2_{l-1}), c_l) = (q_{j-1}, q^2_l)$. For $c_{i_j}$ it holds that  
either $c_{i_j}\in \alpha M_1$ is not a communication action and then $$\delta(q_{j-1}, q^2_{k+m-1}), c_{i_j}) = (\delta_1(q_{j-1}, c_{i_j}), q^2_{k+m-1}))$$ or that $c_{i_j}=(c_{i_j}^1, c_{i_j}^2)$ is a communication action and then  
$$\delta((q_{j-1}, q^2_{k+m-1}), c_{i_j}) = ({q_{j-1}}', {q^2_{k+m-1}}')$$ and $$\delta(({q_{j-1}}', {q^2_{k+m-1}}'), c_{i_1+1}) = (\delta_1({q_{j-1}}, c_{i_j}^1), \delta_2({q^2_{k+m-1}}, c_{i_j}^2))$$ In both cases, in $M_1$ it holds that 
$\delta_1(q_{j-1}, c_{i_j})$ 
is defined in $M_1$
 and thus $(c_{i_1},\cdots, c_{i_j})$ is a trace of $M_1$, as needed. 
\qed
\end{proof}

\paragraph{Lemma~\ref{lemma:feasible_traces12} (restated)}
Let  $M_1, M_2$ be  two programs, and let $t$ be a \textbf{feasible} trace of $M_1|| M_2$. Then $t\downarrow_{\alpha M_i}$ is a \textbf{feasible} trace of $M_i$ for $i\in{1,2}$.

\begin{proof}[Proof of Lemma~\ref{lemma:feasible_traces12}]
Denote by $X_i$ the set of variables of $M_i$ for $i\in\{1,2\}$.
Let $t\in M_1|| M_2$ be a feasible trace. Then, there exists an execution $u$ on $t$. Denote $t = (b_1, \cdots, b_n)$ and 
\\ $u = (\beta_0, b_1, \beta_1, \cdots, b_n, \beta_n)$. We build an execution $e$ on $t\downarrow _{\alpha M_1}$ as follows (in the same way, we can build an execution on $t\downarrow _{\alpha M_2}$).
Let $t\downarrow _{\alpha M_1} = (c_1, \cdots, c_k)$.
We define $e=(\gamma_0, c_1, \gamma_1, \cdots , c_k, \gamma_k)$ as follows. 
\begin{enumerate}
    \item Set $j:=0, i:=0$.
    \item Define $\gamma_0:= \beta_0(X_1)$ and set $j:=j+1$. 
    \item Repeat until $j = k$ :
    \begin{itemize}
        \item[--] Let $i'> i$ be the minimal index such that
    $b_{i'} = c_j$  or 
    $b_{i'} = (g*x, g*y)$ for  $c_j = g*x\in \alpha M_1$. 
        \item[--] Define $\gamma_j : = \beta_{i'}(X_1)$ and set $j:= j+1, i:= i'+1$. 
    \end{itemize}
  \end{enumerate}
  
    Note that for each $i<l<i'$ is holds that $b_l\notin \alpha M_1$ and there are no $x\in X_1$ and $y\in X_2$ such that $b_l = (g*x, g*y)$ for $g*x \in \alpha M_1$. 
    Otherwise, $b_l$ is either a synchronization between $M_1$ and $M_2$, or $b_l$ is a letter of $M_1$ that 
    belongs to $\alpha M_1$ and is part of $t\downarrow_{\alpha M_1}$. 
    %must synchronize with $t\downarrow _{\alpha M_1}$. 
    Both are contradiction to the fact that $i'$ is the minimal index for which $b_{i'}$ and $c_j$ 
    are either equal or that $c_j$ is a read/write action and $b_{i'}$ is a synchronization on that action. 

    %Thus, by the definition of an execution it holds that $\beta_i(X_1) = \beta_l(X_1)$
    
    Since $u$ is an execution, and for all $i<l<i'$ it holds that $b_l$  does not contain variables of $X_1$, 
    we have that $\beta_i(X_1) = \beta_l(X_1)$
    for all $i<l<i'$. %, by the definition of an execution. 
    %by the definition of an execution, it holds that $\beta_i(X_1) = \beta_l(X_1)$, 
    This is
    since an assignment to a variable may only change if the variable is involved in the action alphabet. In particular, it holds that $\beta_{i'-1}(X_1) = \beta_i(X_1) = \gamma_{j-1}(X_1)$. We now can assign $\gamma_j$ to be the same as $\beta_{i'}(X_1)$ and result in a valid assignment, as needed.
\qed
\end{proof}

\paragraph{Lemma~\ref{lemma:weakest_reg} (restated)}
Let $P$ be a property and $M_1$ and $M_2$ be communicating programs such that $\alpha P$, $\alpha M_1$ and $\alpha M_2$ 
do not contain constraints. 
%Let $M_2$ be some communicating program. 
Then, the weakest assumption $A_w$ with respect to $M_1, M_2$ and $P$, is regular.

\begin{proof}
Denote
$M_1 = \langle Q_1, X_1, \alpha M_1, \delta_1, {q^1_0}, F_1 \rangle$ and \\
$P = \langle Q_P, X_P,  \alpha P, \delta_P, {q^P_0}, F_P \rangle$. 
We construct a communicating system $A$ over $\alpha M_2$, and prove that $\aut{T}(A) = \aut{L}(A_w)$.

We construct $A$ in stages. 
First, we compute $M_1 \times P$ with a transition function $\delta_{M_1,P}$. We define $A' = \langle Q', X_2, \alpha M_2, $ $\delta', q_0', F'
\rangle$ 
where  
\begin{enumerate}
    \item $Q' = Q_1 \times Q_P$
    \item $q_0' = (q^1_0, q^P_0)$
    \item $F' = ((Q_1\setminus F_1) \times Q_P) \cup (Q_1 \times F_P$) . Intuitively, $A$ contains all traces that are not error traces in $M_1 \times P$. That is, either $M_1 \times P$ ends in a state that is not accepting for $M_1$, or it is accepted by $P$, and so does not reach an error state. 
    \item $\delta'$ is defined as follows. 
    \begin{enumerate}
        \item For every $\sigma \in \alpha M_2 \setminus (\alpha M_1 \cup \alpha P)$
        and for every
        $(q_1, q_p) \in Q_1 \times Q_P$
        we define $\delta' ((q_1,q_p), \sigma) = (q_1,q_p) $. That is, we add self-loops with the alphabet of $M_2$ that does not synchronize with $M_1$ and~$P$. 
\item For every $\sigma \in (\alpha M_1 \cup \alpha P)\setminus \alpha M_2$ and every $(q_1,q_p)\in Q_1 \times Q_P$ such that $\delta_{M_1,P} ((q_1,q_p), \sigma) = (q_1',q_p')$, we define $\delta' ((q_1,q_p), \epsilon) = (q_1',q_p')$.\label{item:delta}
\item For every $\sigma \in \alpha M_2 \cap (\alpha M_2 \cup \alpha M_2)$ and every $(q_1,q_p)\in Q_1 \times Q_P$ we define
$\delta' ((q_1,q_p), \sigma) = \delta_{M_1,P}((q_1,q_p), \sigma)$. That is, for letters in the intersection of the alphabets of $M_2$ and $M_1\times P$, the transitions do not change. 
    \end{enumerate} 
 
\end{enumerate}
 
We now remove $\epsilon$ transitions from $A'$ and determinize it to a program $A$. %Here we need both semantic and syntactic determinization as the alphabet of $A'$ contains constraints.
Finally, we set the set of accepting states of $A$ to $F = \{ S :~ S\subseteq F' \}$. That is, $t$ is an accepting trace of $A$ 
iff all the states that $A$ reaches when reading $t$ are either not accepting in $M_1$, or accepting (hence, not rejecting) in~$P$. 

\textbf{Claim 1}
$\aut{L}(A) = \aut{L}(A_w)$.

Note that in case that there are no constraints, all traces are feasible, that is, $\aut{L}(A) = \aut{T}(A)$. To show the correctness of {Claim 1}, we prove the following claim. 

\textbf{Claim 2}
For a trace $t\in(\alpha M_2)^*$, the program $A$ reaches the same set of states $S$ when reading $t$, as the set of states that $M_1 \times P$ reaches given $t$.

%The run of $A$ on a trace $t\in(\alpha M_2)^*$ reaches the set of state $S$ that is exactly the set of states $M_1 \times P$ reach when given $t$. 

{Claim $1$} then follows from the definition of the set $F$ of accepting states. 

We now prove Claim $2$ by induction on the trace $t$. 

For the base case, let $t = \epsilon$ be the empty trace. Then, $A$ stays in its initial state $S_0$ when reading $t$. The composition $M_1 \times P$ can only advance on letters $\sigma \in (\alpha M_1 \cup \alpha P)\setminus \alpha M_2$ due to the definition of the parallel and conjunctive compositions. Since in the construction of $A'$ we define an $\epsilon$-transition for all such letters $\sigma$, it holds that all states that can be reached in $M_1 \times P$ 
by only reading letters in $(\alpha M_1 \cup \alpha P)\setminus \alpha M_2$ are, after the elimination of $\epsilon$ transitions, in the set of initial states. After the determinization, they form exactly the initial state $S_0$ in $A$. Therefore, the claim holds for the case that $t=~\epsilon$.

Now, assume that $t = u\cdot \sigma$, where the claim holds for $u$. 
Let $S_u$ be the state $A$ reaches when reading $u$. Recall that $S_u$ corresponds to a set of states in $M_1 \times P$, and by the induction hypothesis, this is the set of states that $M_1 \times P$ reaches when reading $u$. 
We consider the three following cases.
\begin{itemize}
    \item [(a)] $\sigma \in \alpha M_2 \setminus (\alpha M_1 \cup \alpha P)$. According to the definition of composition, when computing $(M_1 || t)\times P$, it holds that  $M_1 \times P$ does not advance on $\sigma$ since it is not in its alphabet. Thus, the set of states $M_1 \times P$ reaches when reading $u\cdot \sigma$ is  $S_u$. In the definition of $A'$, 
    for letters in $\alpha M_2 \setminus (\alpha M_1 \cup \alpha P)$ we only define  self-loops for the states of $M_1\times P$. Thus, the set $S_u$ does not change when reading $\sigma$, and this is also the set $A$ reaches after reading $u\cdot \sigma$. Therefore, for the case of $\sigma\in \alpha M_2 \setminus (\alpha M_1 \cup \alpha P)$, both $A$ and $M_1\times P$ reach the same set of states when reading $t=u\cdot \sigma$ and the claim holds. 
    \item [(b)] $\sigma \in (\alpha M_1 \cup \alpha P)\setminus \alpha M_2$. This case cannot occur as we only consider traces $t\in (\alpha M_2)^*$. Note that the definition of $\delta'$ (Item~\ref{item:delta}) is used for the base case of the induction. 
    \item [(c)] $\sigma \in \alpha M_2 \cap (\alpha M_2 \cup \alpha M_2)$. In this case $M_1 \times P$ synchronizes with $t$ while reading $\sigma$. Thus, the set of states $M_1 \times P$ reaches when reading $u\cdot \sigma$ is $S' = \{ (q_1', q_p') : \exists (q_1, q_p) \in S_u~ ~s.t. ~\delta_{M_1,P}((q_1, q_p), \sigma) = (q_1', q_p')  \}$. Since for letters in the intersection $\alpha M_2 \cap (\alpha M_2 \cup \alpha M_2)$ we have defined the transition relation of $A'$ to be $\delta_{M_1, P}$, this is exactly the set of states $A'$ (and thus $A$) reaches when reading $u\cdot \sigma$.
\end{itemize}

To conclude, $M_1\times P$ and $A$ reach the same set of states when reading every trace $t\in(\alpha M_2)^*$. That is, all runs in $(M_1 || t) \times P $ are accepting iff $A$ reaches an accepting state when reading $t$. This concludes the proof. 
\qed
\end{proof}

\end{document}